\pdfoutput=1
\documentclass[nonacm,acmsmall,dvipsnames,x11names,svgnames]{acmart} %[anonymous,review
\acmJournal{PACMPL}
\acmVolume{1}
\acmNumber{POPL} % CONF = POPL or ICFP or OOPSLA
\acmArticle{1}
\acmYear{2025}
\acmMonth{1}
\acmDOI{} % \acmDOI{10.1145/nnnnnnn.nnnnnnn}
\startPage{1}

\setcopyright{none}
\usepackage{booktabs}   %% For formal tables:
\usepackage{subcaption} %% For complex figures with subfigures/subcaptions
\usepackage{adjustbox}
\usepackage{amsmath,amsfonts}
\usepackage{csquotes}
\usepackage{marvosym}
\usepackage{graphicx}
\usepackage{nameref}
\usepackage{stmaryrd}
\usepackage{xcolor}
\usepackage{xfrac}
\usepackage{xspace}
\usepackage{cancel}
\usepackage{tikz-cd}
\usepackage{thmtools}
\usepackage{thm-restate}
\usepackage{thm-autoref}
\usepackage{wrapfig}

\usepackage{adjustbox}

\definecolor{mygreen}{HTML}{3C8031}
\definecolor{myorange}{HTML}{F58137}
\definecolor{rot}{RGB}{200,34,84}

\usepackage{listings,multicol,multirow}
\usepackage{tensor}
\usepackage{tikz}
\usetikzlibrary{patterns,decorations.pathreplacing,arrows,arrows.meta,decorations.pathmorphing,positioning,fit,trees,shapes,shadows,automata,calc,calligraphy,tikzmark}

\usepackage{pgfplots}

\tikzset{
	katexpr/.style={
		line width=0.2pt,
		rounded corners=0.1cm,
		draw,
		fill=rot,
		opacity=.2,
		color=rot,
		inner xsep=-2.5pt,
		inner ysep=-1pt
	},
	collapsewp/.style={
		line width=1.5pt,
		dashed,
		draw,
		color=DodgerBlue3,
		inner xsep=-2.5pt,
		inner ysep=-1pt
	},
	collapsesp/.style={
		line width=1.5pt,
		rounded corners=0.2cm,
		draw,
		color=DodgerBlue3,
		inner xsep=-2.5pt,
		inner ysep=-1pt
	},
	contrapos/.style={
		draw=mygreen,
		thick,
		dotted,
		line width=.75pt
	},
	weakcontrapos/.style={
		draw=mygreen!40,
		thick,
		dotted,
		line width=1pt
	},
	implication/.style={
		-implies,double equal sign distance, thick,shorten <=1pt,shorten >=1pt
	},
	weakimplication/.style={
		-implies,double equal sign distance, thick,shorten <=1pt,shorten >=1pt,lightgray
	},
	implicationshort/.style={ % shorter implications for large picture
		-implies,double equal sign distance, thick,shorten <=4pt,shorten >=4pt
	},
	implicationup/.style={
		-implies,double equal sign distance, thick, shorten >=6pt, shorten <=1pt
	},
	implicationdown/.style={
		-implies,double equal sign distance, thick, shorten <=6pt, shorten >=1pt
	},
	implicationupkat/.style={
		-implies,double equal sign distance, thick, shorten >=7pt, shorten <=1pt
	},
	implicationdownkat/.style={
		-implies,double equal sign distance, thick, shorten <=7pt, shorten >=1pt
	},
	implicationupmid/.style={
		-implies,double equal sign distance, thick, shorten >=11pt, shorten <=1pt
	},
	implicationdownmid/.style={
		-implies,double equal sign distance, thick, shorten <=11pt, shorten >=1pt
	},
	equivalence/.style={
		implies-implies,double equal sign distance, thick
	},
	equivalencemid/.style={
		implies-implies,double equal sign distance, thick, shorten >=1pt, shorten <=11pt
	},
	equivalencekat/.style={
		implies-implies,double equal sign distance, thick, shorten <=2pt, shorten >=2pt
	},
	galois/.style={
		draw=myorange,
		implies-implies,
		double equal sign distance,
		thick,
		shorten <=1pt,shorten >=1pt
	},
	axis/.style={
		line width=2pt,
		semitransparent,
	},
	weakgalois/.style={
		draw=myorange!40,
		implies-implies,
		double equal sign distance,
		thick,
		shorten <=1pt,shorten >=1pt
	},
	galoisshort/.style={
		draw=myorange,
		implies-implies,
		double equal sign distance,
		thick,
		shorten <=6pt
	},
	implic/.style={
		draw=myorange,
		thick
	},	
	mystate/.style = {
		circle,
		inner sep=3pt,
		draw,
		font=\small
	},
	pre/.style={
		fill=mygreen!40%DodgerBlue3!40!white,
	},
	post/.style={
		fill=DodgerBlue3!40,
	},
	selected/.style={
		line width=2.0pt,
		color=rot
	},
	nonselected/.style={
		line width=1pt,
		color=black %!60
	},
	selectnode/.style={
		line width=2.0pt,
		draw=rot, %only lines
	},
	program/.style={
		->,thick
	}
}

\usepackage{caption}
\usepackage{proof}
\usepackage{mathtools}

\usepackage{enumitem}                
            
\usepackage{mdframed}
\usepackage{amsthm}

\usepackage{wasysym}
\usepackage{natbib}

\newcommand{\qedtriangle}{\hfill\raisebox{-.15ex}{\rotatebox{90}{$\triangle$}}}

\newcommand{\hoare}[3]{\left\langle \,{#1}\vphantom{#3}\, \right\rangle \mathrel{#2} \left\langle \, {#3}\vphantom{#1} \, \right\rangle}

\newcommand{\sfsymbol}[1]{\textsf{\upshape {#1}}}

\newcommand{\ttsymbol}[1]{\texttt{\upshape {#1}}}
\newcommand{\wpsymbol}{\sfsymbol{wp}}

\renewcommand{\wp}[2]{\wpsymbol\,\llbracket#1\rrbracket(#2)}

\newcommand{\wpC}[1]{\wpsymbol\llbracket#1\rrbracket}

\newcommand{\boldawp}[2]{\boldawpsymbol\,\boldsymbol{\llbracket#1\rrbracket(#2)}}
\newcommand{\bolddwp}[2]{\bolddwpsymbol\,\boldsymbol{\llbracket#1\rrbracket(#2)}}
\newcommand{\boldawlp}[2]{\boldawlpsymbol\,\boldsymbol{\llbracket#1\rrbracket(#2)}}
\newcommand{\spsymbol}{\sfsymbol{sp}}

\renewcommand{\sp}[2]{\spsymbol\,\llbracket#1\rrbracket(#2)}

\newcommand{\spC}[1]{\spsymbol\llbracket#1\rrbracket}

\newcommand{\aspsymbol}{\sfsymbol{asp}}
\newcommand{\boldaspsymbol}{\textbf{\sfsymbol{asp}}}
\newcommand{\asp}[2]{\aspsymbol\,\llbracket#1\rrbracket(#2)}
\newcommand{\boldasp}[2]{\boldaspsymbol\,\boldsymbol{\llbracket#1\rrbracket(#2)}}

\newcommand{\dspsymbol}{\sfsymbol{dsp}}

\newcommand{\dsp}[2]{\dspsymbol\,\llbracket#1\rrbracket(#2)}

\newcommand{\slpsymbol}{\sfsymbol{slp}}

\newcommand{\slp}[2]{\slpsymbol\llbracket#1\rrbracket(#2)}

\newcommand{\aslpsymbol}{\sfsymbol{aslp}}

\newcommand{\aslp}[2]{\aslpsymbol\llbracket#1\rrbracket(#2)}

\newcommand{\bolddslp}[2]{\bolddslpsymbol\,\boldsymbol{\llbracket#1\rrbracket(#2)}}
\newcommand{\dslpsymbol}{\sfsymbol{dslp}}
\newcommand{\bolddslpsymbol}{\textbf{\sfsymbol{dslp}}}
\newcommand{\dslp}[2]{\dslpsymbol\llbracket#1\rrbracket(#2)}

\newcommand{\awpsymbol}{\sfsymbol{awp}}
\newcommand{\boldawpsymbol}{\textbf{\sfsymbol{awp}}}
\newcommand{\awp}[2]{\awpsymbol\llbracket#1\rrbracket(#2)}

\newcommand{\dwpsymbol}{\sfsymbol{dwp}}
\newcommand{\bolddwpsymbol}{\textbf{\sfsymbol{dwp}}}
\newcommand{\dwp}[2]{\dwpsymbol\llbracket#1\rrbracket(#2)}

\newcommand{\wlpsymbol}{\sfsymbol{wlp}}

\newcommand{\wlp}[2]{\wlpsymbol\llbracket#1\rrbracket(#2)}

\newcommand{\bolddwlp}[2]{\bolddwlpsymbol\,\boldsymbol{\llbracket#1\rrbracket(#2)}}
\newcommand{\dwlpsymbol}{\sfsymbol{dwlp}}
\newcommand{\bolddwlpsymbol}{\textbf{\sfsymbol{dwlp}}}
\newcommand{\dwlp}[2]{\dwlpsymbol\llbracket#1\rrbracket(#2)}

\newcommand{\awlpsymbol}{\sfsymbol{awlp}}
\newcommand{\boldawlpsymbol}{\textbf{\sfsymbol{awlp}}}
\newcommand{\awlp}[2]{\awlpsymbol\llbracket#1\rrbracket(#2)}

\newcommand{\conditionalPair}[2]{{\let\oldarraystretch\arraystretch}\renewcommand{\arraystretch}{1}~\holter{~\raisebox{.5ex}{${#1}$}~}{~\raisebox{.125ex}{${#2}$}~}~\renewcommand{\arraystretch}{\oldarraystretch}}

\newcommand{\guard}{\ensuremath{\guardb}} % conditial guard
\newcommand{\guardb}{\ensuremath{b}} % conditial guard
\newcommand{\ee}{\ensuremath{e}} % expressions

\newcommand{\SKIP}{\ttsymbol{skip}}

\newcommand{\DIVERGE}{\ensuremath{\textnormal{\texttt{diverge}}}}

\newcommand{\AssignSymbol}{\coloneqq}
\newcommand{\ASSIGN}[2]{\ensuremath{#1 \AssignSymbol #2}}

\newcommand{\AVAILLOC}[1]{\PosNats}
\usepackage{actuarialangle}

\newcommand{\COMPOSE}[2]{\ensuremath{{#1}{\,\fatsemi}~ {#2}}}

\newcommand{\NDCHOICE}[2]{\ensuremath{\left\{\, {#1} \,\right\}\mathrel{\Box}\left\{\, {#2} \,\right\}}}

\newcommand{\IFSYMBOL}{\ensuremath{\textnormal{\texttt{if}}}}

\newcommand{\ELSESYMBOL}{\ensuremath{\textnormal{\texttt{else}}}}

\newcommand{\ITE}[3]{\ensuremath{\IFSYMBOL\,\left(\, {#1} \,\right)\,\left\{\, {#2} \,\right\}\,\ELSESYMBOL\,\left\{\, {#3} \,\right\}}}

\newcommand{\WHILESYMBOL}{\ensuremath{\textnormal{\texttt{while}}}}

\newcommand{\WHILEDO}[2]{\ensuremath{\WHILESYMBOL \left(\, {#1} \,\right)\left\{\, {#2} \,\right\}}}

\newcommand{\ngcl}{\textnormal{\sfsymbol{nGCL}}\xspace}   % Programs
\newcommand{\Vars}{\ensuremath{\mathsf{Vars}}\xspace}   % Variables
\newcommand{\Pred}{\ensuremath{\textsf{Pred}}} % Predicates

\newcommand{\PosNats}{\ensuremath{\mathbb{N}_{>0}}\xspace}
\newcommand{\Ints}{\ensuremath{\mathbb{Z}}\xspace}

\newcommand{\B}{\mathbb{B}}

\newcommand{\iverson}[1]{\left[ {#1} \right]}

\newcommand{\subst}[2]{\left[ {#1} \middle/ {#2}\right]}

\newcommand{\eval}[1]{\ensuremath{\llbracket {#1} \rrbracket}}

\newcommand{\evalInv}[1]{\ensuremath{\eval{#1}^{-1}}}
\newcommand{\sem}[2]{\ensuremath{\llbracket {#1} \rrbracket}(#2)}

\newcommand{\seminv}[2]{\ensuremath{\llbracket {#1} \rrbracket}^{{-}1}(#2)}

\newcommand{\States}{\Sigma}

\newcommand{\true}{\mathsf{true}}
\newcommand{\false}{\mathsf{false}}
\newcommand{\mydot}{\text{{\Large\textbf{.}}~}\xspace}
\newcommand{\mylambda}[1]{\lambda\, #1\mydot}

\newcommand*\oline[1]{%
  \vbox{%
    \hrule height 0.3pt%                  % Line above with certain width
    \kern0.25ex%                          % Distance between line and content
    \hbox{%
      \kern0.04em%                        % Distance between content and left side of box, negative values for lines shorter than content
      \ifmmode#1\else\ensuremath{#1}\fi%  % The content, typeset in dependence of mode
      \kern0.04em%                        % Distance between content and left side of box, negative values for lines shorter than content
    }% end of hbox
  }% end of vbox
}

\newcommand*\olinedbl[1]{%
  \vbox{%
    \hrule height 0.6pt%                  % Line above with certain width
    \kern0.25ex%                          % Distance between line and content
    \hbox{%
      \kern0.04em%                        % Distance between content and left side of box, negative values for lines shorter than content
      \ifmmode#1\else\ensuremath{#1}\fi%  % The content, typeset in dependence of mode
      \kern0.04em%                        % Distance between content and left side of box, negative values for lines shorter than content
    }% end of hbox
  }% end of vbox
}

\newcommand{\negate}[1]{\oline{#1}}	
\newcommand{\negatedbl}[1]{\olinedbl{#1}}	

\newcommand{\vspaace}{\vphantom{{\LARGE A}}}

\newcommand{\qiff}{\quad\textnormal{iff}\quad}
\newcommand{\qqiff}{\qquad\textnormal{iff}\qquad}

\newcommand{\cupdot}{\xspace \ensuremath{\dot{\cup}} \xspace}
\newcommand{\lordot}{\xspace \ensuremath{\dot{\lor}} \xspace}

\newcommand{\qqand}{\qquad\textnormal{and}\qquad}

\newcommand{\qqorequiv}{\qqmorespace{\textnormal{or equivalently}}}

\newcommand{\qqimplies}{\qquad\textnormal{implies}\qquad}

\newcommand{\morespace}[1]{~{}#1{}~}

\newcommand{\qqmorespace}[1]{\qquad{}#1{}\qquad}

\newcommand{\eeq}{~{}={}~}
\newcommand{\ccup}{~{}\cup{}~}

\newcommand{\qmid}{\quad{}|{}\quad}

\newcommand{\qeq}{\quad{}={}\quad}

\newcommand{\mmodels}{~{}\models{}~}

\newcommand{\setcomp}[2]{\left\{\, {#1} ~\middle|~ {#2} \,\right\}}
\definecolor{webgreen}{rgb}{0,.5,0}

\newcommand{\lightgray}[1]{\textcolor{lightgray}{#1}}

\newcommand{\green}[1]{\textcolor{webgreen}{#1}}

\newcounter{computationarrowsone}
\newcounter{computationarrowstwo}

\newcounter{sarrow}

\newcommand{\lfp}{\ensuremath{\textnormal{\sfsymbol{lfp}}~}}
\newcommand{\gfp}{\ensuremath{\textnormal{\sfsymbol{gfp}}~}}

\newcommand{\Conf}{\ensuremath{\ensuremath{\mathcal{P}(\States)}}}

\newcommand{\pre}{\ensuremath{b}}
\newcommand{\post}{\ensuremath{c}}
\newcommand{\program}{\ensuremath{p}}
\newcommand{\secprogram}{\ensuremath{q}}

\newcommand{\plogin}{\ensuremath{\program_{\text{login}}}}
\newcommand{\pwdcorrect}{\ensuremath{\textit{pwd correct}}}
\newcommand{\pwdincorrect}{\ensuremath{\textit{pwd incorrect}}}
\newcommand{\pwdabcd}{\ensuremath{\textit{pwd = }1234}}
\newcommand{\access}{\ensuremath{\textit{access granted}}}
\newcommand{\noaccess}{\ensuremath{\textit{access denied}}}
\newcommand{\error}{\ensuremath{\textit{error}}}

\newcommand{\NDCHOICEC}{\ensuremath{\NDCHOICE{\program_1}{\program_2} }}

\newcommand{\KAT}{\textup{\textsf{KAT}}\xspace}

\newcommand{\TopKAT}{\textup{\textsf{TopKAT}}\xspace}

\newcommand{\awpCAPdwlp}[2]{\ensuremath{(\awpsymbol \cap \dwlpsymbol) \llbracket #1 \rrbracket\left( #2 \right)}}
\newcommand{\awpCUPdwlp}[2]{\ensuremath{(\awpsymbol \cup \dwlpsymbol) \llbracket #1 \rrbracket\left( #2 \right)}}
\newcommand{\aspCAPdslp}[2]{\ensuremath{(\aspsymbol \cap \dslpsymbol) \llbracket #1 \rrbracket\left( #2 \right)}}
\newcommand{\aspCUPdslp}[2]{\ensuremath{(\aspsymbol \cup \dslpsymbol) \llbracket #1 \rrbracket\left( #2 \right)}}

\definecolor{reeed}{rgb}{0.59, 0.0, 0.09}
\newcommand{\katexpr}[1]{ \color{reeed}{#1}}

\newcommand{\statelabel}[1]{{\Large \boldsymbol{\mathsf{#1}}}}

\usepackage[nameinlink]{cleveref}

\theoremstyle{acmdefinition}
\newtheorem{thm}{Theorem} % workaround for "no counter defined" error
\newtheorem{remark}[thm]{Remark}
\newtheorem{observation}[thm]{Observation}

\begin{document}

%% Title information
\title{A Taxonomy of Hoare-Like Logics}
%\title{Quantitative Strongest Post}         %% [Short Title] is optional;
                                        %% when present, will be used in
                                        %% header instead of Full Title.
%\titlenote{with title note}             %% \titlenote is optional;
                                        %% can be repeated if necessary;
                                        %% contents suppressed with 'anonymous'
\subtitle{Towards a Holistic View using Predicate Transformers and Kleene Algebras with Top and Tests}
                                        %\subtitle{A Calculus for Reasoning about the Flow of Quantitative Information}                     %% \subtitle is optional
%\subtitlenote{with subtitle note}       %% \subtitlenote is optional;
                                        %% can be repeated if necessary;
                                        %% contents suppressed with 'anonymous'

%% Author information
%% Contents and number of authors suppressed with 'anonymous'.
%% Each author should be introduced by \author, followed by
%% \authornote (optional), \orcid (optional), \affiliation, and
%% \email.
%% An author may have multiple affiliations and/or emails; repeat the
%% appropriate command.
%% Many elements are not rendered, but should be provided for metadata
%% extraction tools.

%% Author with single affiliation.
% \author{Linpeng Zhang}
% %\authornote{Both authors contributed equally to this research.}
% \email{linpeng.zhang.20@ucl.ac.uk}
% \orcid{0000-0002-1485-327X}
% \affiliation{%
%     \institution{University College London}
%     %  \streetaddress{P.O. Box 1212}
%     \city{London}
%     \country{United Kingdom}
%     %  \state{Ohio}
%     %  \postcode{43017-6221}
% }

\author{Lena Verscht}
\email{lverscht@cs.uni-saarland.de}
\orcid{0000-0001-6823-7918}
\affiliation{%
	\institution{Saarland University}
    \city{Saarbrücken}
	\country{Germany}
}
\affiliation{%
	\institution{RWTH Aachen}
    \city{Aachen}
	\country{Germany}
}

\author{Benjamin Lucien Kaminski}
%\authornotemark[1]
\email{kaminski@cs.uni-saarland.de}
\orcid{0000-0001-5185-2324}
\affiliation{%
	\institution{Saarland University}
    \city{Saarbrücken}
	\country{Germany}
}
\affiliation{%
    \institution{University College London}
    \city{London}
    \country{UK}
    %  \streetaddress{P.O. Box 1212}
    %  \state{Ohio}
    %  \postcode{43017-6221}
}

%\author{Tobias Winkler}
%%\authornote{Both authors contributed equally to this research.}
%\email{tobias.winkler@cs.rwth-aachen.de}
%\orcid{0000-0003-1084-6408}
%\affiliation{%
%    \institution{RWTH Aachen University}
%    %  \streetaddress{P.O. Box 1212}
%    \city{Aachen}
%    \country{Germany}
%    %  \state{Ohio}
%    %  \postcode{43017-6221}
%}

%%% Author with two affiliations and emails.
%\author{First2 Last2}
%\authornote{with author2 note}          %% \authornote is optional;
%                                        %% can be repeated if necessary
%\orcid{nnnn-nnnn-nnnn-nnnn}             %% \orcid is optional
%\affiliation{
%  \position{Position2a}
%  \department{Department2a}             %% \department is recommended
%  \institution{Institution2a}           %% \institution is required
%  \streetaddress{Street2a Address2a}
%  \city{City2a}
%  \state{State2a}
%  \postcode{Post-Code2a}
%  \country{Country2a}                   %% \country is recommended
%}
%\email{first2.last2@inst2a.com}         %% \email is recommended
%\affiliation{
%  \position{Position2b}
%  \department{Department2b}             %% \department is recommended
%  \institution{Institution2b}           %% \institution is required
%  \streetaddress{Street3b Address2b}
%  \city{City2b}
%  \state{State2b}
%  \postcode{Post-Code2b}
%  \country{Country2b}                   %% \country is recommended
%}
%\email{first2.last2@inst2b.org}         %% \email is recommended

%% Abstract
%% Note: \begin{abstract}...\end{abstract} environment must come
%% before \maketitle command
\begin{abstract}
	We study Hoare-like logics, including partial and total correctness Hoare logic, incorrectness logic, Lisbon logic, and many others through the lens of predicate transformers à la Dijkstra and through the lens of Kleene algebra with top and tests (\TopKAT).
	Our main goal is to give an overview -- a taxonomy -- of how these program logics relate, in particular under different assumptions like for example program termination, determinism, and reversibility.
	As a byproduct, we obtain a \TopKAT characterization of Lisbon logic, which \mbox{-- to the} best of \mbox{our knowledge --} is a novel result.
\end{abstract}

%% 2012 ACM Computing Classification System (CSS) concepts
%% Generate at 'http://dl.acm.org/ccs/ccs.cfm'.

%% End of generated code

%% Keywords
%% comma separated list
  %% \keywords are mandatory in final camera-ready submission

%% \maketitle
%% Note: \maketitle command must come after title commands, author
%% commands, abstract environment, Computing Classification System
%% environment and commands, and keywords command.

\maketitle

%\pagebreak

%!TEX root = ../main.tex

\section{Introduction}
\label{sec:intro}

% \begin{flushright}
% 	\enquote{Program correctness and incorrectness are\\ two sides of the same coin.}\\[.5em] 
% 	--- \citet{OHearn19}
% \end{flushright}

% 1. program properties: correctness and incorrectness & others
Arguably one of the most prominent and well-studied program logics is \citet{Hoa69} logic for verifying program \emph{correctness}: the \emph{\underline{abilit}y of} a specified set of \emph{initial states to reach \underline{onl}y} some specified set of (safe) \emph{final states}.
%\citet{Hoa69} logic is a well-studied program logic for proving correctness.
More recently, \emph{incorrectness} logic has attracted considerable attention \cite{OHearn19,zhang2022incorrectness,ZhangKaminski22,zilberstein2023outcome}.
Here, the property of interest is to prove that an entire specified set of undesired \emph{final states \underline{is reachable} from} a specified set of \emph{initial states}, thus proving the (true positive) presence of a bug.
Program correctness and incorrectness were described as \enquote{two sides of the same coin} \cite{OHearn19}.
We argue that there are at least three sides or rather \emph{dimensions} to the program logic coin:%
\begin{enumerate}
	\item
		\emph{correctness} (being able to reach) vs.~\emph{incorrectness} (being reachable)
	\item
		\emph{totality} vs.~\emph{partiality}
	\item
		\emph{angelic} vs.~\emph{demonic} resolution of nondeterminism
\end{enumerate}%
We will explore how one can pigeonhole \emph{total} and \emph{partial correctness}, as well as \emph{incorrectness} into this view and explore further program properties that emerge from exhausting all combinations of the above dimensions.
Our primary objective is to build a taxonomy of program logics that is in a sense exhaustive and provides an overview of existing logics, which are often referenced to under different names (e.g.\ reverse Hoare logic \cite{VK11} essentially is the same as incorrectness logic \cite{OHearn19}).
% To this end, we use predicate transformers à la \citet{dijkstra1976discipline} and the language of Kleene algebras with top and tests \cite{zhang2022incorrectness} for expressing program properties.

% 2. formalization of program properties
% \paragraph{Related Work}
% 
We choose \emph{predicate transformers} \cite{dijkstra1976discipline} as our primary formalism for expressing program properties. 
For example, using \emph{weakest liberal preconditions}, we can express the validity of partial correctness Hoare triples as $\pre \subseteq \wlp{\program}{\post}$, expressing that the precondition $\pre$ is included in the set $\wlp{\program}{\post}$ of states on which $\program$ either terminates in $\post$ or not at all.
%We will study a in total 8 different predicate transformers and by under- and overapproximating the resulting sets construct 16 different program logics out of them.
As a secondary formalism, we choose \emph{Kleene algebra with tests} \cite{kozen1997kleene} or rather \emph{Kleene algebra with top and tests (\TopKAT)} \cite{zhang2022incorrectness}.
The basic idea is to model programs (as well as pre- and postconditions) as elements of a relational algebra.
We then build \emph{terms} by composing these relations and express program properties as equations between terms.
For instance, partial correctness can be expressed in \TopKAT as $\top \pre \program \post = \top \pre \program$.
This essentially expresses that the codomains of the relations $\pre \program \post$ (filter initial states for $\pre$, execute $\program$, filter final states for $\post$) and $\pre \program$ (filter initial states for $\pre$, execute $\program$) are~equal.

\paragraph{Related Work}
Given a recent surge in novel program logics, there has been substantial interest in developing frameworks which unify those logics.
\citet{zilberstein2023outcome} introduce \emph{outcome logic}, a framework capable of accommodating both correctness and incorrectness reasoning.
Similarly, \citet{cousot2024calculational} explores program logics through the lens of abstract interpretation, constructing and comparing various logics within that framework.
\citet{ascari2023sufficient} contrast Hoare logic, incorrectness logic, and the related logic of \emph{necessary preconditions}.
In their resulting taxonomy, they address a gap by introducing what they term \emph{sufficient incorrectness logic}.
More work in this direction includes \cite{wickopedia,bruni2023correctness,maksimovic2022exact,milanese2022local}.
We give many more pointers to related work scattered across the paper, in particular \Cref{ssec:comparison}.

\paragraph{Contributions \& Organization}
In \Cref{sec:preliminiaries}, we present syntax and semantics of a simple nondeterministic model programming language.
In \Cref{sec:transformer}, we (gently and systematically) introduce 8 different predicate transformers à la Dijkstra.
Some were described before, some are new, emerging from an investigation of nondeterminism in forward analysis.

We furthermore identify a simple, yet crucial difference between forward and backward analyses:
Given a program $\program$ and an initial state $\sigma$, executing $\program$ on $\sigma$ may (nondeterministically) both terminate in some final state (execution leads to somewhere) and also not terminate at all (execution leads to nowhere).
However, given a final state $\tau$, it cannot be that by executing $\program$ the state $\tau$ was both reachable (execution came from somewhere) and unreachable (execution cam from nowhere).

In \Cref{sec:taxonomy}, we use these predicate transformers to define 16 different program logics (by over and underapproximating each of the 8 transformers).
These include well-known ones like partial and total correctness Hoare logic and incorrectness logic, but also lesser known (Lisbon logic) and new ones.
We study relationships across these logics, thus building a taxonomy of them.
We furthermore study how this taxonomy (partly) collapses under different assumptions like e.g.\ program termination.

In \Cref{sec:kat}, we study how to express various of these 16 program logics in Kleene algebra with top and tests, \emph{including \underline{Lisbon lo}g\underline{ic}, which to the best of our knowledge is a new result.}
In \Cref{sec:taxonomy-rev}, we give an updated taxonomy through the lens of \TopKAT.

Throughout this paper, we will discover several incongruities in an otherwise quite symmetrical and dual taxonomy.
We summarize these in \Cref{sec:open}.

% Version comment: this is the extended version
% Proofs, definitions, and additional examples are provided in an extended version of this paper \cite{verscht2025taxonomy}.
This is the extended version of the publication at POPL 2025.
% !TEX root = ../main.tex

\section{The Nondeterministic Guarded Command Language}
\label{sec:preliminiaries}
We consider programs in a simple \emph{nondeterministic guarded command language} ($\ngcl$):%
\begin{align*}
	\program  ~{}~{}\,{}\Coloneqq&{}~{}~{}\, 
	\SKIP 
	\qmid \ASSIGN{x}{\ee}
	\qmid \COMPOSE{\program}{\program} 
	\qmid \NDCHOICE{\program}{\program} 
	\qmid \ITE{\guard}{\program}{\program} \\
	& \qmid \WHILEDO{\guard}{\program}
\end{align*}%
Here, $x\in\Vars$ is a variable.
The set of \emph{program states} ($\sigma$, $\tau$, \dots) is given by $\States = \setcomp{\sigma}{\sigma:\Vars\to\Ints}$, i.e.\ functions mapping program variables to integers.
By $\sigma\subst{x}{v}$, we denote a new state that is obtained from $\sigma$ by letting variable~$x\in\Vars$ evaluate to~$v\in\Ints$. 
Formally: 
$\sigma\subst{x}{v}(y) = v$, if $y = x$; and $\sigma\subst{x}{v}(y) = \sigma(y)$, otherwise.

In the grammar above, $\ee$ is an arithmetic expression, and $\guard$ is a \emph{predicate} or \emph{test} or \emph{condition} (we use these interchangeably).
For our intents and purposes, a predicate is a subset of program states $\guard \subseteq \States$, i.e.\ a function mapping program states to truth values $\{0,\, 1\}$. 
We denote by $\B = \{0,\, 1\}^\States$ the set of all predicates. 
%We often represent a predicate as a first-order formula $\varphi$ over the program variables, corresponding to the set $\{ \sigma \in \States \mid \sigma \models \varphi \}$.
The negation of a predicate $\pre$ is denoted by $\negate{\pre}$. % = \{ \sigma \in \States \mid \sigma \not \models \varphi \}$.
Given a program state $\sigma$, we denote by $\sigma(\xi)$ the evaluation of an arithmetic expression or predicate~$\xi$ in~$\sigma$. %, i.e. the value that is obtained by evaluating $\xi$ after replacing any occurrence of any variable~$x$ in~$\xi$ by the value~$\sigma(x)$. 

\newpage
The semantics of \ngcl programs is given in terms of a \emph{collecting} semantics (as is standard in program analysis, see~\cite{CousotC77,hecht,RY20}).%
\begin{definition}[Collecting Semantics for $\ngcl$ Programs]%
%	Let $\Conf = \powerset{\States}$ be the set of \emph{program configurations}, i.e.\ a single configuration is a \emph{set of program states}; 
	Let $S \subseteq \States$ be a set of program states and let $\eval{\guard}S = \setcomp{\sigma\in S}{\sigma \mmodels \guard}$ be a filtering of $S$ to only those states where the predicate $\guard$ holds. 
	The \emph{collecting semantics} $\eval{\program}\colon \Conf \to \Conf$ of an $\ngcl$ program $\program$ is defined inductively by%
	\allowdisplaybreaks%
	\begin{align*}
	    \eval{\SKIP} S  \eeq & S \tag{effectless program}\\
		\eval{\ASSIGN{x}{\ee}} S  \eeq & \{\sigma\subst{x}{\sigma(\ee)}\mid \sigma \in S\}\tag{assignment}\\
		\eval{\COMPOSE{\program_1}{\program_2}} S \eeq & (\eval{\program_2}\circ \eval{\program_1}) S \tag{sequential composition} \\
		\eval{\ITE{\guard}{\program_1}{\program_2}} S  \eeq & (\eval{\program_1}\circ \eval{\guard}) S \cup (\eval{\program_2}\circ \eval{\negate{\guard}}) S \tag{conditional choice}\\
		\eval{\WHILEDO{\guard}{\program}} S \eeq & \eval{\negate{\guard}}\bigl(\lfp X\mydot S \ccup \bigl(\eval{\program} \circ \eval{\guard} \bigr) X \bigr) \tag{loop}\\
		\eval{\NDCHOICE{\program_1}{\program_2}} S  \eeq &\eval{\program_1} S \cup \eval{\program_2}S~. \tag{nondeterministic choice}  
	\end{align*}%
	\allowdisplaybreaks[0]%
	By slight abuse of notation, we write $\eval{\program}(\sigma)$ for $\eval{\program}\{\sigma\}$.
	We denote the inverse semantics of $\program$ by $\evalInv{\program} T = \setcomp{ \sigma }{T \cap \eval{\program}(\sigma) \neq \emptyset}$.
	If a program $\program$ does not terminate, we say that \emph{$\program$ diverges}.
	\qedtriangle%
\end{definition}%
\noindent%
In words, for a program $\program \in \ngcl$, the collecting semantics $\eval{\program} S$ maps a set of initial program states $S \in \Conf$ to the set of all final states reachable from a state in $S$.
%\lv{\atblk die zwei sätze hier nötig? ich find eher nicht (ein reviewer hatte gefragt)}
%We will base the formal definitions of individual transformers on forward semantics.
%This is a deliberate design choice, as we could equally well have used inverse semantics.%
%
\begin{remark}[Diverging Programs]
\label{rem:diverge-semantics}
%	If a program $\program$ does \emph{not} terminate, e.g.\ $\program = \WHILEDO{\true}{\SKIP}$, then we say that $\program$ \emph{diverges}.
	Consider the program $\SKIP$ that does nothing and the program $\NDCHOICE{\SKIP}{\WHILEDO{\true}{\SKIP}}$ that nondeterministically decides between doing nothing and diverging.
	Their collecting semantics are given by%
	\begin{align*}
		\eval{\SKIP} S \eeq  S \cup \emptyset  \eeq & \eval{\SKIP} S \ccup \eval{\WHILEDO{\true}{\SKIP}}S \\
		\eeq & \eval{\NDCHOICE{\SKIP}{\WHILEDO{\true}{\SKIP}}}S~,
	\end{align*}%
	i.e.\ their collecting semantics coincide, despite their arguably different behaviors.%
	\qedtriangle%
\end{remark}%

% \begin{remark}[Unbounded Nondeterminism]
% 	\label{rem:nondeterminism}
% %	If a program $\program$ does \emph{not} terminate, e.g.\ $\program = \WHILEDO{\true}{\SKIP}$, then we say that $\program$ \emph{diverges}.
% 	We only allow bounded nondeterminsm.
% 	\todoin{?}
% 	\qedtriangle%
% \end{remark}%

% \lvcommentinline{
% 	Should we define another semantics including a bottom element for termination?
% 	Vorteil: ?
% 	Nachteil: Definition nicht obvious, man sieht den Zusammenhang zwischen sp und wp schlechter in den Definitionen, finde ich (siehe unten)
% }
% We extend the collecting semantics to include divergence by defining $\evalDiv{C}S\colon \Conf \cup \{\bot\} \to \Conf \cup \{\bot\} $ of an $\ngcl$ program $C$ as
% \[
% 	\evalDiv{C}S = \eval{C}S \ccup \{ \bot \mid \exists \sigma \in S.\  \text{computation started in } \sigma \text{ can diverge} \}
% \]
% Then for the programs $C_1$ and $C_2$ that were considered above,
% \[
% 	\evalDiv{C_1} S \eeq S \quad \text{and} \quad \evalDiv{C_2} S \eeq S \cup \{\bot\},
% \]
% so the extended collecting semantics can distinguish programs $C_1$ and $C_2$.
% This coincides with the definition of relational natural semantics and their angelic variant of Cousot \cite{cousot2024calculational}.
% !TEX root = ../main.tex

\section{Predicate Transformers}
\label{sec:transformer}

Deductive reasoning on source code level with \emph{predicate transformers} is due to~\citet{DBLP:journals/cacm/Dijkstra75}.
There are fundamentally two types of predicate transformers: backward moving weakest preconditions and forward moving strongest postconditions.
The terms \emph{weakest} and \emph{strongest} are rooted historically in their relationship to Hoare logic.
They loose their plausibility in other contexts like incorrectness logic, but -- also for lack of better names -- we stick here to these historical terms.
%This connection will become evident in \Cref{sec:taxonomy}, where we use predicate transformers to construct various program logics.

\subsection{Weakest Preconditions}
\label{ssec:wp}

The weakest precondition calculus features predicate transformers of type
%\begin{align*}
	$
	\wpC{\program}\colon \B \to \B
	$
%\end{align*}%
%associating to each program $\program \in \ngcl$ and predicate another predicate.
for a program $\program \in \ngcl$.
Given a postcondition $\post \in \B$, the weakest precondition of $\post$ (under~$\program$) is a predicate~$\wp{\program}{\post}$ containing precisely those (initial) states from which the computation of $\program$ (1)~terminates and (2)~does so in a (final) state satisfying $\post$.
We will also say that $\wp{\program}{\post}$ are the initial states that are \emph{coreachable} from postcondition $\post$ by executing $\program$.

In the presence of nondeterminism, we must decide whether, for each initial state $\sigma$, we require \emph{all} computation paths emerging from $\sigma$ to terminate in $\post$ or whether we require merely the \emph{existence} of such a path.
Traditionally, the former (\emph{all}) is referred to as \emph{demonic} nondeterminism, and the latter (\emph{exist}) as \emph{angelic}.
These two design choices can be accommodated within Dijkstra's calculi, resulting in two slightly different variants of weakest precondition transformers.%
\begin{definition}[Weakest Precondition Transformers \textnormal{\cite{dijkstra1976discipline}}]
	\label{def:wp}
	Given a program $\program$ and a postcondition $\post$, the \emph{\underline{an}g\underline{elic} weakest precondition} is defined as%
	\begin{align*}
    		\awp{\program}{\post} &\qeq \mylambda{\sigma} \bigvee_{\tau \in \sem{\program}{\sigma}} \post(\tau)~.
	\intertext{The \emph{\underline{demonic} weakest precondition} is defined as}
		\dwp{\program}{\post} &\qeq \mylambda{\sigma}
		\begin{cases}
			\false, & \text{ if $\program$ can diverge on } \sigma \\[.5em]
			{\displaystyle \bigwedge_{\tau \in \sem{\program}{\sigma}} \post(\tau)}, & \text{ otherwise}~.
		\end{cases}%
		\tag*{\raisebox{-1.8em}{\qedtriangle}}
		%
		%\qquad \bigwedge\limits_{\tau \in \sem{\program}{\sigma}} \post(\tau) \lland \bot \not \in \semDiv{\program}{\sigma}
	\end{align*}%
\end{definition}%
\noindent%
It is worthwhile to note that $\dwpsymbol$ is the standard transformer introduced by Dijkstra.

Both $\awpsymbol$ and $\dwpsymbol$ are \emph{total correctness} transformers in the sense that both deem nontermination undesired behavior.
More precisely, $\dwp{\program}{\post}$ indeed contains no initial states whatsoever on which $\program$ could possibly diverge.
$\awp{\program}{\post}$, on the other hand, may contain such initial states $\sigma$, but there must then also exist a (nondeterministic) possibility for $\sigma$ to terminate in $\post$.

\subsubsection{Liberality}

Proving total correctness can be separated into two subtasks, namely (1) proving \emph{partial correctness} (i.e.\ correctness modulo termination) and (2) proving \emph{termination}.
Proving partial correctness motivates the definition of \emph{liberal} variants of the aforementioned calculi:
Given a postcondition $\post$, the weakest \emph{liberal} precondition of $\post$ (under~$\program$) is a predicate~$\wlp{\program}{\post}$ containing precisely those states from which the computation of $\program$ either (1)~diverges or (2)~terminates in a state satisfying $\post$.
In other words, the computation reaches $\post$, \emph{if} it terminates at all.
As in the non-liberal case, we again have both demonic and angelic variants of $\wlpsymbol$.%
\begin{definition}[Weakest Liberal Precondition Transformers \textnormal{\cite{dijkstra1976discipline}}]
	\label{def:wlp}
	Given a program~$\program$ and a postcondition~$\post$, the \emph{\underline{demonic} weakest liberal precondition} is defined as%
	\begin{align*}
		\dwlp{\program}{\post} &\eeq \mylambda{\sigma} \bigwedge\limits_{\tau \in \sem{\program}{\sigma}} \post(\tau)~.
	\intertext{The \emph{\underline{an}g\underline{elic} weakest liberal precondition} is defined as}
		\awlp{\program}{\post} & \eeq \mylambda{\sigma}
		\begin{cases}
			\true, & \text{ if $\program$ can diverge on } \sigma \\[.5em]
			{\displaystyle \bigvee\limits_{\tau \in \sem{\program}{\sigma}} \post(\tau)}, & \text{ otherwise}~.
		\end{cases} 
		\tag*{\raisebox{-1.8em}{\qedtriangle}}
		% alternative 
		% & \eeq \bigvee\limits_{\tau \in \sem{\program}{\sigma}} \post(\tau) \llor \bot \in \semDiv{\program}{\sigma}
	\end{align*}%
\end{definition}%
\noindent%
%
%
% \lvcommentinline{
% 	Nachteil der bot Definition: 
% 	Man sieht den Zusammenhang zu den sp definitionen mit der Fallunterscheidung eigentlich besser, weil da braucht man die eh.
% 	Außer man führt halt noch ein top symbol ein oder so, aber das ist halt eigenltich unnötig?
% 	Aber wäre für die symmetrie vielleicht das sinnvollere, wenn man sicher das bot will.
% 	Aber eigentlich ist das auch unnötig das jetzt nur für die soundness definition sozusagen von den Transformern einzuführen.
% 	Also die Frage: was ist am schönsten zu lesen und zu verstehen?
% }
Again, we note that the demonic variant $\dwlpsymbol$ is the standard one studied by Dijkstra.

Both $\dwlpsymbol$ and $\awlpsymbol$ are \emph{partial correctness} transformers in the sense that both deem nontermination as acceptable behavior.
More precisely, $\dwlp{\program}{\post}$ indeed contains all initial states on which $\program$ \emph{must} either diverge or terminate in $\post$.
On the other hand, $\awlp{\program}{\post}$ contains all initial states on which $\program$ \emph{may} either diverge or terminate in $\post$.

\subsubsection{Inductive Definitions of Weakest Precondition Transformers}

The angelic/demonic weakest (liberal) precondition transformers can \emph{all} be defined by induction on the program structure, see \Cref{ssec:wp-rules,ssec:wlp-rules}. %\cite[Appendices E.1 and E.2]{verscht2025taxonomy}.

We note also that the direction of analysis is \emph{backwards}, as we start with a predicate over final states and transform it into a predicate on initial states.
The result of the analysis on the other hand is a \emph{forecast}: a weakest precondition forecasts for each initial state whether after the computation the postcondition will be satisfied.

\newpage

\begin{wrapfigure}[19]{r}{.6\textwidth}%
	\vspace{-1\intextsep}%
	\begin{adjustbox}{max width=.599\textwidth}
	\begin{tikzpicture}
		% Define the width and height of the blocks
		\def\topBlockHeight{0.6}
		\def\bottomBlockHeight{0.6}
		\def\topBlockWidth{1.5}
		\def\bottomBlockWidth{5.25}
		\def\blockSpacing{0}
		\def\blockSpacingVert{2.8}
		\def\spaceTransformers{0.2}

		% Draw the top block split into 7 parts
		% \foreach \i in {0,1, ..., 6} { %
		% 	\draw[draw=none] (\i * \topBlockWidth + \i * \blockSpacing, \blockSpacingVert) rectangle (\i * \topBlockWidth + \i * \blockSpacing + \topBlockWidth, \blockSpacingVert + \topBlockHeight);
		% }
		
		% Draw dwp block (loop does not work with fill)
		\def\j{1}
		\def\i{0}
		\draw[pattern=crosshatch,pattern color=mygreen!70] (\i * \topBlockWidth + \i * \blockSpacing, \blockSpacingVert+\spaceTransformers+\j*\topBlockHeight) rectangle (\i * \topBlockWidth + \i * \blockSpacing + \topBlockWidth, \blockSpacingVert+\spaceTransformers + \j*\topBlockHeight + \topBlockHeight); % node[midway] {\textsf{X}};
		\def\i{1}
		\draw[] (\i * \topBlockWidth + \i * \blockSpacing, \blockSpacingVert+\spaceTransformers+\j*\topBlockHeight) rectangle (\i * \topBlockWidth + \i * \blockSpacing + \topBlockWidth, \blockSpacingVert+\spaceTransformers + \j*\topBlockHeight + \topBlockHeight);
		\def\i{2}
		\draw[] (\i * \topBlockWidth + \i * \blockSpacing, \blockSpacingVert+\spaceTransformers+\j*\topBlockHeight) rectangle (\i * \topBlockWidth + \i * \blockSpacing + \topBlockWidth, \blockSpacingVert+\spaceTransformers + \j*\topBlockHeight + \topBlockHeight);
		\def\i{3}
		\draw[] (\i * \topBlockWidth + \i * \blockSpacing, \blockSpacingVert+\spaceTransformers+\j*\topBlockHeight) rectangle (\i * \topBlockWidth + \i * \blockSpacing + \topBlockWidth, \blockSpacingVert+\spaceTransformers  + \j*\topBlockHeight+ \topBlockHeight);
		\def\i{4}
		\draw[] (\i * \topBlockWidth + \i * \blockSpacing, \blockSpacingVert+\spaceTransformers+\j*\topBlockHeight) rectangle (\i * \topBlockWidth + \i * \blockSpacing + \topBlockWidth, \blockSpacingVert+\spaceTransformers  + \j*\topBlockHeight +\topBlockHeight);
		\def\i{5}
		\draw[] (\i * \topBlockWidth + \i * \blockSpacing, \blockSpacingVert+\spaceTransformers+\j*\topBlockHeight) rectangle (\i * \topBlockWidth + \i * \blockSpacing + \topBlockWidth, \blockSpacingVert+\spaceTransformers  + \j*\topBlockHeight + \topBlockHeight);
		\def\i{6}
		\draw[] (\i * \topBlockWidth + \i * \blockSpacing, \blockSpacingVert+\spaceTransformers+\j*\topBlockHeight) rectangle (\i * \topBlockWidth + \i * \blockSpacing + \topBlockWidth, \blockSpacingVert+\spaceTransformers  + \j*\topBlockHeight + \topBlockHeight);

		% Draw awp block (loop does not work with fill)
		\def\j{2}
		\def\i{0}
		\draw[pattern=crosshatch,pattern color=mygreen!70] (\i * \topBlockWidth + \i * \blockSpacing, \blockSpacingVert+\spaceTransformers+\j*\topBlockHeight) rectangle (\i * \topBlockWidth + \i * \blockSpacing + \topBlockWidth, \blockSpacingVert+\spaceTransformers + \j*\topBlockHeight + \topBlockHeight);% node[midway] {\textsf{X}};
		\def\i{1}
		\draw[pattern=crosshatch,pattern color=mygreen!70] (\i * \topBlockWidth + \i * \blockSpacing, \blockSpacingVert+\spaceTransformers+\j*\topBlockHeight) rectangle (\i * \topBlockWidth + \i * \blockSpacing + \topBlockWidth, \blockSpacingVert+\spaceTransformers + \j*\topBlockHeight + \topBlockHeight);%  node[midway] {\textsf{X}};
		\def\i{2}
		\draw[pattern=crosshatch,pattern color=mygreen!70] (\i * \topBlockWidth + \i * \blockSpacing, \blockSpacingVert+\spaceTransformers+\j*\topBlockHeight) rectangle (\i * \topBlockWidth + \i * \blockSpacing + \topBlockWidth, \blockSpacingVert+\spaceTransformers + \j*\topBlockHeight + \topBlockHeight);%  node[midway] {\textsf{X}};
		\def\i{3}
		\draw[] (\i * \topBlockWidth + \i * \blockSpacing, \blockSpacingVert+\spaceTransformers+\j*\topBlockHeight) rectangle (\i * \topBlockWidth + \i * \blockSpacing + \topBlockWidth, \blockSpacingVert+\spaceTransformers  + \j*\topBlockHeight+ \topBlockHeight);
		\def\i{4}
		\draw[pattern=crosshatch,pattern color=mygreen!70] (\i * \topBlockWidth + \i * \blockSpacing, \blockSpacingVert+\spaceTransformers+\j*\topBlockHeight) rectangle (\i * \topBlockWidth + \i * \blockSpacing + \topBlockWidth, \blockSpacingVert+\spaceTransformers  + \j*\topBlockHeight +\topBlockHeight); % node[midway] {\textsf{X}};
		\def\i{5}
		\draw[] (\i * \topBlockWidth + \i * \blockSpacing, \blockSpacingVert+\spaceTransformers+\j*\topBlockHeight) rectangle (\i * \topBlockWidth + \i * \blockSpacing + \topBlockWidth, \blockSpacingVert+\spaceTransformers  + \j*\topBlockHeight + \topBlockHeight);
		\def\i{6}
		\draw[] (\i * \topBlockWidth + \i * \blockSpacing, \blockSpacingVert+\spaceTransformers+\j*\topBlockHeight) rectangle (\i * \topBlockWidth + \i * \blockSpacing + \topBlockWidth, \blockSpacingVert+\spaceTransformers  + \j*\topBlockHeight + \topBlockHeight);

		% Draw dwlp block (loop does not work with fill)
		\def\j{3}
		\def\i{0}
		\draw[pattern=crosshatch,pattern color=mygreen!70] (\i * \topBlockWidth + \i * \blockSpacing, \blockSpacingVert+\spaceTransformers+\j*\topBlockHeight) rectangle (\i * \topBlockWidth + \i * \blockSpacing + \topBlockWidth, \blockSpacingVert+\spaceTransformers + \j*\topBlockHeight + \topBlockHeight); % node[midway] {\textsf{X}};
		\def\i{1}
		\draw[pattern=crosshatch,pattern color=mygreen!70] (\i * \topBlockWidth + \i * \blockSpacing, \blockSpacingVert+\spaceTransformers+\j*\topBlockHeight) rectangle (\i * \topBlockWidth + \i * \blockSpacing + \topBlockWidth, \blockSpacingVert+\spaceTransformers + \j*\topBlockHeight + \topBlockHeight); % node[midway] {\textsf{X}};
		\def\i{2}
		\draw[] (\i * \topBlockWidth + \i * \blockSpacing, \blockSpacingVert+\spaceTransformers+\j*\topBlockHeight) rectangle (\i * \topBlockWidth + \i * \blockSpacing + \topBlockWidth, \blockSpacingVert+\spaceTransformers + \j*\topBlockHeight + \topBlockHeight);
		\def\i{3}
		\draw[pattern=crosshatch,pattern color=mygreen!70] (\i * \topBlockWidth + \i * \blockSpacing, \blockSpacingVert+\spaceTransformers+\j*\topBlockHeight) rectangle (\i * \topBlockWidth + \i * \blockSpacing + \topBlockWidth, \blockSpacingVert+\spaceTransformers  + \j*\topBlockHeight+ \topBlockHeight);%  node[midway] {\textsf{X}};
		\def\i{4}
		\draw[] (\i * \topBlockWidth + \i * \blockSpacing, \blockSpacingVert+\spaceTransformers+\j*\topBlockHeight) rectangle (\i * \topBlockWidth + \i * \blockSpacing + \topBlockWidth, \blockSpacingVert+\spaceTransformers  + \j*\topBlockHeight +\topBlockHeight);
		\def\i{5}
		\draw[] (\i * \topBlockWidth + \i * \blockSpacing, \blockSpacingVert+\spaceTransformers+\j*\topBlockHeight) rectangle (\i * \topBlockWidth + \i * \blockSpacing + \topBlockWidth, \blockSpacingVert +\spaceTransformers + \j*\topBlockHeight + \topBlockHeight);
		\def\i{6}
		\draw[] (\i * \topBlockWidth + \i * \blockSpacing, \blockSpacingVert+\spaceTransformers+\j*\topBlockHeight) rectangle (\i * \topBlockWidth + \i * \blockSpacing + \topBlockWidth, \blockSpacingVert+\spaceTransformers  + \j*\topBlockHeight + \topBlockHeight);

		% Draw awlp block (loop does not work with fill)
		\def\j{4}
		\def\i{0}
		\draw[pattern=crosshatch,pattern color=mygreen!70] (\i * \topBlockWidth + \i * \blockSpacing, \blockSpacingVert+\spaceTransformers+\j*\topBlockHeight) rectangle (\i * \topBlockWidth + \i * \blockSpacing + \topBlockWidth, \blockSpacingVert+\spaceTransformers + \j*\topBlockHeight + \topBlockHeight); % node[midway] {\textsf{X}};
		% \draw[] (\i * \topBlockWidth + \i * \blockSpacing, \blockSpacingVert+\spaceTransformers+\j*\topBlockHeight) -- (\i * \topBlockWidth + \i * \blockSpacing + \topBlockWidth, \blockSpacingVert+\spaceTransformers + \j*\topBlockHeight + \topBlockHeight);
		% \draw[] (\i * \topBlockWidth + \i * \blockSpacing, \blockSpacingVert+\spaceTransformers + \j*\topBlockHeight + \topBlockHeight) -- (\i * \topBlockWidth + \i * \blockSpacing + \topBlockWidth, \blockSpacingVert+\spaceTransformers+\j*\topBlockHeight);
		\def\i{1}
		\draw[pattern=crosshatch,pattern color=mygreen!70] (\i * \topBlockWidth + \i * \blockSpacing, \blockSpacingVert+\spaceTransformers+\j*\topBlockHeight) rectangle (\i * \topBlockWidth + \i * \blockSpacing + \topBlockWidth, \blockSpacingVert+\spaceTransformers + \j*\topBlockHeight + \topBlockHeight); % node[midway] {\textsf{X}};
		\def\i{2}
		\draw[pattern=crosshatch,pattern color=mygreen!70] (\i * \topBlockWidth + \i * \blockSpacing, \blockSpacingVert+\spaceTransformers+\j*\topBlockHeight) rectangle (\i * \topBlockWidth + \i * \blockSpacing + \topBlockWidth, \blockSpacingVert+\spaceTransformers + \j*\topBlockHeight + \topBlockHeight); % node[midway] {\textsf{X}};
		\def\i{3}
		\draw[pattern=crosshatch,pattern color=mygreen!70] (\i * \topBlockWidth + \i * \blockSpacing, \blockSpacingVert+\spaceTransformers+\j*\topBlockHeight) rectangle (\i * \topBlockWidth + \i * \blockSpacing + \topBlockWidth, \blockSpacingVert+\spaceTransformers  + \j*\topBlockHeight+ \topBlockHeight); % node[midway] {\textsf{X}};
		\def\i{4}
		\draw[pattern=crosshatch,pattern color=mygreen!70] (\i * \topBlockWidth + \i * \blockSpacing, \blockSpacingVert+\spaceTransformers+\j*\topBlockHeight) rectangle (\i * \topBlockWidth + \i * \blockSpacing + \topBlockWidth, \blockSpacingVert +\spaceTransformers + \j*\topBlockHeight +\topBlockHeight); % node[midway] {\textsf{X}};
		\def\i{5}
		\draw[pattern=crosshatch,pattern color=mygreen!70] (\i * \topBlockWidth + \i * \blockSpacing, \blockSpacingVert+\spaceTransformers+\j*\topBlockHeight) rectangle (\i * \topBlockWidth + \i * \blockSpacing + \topBlockWidth, \blockSpacingVert +\spaceTransformers + \j*\topBlockHeight + \topBlockHeight); % node[midway] {\textsf{X}};
		\def\i{6}
		\draw[] (\i * \topBlockWidth + \i * \blockSpacing, \blockSpacingVert+\spaceTransformers+\j*\topBlockHeight) rectangle (\i * \topBlockWidth + \i * \blockSpacing + \topBlockWidth, \blockSpacingVert +\spaceTransformers + \j*\topBlockHeight + \topBlockHeight);

		% Draw numbers (loop does not work with fill)
		\def\j{5}
		\def\i{0}
		\draw[draw=none] (\i * \topBlockWidth + \i * \blockSpacing, \blockSpacingVert+\spaceTransformers+\j*\topBlockHeight) rectangle (\i * \topBlockWidth + \i * \blockSpacing + \topBlockWidth, \blockSpacingVert+\spaceTransformers + \j*\topBlockHeight + \topBlockHeight) node[midway] {(1)};
		\def\i{1}
		\draw[draw=none] (\i * \topBlockWidth + \i * \blockSpacing, \blockSpacingVert+\spaceTransformers+\j*\topBlockHeight) rectangle (\i * \topBlockWidth + \i * \blockSpacing + \topBlockWidth, \blockSpacingVert+\spaceTransformers + \j*\topBlockHeight + \topBlockHeight) node[midway] {(2)};
		\def\i{2}
		\draw[draw=none] (\i * \topBlockWidth + \i * \blockSpacing, \blockSpacingVert+\spaceTransformers+\j*\topBlockHeight) rectangle (\i * \topBlockWidth + \i * \blockSpacing + \topBlockWidth, \blockSpacingVert+\spaceTransformers + \j*\topBlockHeight + \topBlockHeight) node[midway] {(3)};
		\def\i{3}
		\draw[draw=none] (\i * \topBlockWidth + \i * \blockSpacing, \blockSpacingVert+\spaceTransformers+\j*\topBlockHeight) rectangle (\i * \topBlockWidth + \i * \blockSpacing + \topBlockWidth, \blockSpacingVert+\spaceTransformers  + \j*\topBlockHeight+ \topBlockHeight) node[midway] {(4)};
		\def\i{4}
		\draw[draw=none] (\i * \topBlockWidth + \i * \blockSpacing, \blockSpacingVert+\spaceTransformers+\j*\topBlockHeight) rectangle (\i * \topBlockWidth + \i * \blockSpacing + \topBlockWidth, \blockSpacingVert +\spaceTransformers + \j*\topBlockHeight +\topBlockHeight) node[midway] {(5)};
		\def\i{5}
		\draw[draw=none] (\i * \topBlockWidth + \i * \blockSpacing, \blockSpacingVert+\spaceTransformers+\j*\topBlockHeight) rectangle (\i * \topBlockWidth + \i * \blockSpacing + \topBlockWidth, \blockSpacingVert +\spaceTransformers + \j*\topBlockHeight + \topBlockHeight) node[midway] {(6)};
		\def\i{6}
		\draw[draw=none] (\i * \topBlockWidth + \i * \blockSpacing, \blockSpacingVert+\spaceTransformers+\j*\topBlockHeight) rectangle (\i * \topBlockWidth + \i * \blockSpacing + \topBlockWidth, \blockSpacingVert +\spaceTransformers + \j*\topBlockHeight + \topBlockHeight) node[midway] {(7)};

		% Draw the bottom block split into 2 parts
		\draw[pattern=crosshatch,pattern color=DodgerBlue3!70] (0, 0) rectangle (\bottomBlockWidth, \bottomBlockHeight);
		\draw[] (\bottomBlockWidth + \blockSpacing, 0) rectangle (\bottomBlockWidth + \blockSpacing + \bottomBlockWidth, \bottomBlockHeight);

		% Add labels to the top block parts
		\foreach \i in {0, 1, ..., 6} {
			\node (top\i) at (\i * \topBlockWidth + \i * \blockSpacing + \topBlockWidth/2, \blockSpacingVert + \topBlockHeight + 0.15) {};
		}

		% Add labels to the transformer parts
		\node (dwp) at (-0.5, \blockSpacingVert + \spaceTransformers + \topBlockHeight/2 + \topBlockHeight) {$\dwpsymbol$};
		\node (awp) at (-0.5, \blockSpacingVert + \spaceTransformers + \topBlockHeight/2 + 2*\topBlockHeight) {$\awpsymbol$};
		\node (dwlp) at (-0.5, \blockSpacingVert + \spaceTransformers + \topBlockHeight/2 + 3*\topBlockHeight) {$\dwlpsymbol$};
		\node (awlp) at (-0.5, \blockSpacingVert + \spaceTransformers + \topBlockHeight/2 + 4*\topBlockHeight) {$\awlpsymbol$};
		
%		\node[above=top0] {sup};

		% \node (bot0) at (0.8*\bottomBlockWidth/2, \bottomBlockHeight/2) {$\post$};
		% \node (bot1) at (0.8*\bottomBlockWidth + \blockSpacing + \bottomBlockWidth/2, \bottomBlockHeight/2) {$\neg \post$};
		
		\draw [decorate,decoration={brace,amplitude=5pt,mirror,raise=2ex}]
		(0.1,0) -- (\bottomBlockWidth-0.1,0) node[midway,yshift=-2em]{$\post$};
		\draw [decorate,decoration={brace,amplitude=5pt,mirror,raise=2ex}]
		(\bottomBlockWidth+0.1,0) -- (2*\bottomBlockWidth-0.1,0) node[midway,yshift=-2em]{$\negate{\post}$};

		% Add extra labels for nicer arrows
		\foreach \i in {0, 1, ..., 6} {
			\node (botref\i) at (\i * \topBlockWidth + \i * \blockSpacing + \topBlockWidth/2, \bottomBlockHeight) {};
		}

		\path (top0) edge[program] (botref0);

		\path (top1) edge[program] (botref1);
		\node[rectangle,fill=black,scale=0.3] (choice1) at ($(top1)!0.5!(botref1)$) {b};
		\node (spiralstart1) at ($(choice1) - (0.445,0)$) {};		
		\draw [scale=0.25,xscale=1,yscale=1,rotate=-90,domain=0:30,variable=\t,smooth,samples=75,thick,-,shift={(spiralstart1)}] plot  ({\t r}: {-0.002*\t*\t});	
		
		\path (top2) edge[program] (botref2);
		\node[rectangle,fill=black,scale=0.3] (choice2) at ($(top2)!0.4!(botref2)$) {b};
		\draw [rounded corners=4mm,program] (choice2) -- ++(3mm,-6mm) -| ($(botref3) + (0.5,0.12)$);;

		\node (mid3) at ($(top3)!0.33!(botref3)$) {};
		\path (top3) edge[thick] (mid3);
		\node (spiralstart3) at ($(mid3) + (0.445,0.1)$) {};	
		\draw [scale=0.25,xscale=-1,yscale=1,rotate=-90,domain=0:30,variable=\t,smooth,samples=75,thick,-,shift={(spiralstart3)}] plot  ({\t r}: {-0.002*\t*\t});
	
		\path (top4) edge[program] (botref4);
		\node[rectangle,fill=black,scale=0.3] (choice4) at ($(top4)!0.5!(botref4)$) {b};
		\draw [rounded corners=4mm,program] (choice4) -- ++(-3mm,-8mm) -| ($(botref2) + (1,0.12)$);
		\node (spiralstart4) at ($(choice4) + (0.445,0)$) {};		
		\draw [scale=0.25,xscale=-1,yscale=1,rotate=-90,domain=0:30,variable=\t,smooth,samples=75,thick,-,shift={(spiralstart4)}] plot  ({\t r}: {-0.002*\t*\t});

		\path (top5) edge[program] (botref5);
		\node[rectangle,fill=black,scale=0.3] (choice5) at ($(top5)!0.5!(botref5)$) {b};
		\node (spiralstart5) at ($(choice5) + (0.445,0)$) {};	
		\draw [scale=0.25,xscale=-1,yscale=1,rotate=-90,domain=0:30,variable=\t,smooth,samples=75,thick,-,shift={(spiralstart5)}] plot  ({\t r}: {-0.002*\t*\t});

		\path (top6) edge[program] (botref6);
	\end{tikzpicture}
	\end{adjustbox}
	\caption{
		Illustration of different coreachability classes and different $\wpsymbol$ transformers.
		The top part represents initial states, divided (in columns) into coreachability classes.
		The lower part represents final states, divided into those that satisfy postcondition $\post$ and those that do not.
		On top, the colored boxed indicate which coreachability classes are included in which transformers.
	}
	\label{fig:wp}
\end{wrapfigure}
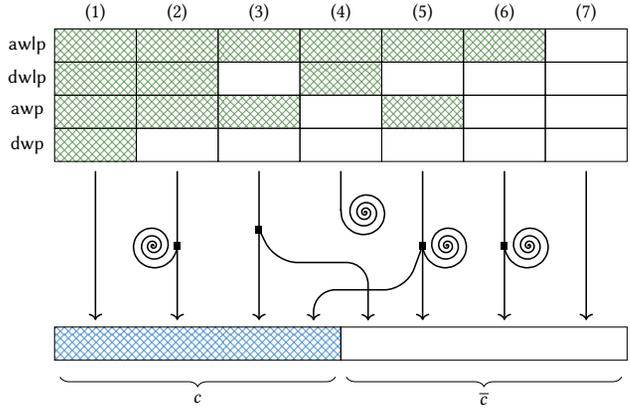%
\subsubsection{Anatomy of Weakest Precondition Transformers}

Given an initial state $\sigma$, a program $\program$, and a postcondition $\post$, we can make out three behavioral dimensions with respect to how $\program$ behaves when executed on $\sigma$:%
\begin{enumerate}
	\item[(wd1)] 
		$\program$ can terminate in $\post$ or not,
		
	\item[(wd2)]
		$\program$ can terminate in $\negate{\post}$ or not,
		
	\item[(wd3)]
		$\program$ can diverge or not.
\end{enumerate}%
In the presence of nondeterminism, these dimensions span a space of \mbox{$2^3 = 8$} different types of behaviors that $\program$ can exhibit.
For example, a program could nondeterministically both terminate in $\post$ or diverge.
We call the behavioral classes in that space \emph{coreachability classes}, since they speak about whether initial states are coreachable from $\post$ (or~$\negate{\post}$).
For example, there is a coreachability class containing all initial states from which $\program$ can both diverge and reach $\post$.
However, it cannot happen that $\program$ neither terminates in $\post$, nor in $\negate{\post}$, nor diverges.
This thus leaves \mbox{$2^3 - 1 = 7$} sensible coreachability classes, which are illustrated in \Cref{fig:wp}:
\begin{enumerate}
	\item $\program$ terminates in $\post$.
	\item $\program$ terminates in $\post$ or diverges (i.e.\ $\program$ does not terminate in $\negate{\post}$).
	\item $\program$ terminates in $\post$ or terminates in $\negate{\post}$ (i.e.\ $\program$ terminates).
	\item $\program$ diverges.
	\item $\program$ terminates in $\post$, or in $\negate{\post}$, or diverges (i.e.\ no restriction on behavior of $\program$).
	\item $\program$ terminates in $\negate{\post}$ or diverges (i.e.\ $\program$ does not terminate in $\post$).
	\item $\program$ terminates in $\negate{\post}$.
\end{enumerate}%
These coreachability classes fully partition the state space.
A precondition transformer can now opt to include a class ($i$) in its result or not.
The four rather natural transformers we described above are illustrated in \Cref{fig:wp}.
Along a \emph{row}, green shaded boxes indicate which classes are included in the respective weakest precondition transformer with respect to postcondition $\post$ (in blue).
For instance, $\awp{\program}{\post}$ includes classes (1), (2), (3), and~(5).

Assuming that the above coreachability classes are meaningful, this makes for $2^7 = 128$ possible weakest precondition transformers\footnote{Without nondeterminism, there would be only 3 coreachability classes (reach $\post$, reach $\negate{\post}$, diverge) and the number of possible precondition transformers would reduce to $2^3 = 8$.} of which some might not even depend on the postcondition~$\post$ (e.g.\ include only (4)), others might even be trivial (e.g.\ include none, include all).
Some are contrapositives of each other.
For example,\ $\awpsymbol$ and $\dwlpsymbol$ are \emph{contrapositive} in the sense that%
\begin{align*}
	\awp{\program}{\post} \eeq \negate{\dwlp{\program}{\negate{c}}} \qqand \dwlp{\program}{\post} \eeq \negate{\awp{\program}{\negate{c}}}
\end{align*}%
This contrapositivity relationship also holds for $\dwpsymbol$ and $\awlpsymbol$.
The contrapositivities can also be rediscovered graphically in \Cref{fig:wp}: 
To get from $\awlpsymbol$ to $\dwpsymbol$, for instance, proceed as follows: 
\begin{enumerate}
	\item Take row 1 ($\awlpsymbol$).
	\item Invert colors (i.e.\ turn green \mbox{shaded boxes} into empty ones and vice versa; corresponds to negating the entire result).
	\item Mirror the entire row horizontally (corresponds to negating the postcondition).
	\item Obtain row 4 ($\dwpsymbol$).
\end{enumerate}

\subsection{Strongest Postconditions}
\label{ssec:sp}

Weakest precondition transformers are well-researched and have been extended for various purposes, including quantities, probabilistic programs, and even quantum programs \cite{kaminski2018weakest,DBLP:journals/toplas/MorganMS96,d2006quantum}.
Dually to weakest preconditions, \citet{Dijkstra1990} defined \emph{strongest postcondition transformers}, also of type
%\begin{align*}
	$
	\spC{\program}\colon \B \to \B
	$.
%\end{align*}%
Given now a \emph{pre}condition $\pre \in \B$, the strongest postcondition of $\pre$ (under~$\program$) is a predicate~$\sp{\program}{\pre}$ containing precisely those (final) states that are \emph{reachable} from an (initial) state in $\pre$ \mbox{by executing $\program$}.%
\begin{definition}[Strongest Postcondition Transformer \textnormal{\cite{dijkstra1976discipline}}]%
\label{def:asp}%
	Given a program $\program \in \ngcl$ and a precondition $\pre \in \Pred$, the \emph{\underline{an}g\underline{elic} strongest postcondition} is defined as%
	\begin{align*}
		\asp{\program}{\pre} \eeq \mylambda{\tau} \bigvee\limits_{\sigma \in \seminv{\program}{\tau}} \pre(\sigma).
		\tag*{\raisebox{-1.3em}{\qedtriangle}}
	\end{align*}%
\end{definition}%
\noindent%
As suggested by the operator name $\aspsymbol$, we have opted here for \emph{angelic} resolution of nondeterminism.
This is also the standard choice of Dijkstra.
We will discuss the nature of this nondeterminism as well as demonic $\spsymbol$ later in \Cref{sssec:nondeterminism}.

\subsubsection{Liberality} 

Dually to strongest postconditions, there are also strongest liberal postconditions, originally described by Cousot and then later given an inductive definition by \citet{ZhangKaminski22}.
Given a \emph{pre}condition $\pre \in \B$, the strongest liberal postcondition of $\pre$ (under~$\program$) is a predicate~$\slp{\program}{\pre}$ containing precisely those (final) states that are reachable (1)~\emph{exclusively} from (initial) states in $\pre$ or (2)~\emph{entirely unreachable} (i.e.\ from \emph{any} initial state in $\States$) by the \mbox{computation of $\program$}.%
\begin{definition}[Strongest Liberal Postcondition Transformers \textnormal{\cite{ZhangKaminski22}}]
	\label{def:dslp}
	Given a program $\program$ and a precondition $\pre$, the \emph{\underline{demonic} strongest liberal postcondition} is defined as%
	\begin{align*}
		\dslp{\program}{\pre} \eeq \mylambda{\tau} \bigwedge\limits_{\sigma \in \seminv{\program}{\tau}} \pre(\sigma).
		\tag*{\raisebox{-1.3em}{\qedtriangle}}
	\end{align*}%
\end{definition}%
\noindent%
As the name $\dslpsymbol$ suggests, we now consider \emph{demonic} resolution of nondeterminism.
As shown above, strongest (liberal) postconditions can again be defined via the collecting semantics.
The inductive rules can be found in \Cref{ssec:sp-rules,ssec:slp-rules}. % \cite[Appendices E.3 and E.4]{verscht2025taxonomy}.
%Note that the liberal variant by Zhang and Kaminski originally was defined as a quantitative calculus.
%Here, we only consider its restriction to qualitative predicates.

The analysis direction is now \emph{forwards}: starting with a predicate on initial states, we transform it into a predicate on final states.
The result of the analysis on the other hand is a \emph{backcast}: a strongest postcondition backcasts for each final state whether the computation started in the precondition.

The \emph{demonic} strongest liberal post is intentional because it makes for a very dual theory as we will see later.
In the context of the strongest post reasoning, liberality refers to unreachability rather than divergence.
A state from which computation diverges does not have a final state that it reaches, i.e.\ it does not terminate.
Conversely, an unreachable state does not have an initial state that reaches it, i.e.\ it is not reachable from any initial state.
The duality of termination and reachability has previously been discussed by \citet{ZhangKaminski22} and \citet{ascari2023sufficient}.

The requirement of exclusive reachability - excluding states that can also be reached from outside of $\pre$ - is what makes this transformer demonic.
This concept will be further explored in the following section, where we discuss how nondeterminism arises in forward analyses, a topic that, to the best of our knowledge, has not been addressed in detail in the literature.

\subsubsection{Resolution of Nondeterminism in Strongest Postconditions}
\label{sssec:nondeterminism}

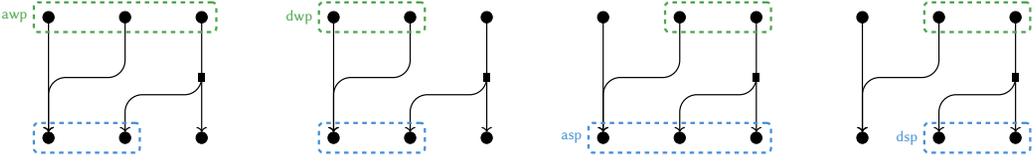
\begin{figure}[t]
	\def\statespace{1.8}
			\def\vertspace{2.8}
			\def\colorpre{webgreen!70}
			\def\colorpost{DodgerBlue3!80}
	\begin{adjustbox}{max width = .99\textwidth}
		\begin{tikzpicture}
				\node[state,fill=black,scale=0.3] (i0) at (0,0) {};
				\node[state,fill=black,scale=0.3] (i1) at (\statespace,0) {};
				\node[state,fill=black,scale=0.3] (i2) at (2*\statespace,0) {};

				\node[state,fill=black,scale=0.3] (f0) at (0,-\vertspace) {};
				\node[state,fill=black,scale=0.3] (f1) at (\statespace,-\vertspace) {};
				\node[state,fill=black,scale=0.3] (f2) at (2*\statespace,-\vertspace) {};
				
				\node[rectangle,fill=black,scale=0.4] (choice) at ($(i2)!0.5!(f2)$) {b};

				\path (i0) edge[program] (f0);
				\draw [rounded corners=4mm,program] (i1) |- ++(-\statespace,-\vertspace/2) -- (f0);
				\path (i2) edge[program] (f2);
				\draw [rounded corners=4mm,program] (choice) |- ++(-\statespace,-0.4) -- (f1);

				% ,inner xsep=\innerxsepp,inner ysep=\innerysepp,rounded corners=\roundedcorn
				\node (awp) [line width=1.5pt,inner sep=6pt, dashed,rounded corners=0.1cm,fit = {(i0)(i1)(i2)}, draw, color=\colorpre, label=left:{\color{\colorpre}$\awpsymbol$}] {};
				\node (post) [line width=1.5pt,inner sep=6pt, dashed,rounded corners=0.1cm,fit = {(f0)(f1)}, draw, color=\colorpost] {};
			
			\end{tikzpicture}
			\qquad\qquad
			\begin{tikzpicture}

				\node[state,fill=black,scale=0.3] (i0) at (0,0) {};
				\node[state,fill=black,scale=0.3] (i1) at (\statespace,0) {};
				\node[state,fill=black,scale=0.3] (i2) at (2*\statespace,0) {};

				\node[state,fill=black,scale=0.3] (f0) at (0,-\vertspace) {};
				\node[state,fill=black,scale=0.3] (f1) at (\statespace,-\vertspace) {};
				\node[state,fill=black,scale=0.3] (f2) at (2*\statespace,-\vertspace) {};
				
				\node[rectangle,fill=black,scale=0.4] (choice) at ($(i2)!0.5!(f2)$) {b};

				\path (i0) edge[program] (f0);
				\draw [rounded corners=4mm,program] (i1) |- ++(-\statespace,-\vertspace/2) -- (f0);
				\path (i2) edge[program] (f2);
				\draw [rounded corners=4mm,program] (choice) |- ++(-\statespace,-0.4) -- (f1);

				\node (dwp) [line width=1.5pt,inner sep=6pt, dashed,rounded corners=0.1cm,fit = {(i0)(i1)}, draw, color=\colorpre, label=left:{\color{\colorpre}$\dwpsymbol$}] {};
				\node (post) [line width=1.5pt,inner sep=6pt, dashed,rounded corners=0.1cm,fit = {(f0)(f1)}, draw, color=\colorpost] {};
			\end{tikzpicture}
			\qquad\qquad
			\begin{tikzpicture}

				\node[state,fill=black,scale=0.3] (i0) at (0,0) {};
				\node[state,fill=black,scale=0.3] (i1) at (\statespace,0) {};
				\node[state,fill=black,scale=0.3] (i2) at (2*\statespace,0) {};

				\node[state,fill=black,scale=0.3] (f0) at (0,-\vertspace) {};
				\node[state,fill=black,scale=0.3] (f1) at (\statespace,-\vertspace) {};
				\node[state,fill=black,scale=0.3] (f2) at (2*\statespace,-\vertspace) {};
				
				\node[rectangle,fill=black,scale=0.4] (choice) at ($(i2)!0.5!(f2)$) {b};

				\path (i0) edge[program] (f0);
				\draw [rounded corners=4mm,program] (i1) |- ++(-\statespace,-\vertspace/2) -- (f0);
				\path (i2) edge[program] (f2);
				\draw [rounded corners=4mm,program] (choice) |- ++(-\statespace,-0.4) -- (f1);

				\node (pre) [line width=1.5pt,inner sep=6pt, dashed,rounded corners=0.1cm,fit = {(i1)(i2)}, draw, color=\colorpre] {};
				\node (asp) [line width=1.5pt,inner sep=6pt, dashed,rounded corners=0.1cm,fit = {(f0)(f1)(f2)}, draw, color=\colorpost, label=left:{\color{\colorpost}$\aspsymbol$}] {};
			\end{tikzpicture}
			\qquad
			\begin{tikzpicture}

				\node[state,fill=black,scale=0.3] (i0) at (0,0) {};
				\node[state,fill=black,scale=0.3] (i1) at (\statespace,0) {};
				\node[state,fill=black,scale=0.3] (i2) at (2*\statespace,0) {};

				\node[state,fill=black,scale=0.3] (f0) at (0,-\vertspace) {};
				\node[state,fill=black,scale=0.3] (f1) at (\statespace,-\vertspace) {};
				\node[state,fill=black,scale=0.3] (f2) at (2*\statespace,-\vertspace) {};
				
				\node[rectangle,fill=black,scale=0.4] (choice) at ($(i2)!0.5!(f2)$) {b};

				\path (i0) edge[program] (f0);
				\draw [rounded corners=4mm,program] (i1) |- ++(-\statespace,-\vertspace/2) -- (f0);
				\path (i2) edge[program] (f2);
				\draw [rounded corners=4mm,program] (choice) |- ++(-\statespace,-0.4) -- (f1);

				\node (pre) [line width=1.5pt,inner sep=6pt, dashed,rounded corners=0.1cm,fit = {(i1)(i2)}, draw, color=\colorpre] {};
				\node (dsp) [line width=1.5pt,inner sep=6pt, dashed,rounded corners=0.1cm,fit = {(f1)(f2)}, draw, color=\colorpost, label=left:{\color{\colorpost}$\dspsymbol$}] {};
				\node (space) [line width=1.5pt,inner sep=6pt, dashed,rounded corners=0.1cm,fit = {(f0)}, label=left:{\phantom{$\dwpsymbol$i}}] {};
			\end{tikzpicture}
	\end{adjustbox}%
	\caption{%
		Duality of angelic and demonic weakest pre versus angelic and demonic strongest post, demonstrated on four copies of the same program.
		Preconditions are circled in dashed green, postconditions in dashed blue.
		The rightmost initial state \emph{can} terminate in the postcondition, but may also terminate outside.
		Therefore, it is included in the angelic but not in the demonic weakest precondition.
		Dually, the leftmost final state is reachable from the precondition but also from outside.
		Therefore, it is included in the angelic but not in the demonic strongest postcondition.%
	}%
	\label{fig:adwpsp_symm}%
\end{figure}
Whereas nondeterminism in weakest preconditions arises from explicit nondeterministic branching in the program, the nondeterminism relevant for strongest postconditions arises from \emph{confluence}, i.e.\ when multiple initial states lead to the same final state, even for deterministic computations.
Consider for instance the fully deterministic program $\ASSIGN{x}{2}$.
Then both states $\sigma(x) = 5$ and $\sigma'(x) = 17$ lead to the same final state $\tau(x) = 2$.
If two different initial states can reach $\tau$, this is somewhat dual to one initial state possibly terminating in two final states.
%Notice that $\ASSIGN{x}{2}$ is a fully deterministic program, which still exhibits this form of backward nondeterminism.
For an illustration, see \Cref{fig:adwpsp_symm}.

When deciding whether a final state $\tau$ is in the strongest postcondition of some precondition~$\pre$, we now need to choose whether we require \emph{all} initial states that can reach $\tau$ to satisfy $\pre$ or if it suffices if there \emph{exists} such a state.
Following backward terminology, we refer to the former as \emph{demonic} and the latter as \emph{angelic} resolution of nondeterminism.
We have previously defined \emph{angelic} strongest (non-liberal) and \emph{demonic} strongest liberal postconditions.
Consequently, we will now define the missing \emph{demonic} strongest (non-liberal) and \emph{angelic} strongest liberal postconditions.%
\begin{definition}[Demonic Strongest Postcondition Transformers]
	\label{def:dsp}
    Given a program $\program$ and a precondition $\pre$, the \emph{\underline{demonic} strongest postcondition} is defined as
	\[
		\dsp{\program}{\pre} \eeq \mylambda{\tau}
		\begin{cases}
			{\displaystyle \bigwedge\limits_{\sigma \in \seminv{\program}{\tau}} \pre(\sigma),}  & \text{ if }\  \seminv{\program}{\tau} \neq \emptyset\\
			\false, & \text{ otherwise .}
		\end{cases}
		\tag*{\raisebox{-1.3em}{\qedtriangle}}
	\]
\end{definition}

The demonic strongest post transformer maps a precondition $\pre$ to the set of states that are exclusively reachable by the initial states contained in $\pre$.
Aligning with the intuition above, this is a stronger requirement than for the angelic strongest post, and demonic in the sense that \emph{all} paths leading to the final state in question have to originate in the given precondition.

\begin{definition}[Angelic Strongest Liberal Postcondition Transformers]
	\label{def:aslp}
    Given a program $\program$ and a precondition $\pre$, the \emph{\underline{an}g\underline{elic} strongest liberal postcondition} is defined as
	\[
		\aslp{\program}{\pre} \eeq \mylambda{\tau}
		\begin{cases}
			{\displaystyle \bigvee\limits_{\sigma \in \seminv{\program}{\tau}} \pre(\sigma),}  & \text{ if }\ \seminv{\program}{\tau} \neq \emptyset \\
			\true, & \text{ otherwise .}
		\end{cases}
		\tag*{\raisebox{-1.3em}{\qedtriangle}}
	\]
\end{definition}

The angelic strongest liberal post maps a precondition $\pre$ to the set of states that either are reachable from $\pre$ or unreachable.
Therefore, this is indeed a liberal extension of the angelic strongest post calculus as it accepts unreachable states.

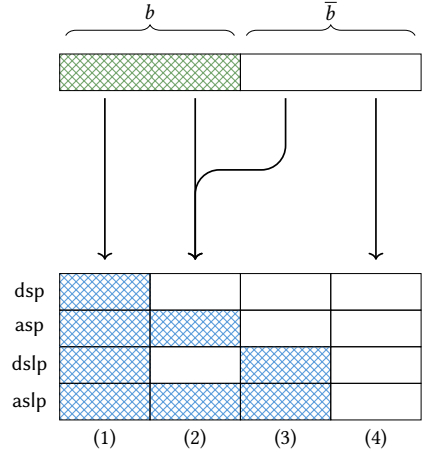
\begin{wrapfigure}[19]{R}{.4\textwidth}%
%	\vspace{-1.9\intextsep}%
	\begin{adjustbox}{max width=.399\textwidth}
	\begin{tikzpicture}
		% Define the width and height of the blocks
		\def\topBlockHeight{0.6}
		\def\bottomBlockHeight{0.6}
		\def\topBlockWidth{3}
		\def\bottomBlockWidth{1.5}
		\def\blockSpacing{0}
		\def\blockSpacingVert{2.8}
		\def\spaceTransformers{0.2}

		% Draw the bottom block split into 4 parts
		% \foreach \i in {0,1,2,3} {
		% 	\draw[] (\i * \bottomBlockWidth + \i * \blockSpacing, 0) rectangle (\i * \bottomBlockWidth + \i * \blockSpacing + \bottomBlockWidth, \bottomBlockHeight);
		% }

		% Draw the top block split into 2 parts
		\draw[pattern=crosshatch,pattern color=mygreen!70] (0, \blockSpacingVert) rectangle (\topBlockWidth, \blockSpacingVert +\bottomBlockHeight);
		\draw[] (\topBlockWidth + \blockSpacing, \blockSpacingVert) rectangle (\topBlockWidth + \blockSpacing + \topBlockWidth, \blockSpacingVert +\bottomBlockHeight);

		% Add labels to the bottom block parts
		\foreach \i in {0, 1,2,3} {
			\node (bot\i) at (\i * \bottomBlockWidth + \i * \blockSpacing + \bottomBlockWidth/2, \bottomBlockHeight/2 - 0.4) {};
		}

		\draw [decorate,decoration={brace,amplitude=5pt,raise=2ex}]
		(0.1,\blockSpacingVert +\bottomBlockHeight) -- (\topBlockWidth-0.1,\blockSpacingVert +\bottomBlockHeight) node[midway,yshift=2em]{$\pre$};
		\draw [decorate,decoration={brace,amplitude=5pt,raise=2ex}]
		(\topBlockWidth+0.1,\blockSpacingVert +\bottomBlockHeight) -- (2*\topBlockWidth-0.1,\blockSpacingVert +\bottomBlockHeight) node[midway,yshift=2em]{$\negate{\pre}$};

		% Add extra labels for nicer arrows
		\foreach \i in {0, 1, 2,3} {
			\node (topref\i) at (\i * \bottomBlockWidth + \i * \blockSpacing + \bottomBlockWidth/2, \blockSpacingVert) {};
		}

		% add program paths
		\path (topref0) edge[program] (bot0);
		\path (topref1) edge[program] (bot1);
		% \path (topref2) edge[program] (bot1);
		\draw [rounded corners=4mm,program] (topref2) |- ++(-1.5,-1.3) -- (bot1);
		\path (topref3) edge[program] (bot3);

		% add transformers
		% dsp
		\def\j{0}
		\def\i{0}
		\draw[pattern=crosshatch,pattern color=DodgerBlue3!70] (\i * \bottomBlockWidth + \i * \blockSpacing, - \spaceTransformers - \j * \bottomBlockHeight) rectangle (\i * \bottomBlockWidth + \i * \blockSpacing + \bottomBlockWidth, -\spaceTransformers -\bottomBlockHeight - \j *\bottomBlockHeight);%  node[midway] {\textsf{X}};
		\def\i{1}
		\draw[] (\i * \bottomBlockWidth + \i * \blockSpacing, - \spaceTransformers - \j * \bottomBlockHeight) rectangle (\i * \bottomBlockWidth + \i * \blockSpacing + \bottomBlockWidth, -\spaceTransformers -\bottomBlockHeight - \j *\bottomBlockHeight);
		\def\i{2}
		\draw[] (\i * \bottomBlockWidth + \i * \blockSpacing, - \spaceTransformers - \j * \bottomBlockHeight) rectangle (\i * \bottomBlockWidth + \i * \blockSpacing + \bottomBlockWidth, -\spaceTransformers -\bottomBlockHeight - \j *\bottomBlockHeight);
		\def\i{3}
		\draw[] (\i * \bottomBlockWidth + \i * \blockSpacing, - \spaceTransformers - \j * \bottomBlockHeight) rectangle (\i * \bottomBlockWidth + \i * \blockSpacing + \bottomBlockWidth, -\spaceTransformers -\bottomBlockHeight - \j *\bottomBlockHeight);

		% asp
		\def\j{1}
		\def\i{0}
		\draw[pattern=crosshatch,pattern color=DodgerBlue3!70] (\i * \bottomBlockWidth + \i * \blockSpacing, - \spaceTransformers - \j * \bottomBlockHeight) rectangle (\i * \bottomBlockWidth + \i * \blockSpacing + \bottomBlockWidth, -\spaceTransformers -\bottomBlockHeight - \j *\bottomBlockHeight);%  node[midway] {\textsf{X}};
		\def\i{1}
		\draw[pattern=crosshatch,pattern color=DodgerBlue3!70] (\i * \bottomBlockWidth + \i * \blockSpacing, - \spaceTransformers - \j * \bottomBlockHeight) rectangle (\i * \bottomBlockWidth + \i * \blockSpacing + \bottomBlockWidth, -\spaceTransformers -\bottomBlockHeight - \j *\bottomBlockHeight);%  node[midway] {\textsf{X}};
		\def\i{2}
		\draw[] (\i * \bottomBlockWidth + \i * \blockSpacing, - \spaceTransformers - \j * \bottomBlockHeight) rectangle (\i * \bottomBlockWidth + \i * \blockSpacing + \bottomBlockWidth, -\spaceTransformers -\bottomBlockHeight - \j *\bottomBlockHeight);
		\def\i{3}
		\draw[] (\i * \bottomBlockWidth + \i * \blockSpacing, - \spaceTransformers - \j * \bottomBlockHeight) rectangle (\i * \bottomBlockWidth + \i * \blockSpacing + \bottomBlockWidth, -\spaceTransformers -\bottomBlockHeight - \j *\bottomBlockHeight);
		
		% dslp
		\def\j{2}
		\def\i{0}
		\draw[pattern=crosshatch,pattern color=DodgerBlue3!70] (\i * \bottomBlockWidth + \i * \blockSpacing, - \spaceTransformers - \j * \bottomBlockHeight) rectangle (\i * \bottomBlockWidth + \i * \blockSpacing + \bottomBlockWidth, -\spaceTransformers -\bottomBlockHeight - \j *\bottomBlockHeight);%  node[midway] {\textsf{X}};
		\def\i{1}
		\draw[] (\i * \bottomBlockWidth + \i * \blockSpacing, - \spaceTransformers - \j * \bottomBlockHeight) rectangle (\i * \bottomBlockWidth + \i * \blockSpacing + \bottomBlockWidth, -\spaceTransformers -\bottomBlockHeight - \j *\bottomBlockHeight);
		\def\i{2}
		\draw[pattern=crosshatch,pattern color=DodgerBlue3!70] (\i * \bottomBlockWidth + \i * \blockSpacing, - \spaceTransformers - \j * \bottomBlockHeight) rectangle (\i * \bottomBlockWidth + \i * \blockSpacing + \bottomBlockWidth, -\spaceTransformers -\bottomBlockHeight - \j *\bottomBlockHeight);%  node[midway] {\textsf{X}};
		\def\i{3}
		\draw[] (\i * \bottomBlockWidth + \i * \blockSpacing, - \spaceTransformers - \j * \bottomBlockHeight) rectangle (\i * \bottomBlockWidth + \i * \blockSpacing + \bottomBlockWidth, -\spaceTransformers -\bottomBlockHeight - \j *\bottomBlockHeight);

		% aslp
		\def\j{3}
		\def\i{0}
		\draw[pattern=crosshatch,pattern color=DodgerBlue3!70] (\i * \bottomBlockWidth + \i * \blockSpacing, - \spaceTransformers - \j * \bottomBlockHeight) rectangle (\i * \bottomBlockWidth + \i * \blockSpacing + \bottomBlockWidth, -\spaceTransformers -\bottomBlockHeight - \j *\bottomBlockHeight);%  node[midway] {\textsf{X}};
		\def\i{1}
		\draw[pattern=crosshatch,pattern color=DodgerBlue3!70] (\i * \bottomBlockWidth + \i * \blockSpacing, - \spaceTransformers - \j * \bottomBlockHeight) rectangle (\i * \bottomBlockWidth + \i * \blockSpacing + \bottomBlockWidth, -\spaceTransformers -\bottomBlockHeight - \j *\bottomBlockHeight) ;% node[midway] {\textsf{X}};
		\def\i{2}
		\draw[pattern=crosshatch,pattern color=DodgerBlue3!70] (\i * \bottomBlockWidth + \i * \blockSpacing, - \spaceTransformers - \j * \bottomBlockHeight) rectangle (\i * \bottomBlockWidth + \i * \blockSpacing + \bottomBlockWidth, -\spaceTransformers -\bottomBlockHeight - \j *\bottomBlockHeight);%  node[midway] {\textsf{X}};
		\def\i{3}
		\draw[] (\i * \bottomBlockWidth + \i * \blockSpacing, - \spaceTransformers - \j * \bottomBlockHeight) rectangle (\i * \bottomBlockWidth + \i * \blockSpacing + \bottomBlockWidth, -\spaceTransformers -\bottomBlockHeight - \j *\bottomBlockHeight);

		% numbers
		\def\j{4}
		\def\i{0}
		\draw[draw=none] (\i * \bottomBlockWidth + \i * \blockSpacing, - \spaceTransformers - \j * \bottomBlockHeight) rectangle (\i * \bottomBlockWidth + \i * \blockSpacing + \bottomBlockWidth, -\spaceTransformers -\bottomBlockHeight - \j *\bottomBlockHeight) node[midway] {(1)};
		\def\i{1}
		\draw[draw=none] (\i * \bottomBlockWidth + \i * \blockSpacing, - \spaceTransformers - \j * \bottomBlockHeight) rectangle (\i * \bottomBlockWidth + \i * \blockSpacing + \bottomBlockWidth, -\spaceTransformers -\bottomBlockHeight - \j *\bottomBlockHeight) node[midway] {(2)};
		\def\i{2}
		\draw[draw=none] (\i * \bottomBlockWidth + \i * \blockSpacing, - \spaceTransformers - \j * \bottomBlockHeight) rectangle (\i * \bottomBlockWidth + \i * \blockSpacing + \bottomBlockWidth, -\spaceTransformers -\bottomBlockHeight - \j *\bottomBlockHeight) node[midway] {(3)};
		\def\i{3}
		\draw[draw=none] (\i * \bottomBlockWidth + \i * \blockSpacing, - \spaceTransformers - \j * \bottomBlockHeight) rectangle (\i * \bottomBlockWidth + \i * \blockSpacing + \bottomBlockWidth, -\spaceTransformers -\bottomBlockHeight - \j *\bottomBlockHeight) node[midway] {(4)};

		% Add labels to the transformer parts
		\node (dsp) at (-0.5, -\spaceTransformers - \topBlockHeight/2 ) {$\dspsymbol$};
		\node (asp) at (-0.5, -\spaceTransformers - \topBlockHeight/2 - 1*\topBlockHeight) {$\aspsymbol$};
		\node (dslp) at (-0.5, -\spaceTransformers - \topBlockHeight/2 - 2*\topBlockHeight) {$\dslpsymbol$};
		\node (aslp) at (-0.5, -\spaceTransformers - \topBlockHeight/2 - 3*\topBlockHeight) {$\aslpsymbol$};

	\end{tikzpicture}
	\end{adjustbox}
	\caption{
		An illustration of the different \spsymbol-style transformers.
		The upper part depicts all possible program executions for a fixed final state and a precondition $\pre$.
		Below, the colored boxed indicate which sets of final states are included in which transformers.
	}
	\label{fig:sp}
\end{wrapfigure}%
\subsubsection{Anatomy of Strongest Postcondition Transformers}

Dual to the coreachability classes of initial states we considered for weakest preconditions, we will now consider \emph{reachability classes} of \emph{final} states for strongest postconditions.
Given a final state $\tau$, a program $\program$, and a precondition $\pre$, we can make out two dimensions regarding what was the case before $\program$ reached $\tau$:%
\begin{enumerate}
	\item[(sd1)]
		$\program$ could have been started in $\pre$ or not,
		
	\item[(sd2)]
		$\program$ could have been started in $\negate{\pre}$ or not,
\end{enumerate}%
%
%\todoin{something about no backward divergence}
These dimensions span a space of $2^2 = 4$ different types of behaviors.
We call the behavioral classes in that space \emph{reachability classes}, since they speak about whether final states are reachable from $\pre$ (or~$\negate{\pre}$).
For example, there is a reachability class containing all final states that are reachable by an execution of $\program$ both from $\pre$ and from $\negate{\pre}$.
The four reachability classes are:
\begin{enumerate}
	\item $\program$ was started in $\pre$.
	\item $\program$ could have been started in $\pre$ or in $\negate{\pre}$.
	\item The unreachable states.
	\item $\program$ was started in $\negate{\pre}$.
\end{enumerate}%
These reachability classes fully partition the state space.

We note that a final state can either be unreachable or not, but never both.
This is in contrast to coreachability where computations from an initial state can very well both diverge and terminate.

A postcondition transformer can now opt to include a class ($i$) in its result or not.
The four natural transformers we described above are illustrated in \Cref{fig:sp}.
Along a \emph{row}, blue shaded boxes indicate which classes are included in the respective strongest postcondition transformer with respect to precondition $\pre$ (in green).
For instance, $\asp{\program}{\post}$ includes classes (1) and (2).

Assuming that the above reachability classes are meaningful, this makes for $2^4 = 16$ possible strongest postcondition transformers of which some might not even depend on the precondition $\pre$ (e.g.\ include only (3)), others might be trivial (e.g.\ include none, include all).
Some are contrapositives of each other.
For example,\ $\aspsymbol$ and $\dslpsymbol$ are \emph{contrapositive} in the sense that%
\begin{align*}
	\asp{\program}{\post} \eeq \negate{\dslp{\program}{\negate{c}}} \qqand \dslp{\program}{\post} \eeq \negate{\asp{\program}{\negate{c}}}
\end{align*}%
This contrapositivity relationship also holds for $\dspsymbol$ and $\aslpsymbol$.
The contrapositivities can also be rediscovered graphically in \Cref{fig:wp} analogously to how this was the case for $\wpsymbol$ transformers.

\subsubsection{Inductive Rules for Strongest Postcondition Transformers}
The transformers $\aspsymbol$ and $\dslpsymbol$ can be defined by induction on the program structure (see \Cref{ssec:sp-rules,ssec:slp-rules} for the concrete rules). %\cite[Appendix E]{verscht2025taxonomy} 
For the novel transformers $\dspsymbol$ and $\aslpsymbol$, we cannot quite give an inductive set of rules.
To see why, consider the nondeterministic program $\NDCHOICEC$.
Recall that the \emph{angelic} strongest postcondition of $\pre$ is the set of states that are reachable from $\pre$.
The set of states reachable by $\NDCHOICEC$ is the union of the states reachable from $\program_1$ and from $\program_2$, i.e.\ we get $\asp{\NDCHOICEC}{\pre} = \asp{\program_1}{\pre} \cup \asp{\program_2}{\pre}$.
Similarly, if we want to compute the demonic strongest liberal post, which contains the set of states unreachable or exclusively reachable from $\pre$, we take the intersection of the results for both subprograms and get $\dslp{\NDCHOICEC}{\pre} = \dslp{\program_1}{\pre} \cup \dslp{\program_2}{\pre}$.

The demonic strongest post of $\pre$ contains all states that are exclusively reachable from $\pre$.
For $\NDCHOICEC$, a final state $\tau$ is contained in this set if it fulfills one of the following three cases:
\begin{enumerate}
	\item $\tau$ is exclusively reachable from $\pre$ by executing $\program_1$ and unreachable by executing $\program_2$.
	\item $\tau$ is exclusively reachable from $\pre$ by executing $\program_2$ and unreachable by executing $\program_1$.
	\item $\tau$ is exclusively reachable from $\pre$ by executing $\program_1$ and exclusively reachable by executing $\program_2$.
\end{enumerate}
If we restrict to using only the demonic strongest post for the subprograms, the only case we can represent is (3) by $\dsp{\program_1}{\pre} \cap \dsp{\program_2}{\pre}$.
For (1) and (2), we need to reason about unreachability, which is not possible using $\dsp{\program_1}{\pre}$ or $\dsp{\program_2}{\pre}$. 
Similar problems arise for $\aslpsymbol$.

As a silver lining, we can characterize $\dspsymbol$ as a combination of other transformers:
We have that
\[
	\dsp{\program}{\pre} = \asp{\program}{\pre} \cap \dslp{\program}{\pre},
\]
%\todoin{This can be easily read off from Figure sonstwas. (steht unten auch schon)}
as a final state $\tau$ must be exclusively reachable from $\pre$ in order to be contained in $\dsp{\program}{\pre}$.
If $\tau \in \dslp{\program}{\pre}$, we know that $\tau$ is either (1) unreachable or (2) exclusively reachable from $\pre$.
If additionally $\tau \in \asp{\program}{\pre}$, we know that $\tau$ is definitely reachable from $\pre$, which excludes case (1).
This can also be seen in \Cref{fig:sp}:
The intersection of $\asp{\program}{\pre}$ and $\dslp{\program}{\pre}$ only contains the first reachability class, which is equivalent to $\dsp{\program}{\pre}$.
Similarly, we have that
\[
	\aslp{\program}{\pre} = \asp{\program}{\pre} \cup \dslp{\program}{\pre}.
\]
Therefore, although we cannot provide an inductive set of rules directly, we can make use of the existing rules and compute $\dspsymbol$ (resp.\ $\aslpsymbol$) with two inductive computations.

\subsection{Backward vs. Forward Analysis}
\label{ssec:backward-forward}

The characterization of the novel transformers \dspsymbol\ and \aslpsymbol\ goes via union and intersection of existing transformers,
raising the question whether we can do the same for the weakest pre calculi, i.e.\ whether we have that
\[
	\dwp{\program}{\post} \morespace{\stackrel{?}{=}}  \awp{\program}{\post} \cap \dwlp{\program}{\post}
%\]
\qqand
%\[
	\awlp{\program}{\post} \morespace{\stackrel{?}{=}} \awp{\program}{\post} \cup \dwlp{\program}{\post}.
\]
This is -- perhaps somewhat surprisingly -- not the case.
To see why, recall \Cref{fig:wp}.
The set $\dwlp{\program}{\post}$ contains states from which computation always either terminates in $\post$ or diverges, and $\awp{\program}{\post}$ contains all states that can reach $\post$.
In \Cref{fig:wp}, their intersection corresponds to coreachability classes (1) and (2).
This is not equivalent to $\dwp{\program}{\post}$, which excludes class (2), containing states from which computation either diverges or terminates in $\post$. 
% A state from which computation either diverges or terminates in $\post$ (class (2)) is contained in the intersection of \dwlpsymbol\ and \awpsymbol, but not in $\dwp{\program}{\post}$.
% Their intersection thus contains all states from which computation, if it terminates, always does so in $\post$, and there must exist at least one terminating execution.

The union of $\awp{\program}{\post}$ and $\dwlp{\program}{\post}$ contains states from which computation either can reach $\post$ or computation always diverges.
A state from which computation either diverges or terminates outside of $\post$ (class no.\ 6) is not contained in this set, but in $\awlp{\program}{\post}$.

Note that both counterexamples are concerned with states with \emph{branching divergence}, i.e.\ a computation that nondeterministically either diverges or terminates in some states.
In fact, if we exclude such behavior, the equations above hold.
% Therefore, the conjunction and disjunction above can be thought of a transformer in between \awpsymbol\ and \dwlpsymbol\, close to \dwpsymbol\ or respectively \awlpsymbol.
% \lv{ich habe den einen satz hier rausgenommen, irgendwie ist der unnötig und die graphische intuition kommt ja später.}

%
\begin{wrapfigure}[12]{r}{.4\textwidth}%
%	\begin{subfigure}{0.4\textwidth}  
		\vspace{-0.3\intextsep}%
		\begin{center}
			\begin{tikzpicture}[->,
				node distance=1cm and 3cm,
				thick,
				mystate/.style = {circle,inner sep=3pt,draw,font=\small}
				]

				\node[state,fill=black,scale=0.3] (I4) {\phantom{1}};
				\node[state,fill=black,scale=0.3] (F4) [right =of I4] {\phantom{1}};

				\node[rectangle,fill=black,scale=0.4] (choice) at ($(I4)!0.5!(F4)$) {b};
				\node (spiral) at ($(I4)!0.5!(F4)$) {};
				
				\path (I4) edge node {} (F4);
				\path (I4) edge [-] node {} (spiral);
				
				\coordinate (spiralstart) at ($(spiral)+(0,-0.45)$);
				\draw let \p1=(I4) in
				[scale=0.25,domain=0:30,variable=\t,smooth,samples=75,-,shift={(spiralstart)}] plot  ({\t r}: {-0.002*\t*\t});
			\end{tikzpicture}
%		\end{center}
%	\end{subfigure}
	%
	
	\bigskip\smallskip
	\hrule
	\bigskip\bigskip
	
%	\begin{subfigure}{0.4\textwidth}
%	\begin{center}
		\begin{tikzpicture}[node distance=1cm and 3cm,
			thick,
			mystate/.style = {circle,inner sep=3pt,draw,font=\small}
			]

			\node[state,fill=black,scale=0.3] (I4) {\phantom{1}};
			\node[state,fill=black,scale=0.3] (F4) [right =of I4] {\phantom{1}};

			\node[rectangle,fill=black,scale=0.4] (choice) at ($(I4)!0.5!(F4)$) {b};
			\node (spiral) at ($(I4)!0.5!(F4)$) {};
			
			\path (I4) edge [->] node {} (F4);
			\path (I4) edge [-] node {} (spiral);

			\coordinate (spiralstart) at ($(spiral)+(0,-0.45)$); % y coordinate is -0.002*30*30*0.25 (last plotted point times scale)
			\draw let \p1=(I4) in
				[scale=0.25,xscale=-1,domain=0:30,variable=\t,smooth,samples=75,-,shift={(spiralstart)}] plot  ({\t r}: {-0.002*\t*\t});
						
			\draw[purple!70!red, line width=0.6mm]
				%(0,0) -- (-5.31,-2.8)
				(1.15,-0.15) -- (2.1,-0.7);
		\end{tikzpicture}
	\end{center}
%	\end{subfigure}
	\caption{Branching divergence versus confluence of unreachability.}
	\label{fig:branching}
\end{wrapfigure}
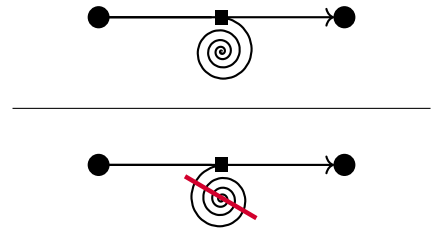%
This also gives intuition as to why the respective equations hold for \spsymbol\ transformers.
As previously mentioned, the dual concept of divergence is unreachability.
Illustrated in \Cref{fig:branching}, branching divergence represents the division of a computation into one that reaches a final state and one that does not.
Therefore, the dual should reflect the confluence of a computation that is reachable from an initial state and one that is not, as shown in the lower half of \Cref{fig:branching}.
But, this is paradox: A final state can never be unreachable and reachable at the same time.
We refer to this observation as the \emph{absence of unreachability confluence}.% or equivalently the \emph{presence of branching divergence}.
\begin{observation}[Absence of Unreachability Confluence]
	\label{obs:branching}
	For a program $\program$, an initial state $\sigma$, and a final state $\tau$, we have that%
	\begin{align*}
		\tau \text{ is unreachable by any computation of } \program &\quad\quad\textnormal{iff}\quad\quad \seminv{\program}{\tau} = \emptyset~, \qquad \textnormal{but}\\ % \quad \iff \quad \tau \text{ is reachable and not unreachable} \qquad \overline{\infty} \not \in \seminv{\program}{\tau}
%	\intertext{but}
		\program \text{ can diverge on input } \sigma  &\quad\xcancel{\quad\textnormal{iff}\quad}\quad \sem{\program}{\sigma} = \emptyset~. %\quad \iff \quad \text{from } \sigma \text{, computation terminates but can still diverge} \qquad \infty \not \in \sem{\program}{\sigma},
		\tag*{\qedtriangle}
	\end{align*}%
	%
	%where $\overline{\infty}$ symbolizes unreachability and $\infty$ symbolizes divergence.
\end{observation}%
\noindent%
\Cref{obs:branching} in particular enables us to characterize unreachability in a relational setting by testing equality to the empty set.
This is not possible for divergence, weakening the previously discussed duality of divergence and unreachability (see e.g.\ \Cref{fig:adwpsp_symm}).
We will see how this intrinsic difference between forward and backward analysis affects program logics in \Cref{sec:kat,sec:taxonomy-rev}.
Another direct consequence of \Cref{obs:branching} is that the demonic weakest pre is not truly dual to demonic strongest post, after all.
For an illustration, see \Cref{sec:illustration}. % \cite[Appendix A]{verscht2025taxonomy}.

\subsection{Relating Predicate Transformers}
\label{ssec:relations}

\begin{table}[t]
	\begin{center}
		\begin{adjustbox}{max width=\textwidth}
			\begin{tabular}{r | l}
				transformer & captured states (precondition $\pre$, postcondition $\post$) \\
				\hline\hline
				angelic weakest pre ($\awpsymbol$) & initial states that can reach $\post$ \\	
				angelic weakest liberal pre ($\awlpsymbol$) & initial states that can reach $\post$ or diverge \\	
				demonic weakest pre ($\dwpsymbol$) & initial states that can only reach $\post$ \\
				demonic weakest liberal pre ($\dwlpsymbol$) & initial states that can only reach $\post$ or diverge \\
				\hline
				angelic strongest post ($\aspsymbol$) & final states that are reachable from $\pre$ \\	
				angelic strongest liberal post ($\aslpsymbol$) & final states that are reachable from $\pre$ or unreachable \\	
				demonic strongest post ($\dspsymbol$) & final states that are exclusively reachable from $\pre$ \\
				demonic strongest liberal post ($\dslpsymbol$) & final states that are exclusively reachable from $\pre$ or unreachable \\
			\end{tabular}
		\end{adjustbox}
	\end{center}
	\caption{
		An overview of the intuition for all presented predicate transformers.
	}
	\label{tab:transformer-overview}
\end{table}

\Cref{tab:transformer-overview} summarizes the sets of states captured by the eight predicate transformers defined in the previous section.
% Notably, the strongest post transformers are described as final states \emph{reachable from} some states, while the weakest pre transformers contain initial states that \emph{can reach}, or are \emph{coreachable from}, some states.
As indicated before, the transformers are closely related.
The remainder of this section is dedicated to formalizing these relations.

First, we observe that we can order \wpsymbol\ and \spsymbol\ transformers by set inclusion.
This can also be seen in \Cref{fig:wp,fig:sp}, respectively.

\begin{theorem}[Ordering on Predicate Transformers]
	\label{theo:order-trans}
	For all programs $p$ and predicates $\post, \pre$, the following inclusions hold:%
	\begin{center}
		\begin{tikzpicture}[node distance = .7cm]
			\node (awp) []  {$\awp{\program}{\post}$};
			\node (awlp) [right =of awp]  {$\awlp{\program}{\post}$};
			\node (dwp) [below =of awp]  {$\dwp{\program}{\post}$};
			\node (dwlp) [right =of dwp]  {$\dwlp{\program}{\post}$};
			
			\path (awp) edge[draw=none] node [sloped, auto=false, allow upside down] {$\subseteq$}  (awlp);
			\path (dwp) edge[draw=none] node [sloped, auto=false, allow upside down] {$\subseteq$} (awp);
			\path (dwlp) edge[draw=none] node [sloped, auto=false, allow upside down] {$\subseteq$} (awlp);
			\path (dwp) edge[draw=none] node [sloped, auto=false, allow upside down] {$\subseteq$} (dwlp);
		\end{tikzpicture}
		\qquad\qquad
		\begin{tikzpicture}[node distance = .7cm]
			\node (asp) []  {$\asp{\program}{\pre}$};
			\node (aslp) [right =of asp]  {$\aslp{\program}{\pre}$};
			\node (dsp) [below =of asp]  {$\dsp{\program}{\pre}$};
			\node (dslp) [right =of dsp]  {$\dslp{\program}{\pre}$};
		
			\path (asp) edge[draw=none] node [sloped, auto=false, allow upside down] {$\subseteq$}  (aslp);
			\path (dsp) edge[draw=none] node [sloped, auto=false, allow upside down] {$\subseteq$} (asp);
			\path (dslp) edge[draw=none] node [sloped, auto=false, allow upside down] {$\subseteq$} (aslp);
			\path (dsp) edge[draw=none] node [sloped, auto=false, allow upside down] {$\subseteq$} (dslp);
		\end{tikzpicture}%
	\end{center}
\end{theorem}%
\begin{proof}
	This follows directly from \Cref{def:wp,def:wlp,def:asp,def:dslp,def:dsp,def:aslp}.
\end{proof}%
\noindent%
It was hinted at before that the transformers are in contrapositive relation.
% Another crucial property is that the transformers complement each other.
% For instance, the demonic weakest precondition with respect to a postcondition $\post$ includes all states from which computation definitely terminates in $\post$.
% From any other initial states, computation either diverges or terminates outside of $\post$.
% Therefore, the (disjoint) union of $\dwp{\program}{\post}$ and $\awlp{\program}{\neg \post}$ covers the entire state space $\Sigma$.
% Similar complementary relations hold for all transformers.

\begin{theorem}[Contrapositive Transformers]
	\label{theo:complements}
	For all programs $\program$ and predicates $\pre, \post$, we have:%
	\begin{multicols}{2}%
		\begin{enumerate}[label=(\arabic*)]
			\setlength\itemsep{0.3em}
			\item $\awp{\program}{\post} \eeq \negatedbl{\dwlp{\program}{\negate{\post}}}$
			\item $\dwp{\program}{\post} \eeq \negatedbl{\awlp{\program}{\negate{\post}}}$
			\item $\asp{\program}{\pre} \eeq \negatedbl{\dslp{\program}{\negate{\pre}}}$
			\item $\dsp{\program}{\pre} \eeq \negatedbl{\aslp{\program}{\negate{\pre}}}$
		\end{enumerate}%
	\end{multicols}%
	%
	% \begin{multicols}{2}%
	% 	\begin{enumerate}[label=(\arabic*)]
	% 		\item $\dwp{\program}{\post}\; \cupdot\; \awlp{\program}{\negate{\post}} = \Sigma$
	% 		\item $\awp{\program}{\post}\; \cupdot\;  \dwlp{\program}{\negate{\post}} = \Sigma$
	% 		\item $\dsp{\program}{\pre}\; \cupdot\;  \aslp{\program}{\negate{\pre}} = \Sigma$
	% 		\item $\asp{\program}{\pre}\; \cupdot\;  \dslp{\program}{\negate{\pre}} = \Sigma$
	% 	\end{enumerate}%
	% %
	% \end{multicols}%
\end{theorem}
\begin{proof}
	Properties (1) to (3) are folklore knowledge.
	For the proof of (4), see \Cref{ssec:proof-complements}. % \cite[Appendix F.1]{verscht2025taxonomy}.
\end{proof}

% !TEX root = ../main.tex

\section{A Taxonomy of Hoare-Like Program Logics}
\label{sec:taxonomy}

Arguably the best known program logic is Hoare logic~\cite{Hoa69}.
Its central notion are \emph{triples} $\hoare{\pre}{\program}{\post}$ consisting of a program $\program$, a precondition $\pre$, and a postcondition $\post$.
Such triples have also been called \emph{asserted programs} in various literature.
To give meaning to a triple, it must be \emph{exegeted} when a triple is considered \emph{valid} and when it is not.
For the standard partial correctness interpretation of Hoare logic, it is well known that this exegesis can be captured in terms of predicate transformers, namely
\begin{align*}
	\hoare{\pre}{\program}{\post} \textnormal{ is valid for partial corr.} 
	\qqiff 
	\pre \subseteq \dwlp{\program}{\post}
	\qqiff
	\asp{\program}{\pre} \subseteq \post~.
\end{align*}%
Another program logic that somewhat recently (re)gained attention under the name \emph{incorrectness logic}~\cite{OHearn19} features triples that are exegeted through
\begin{align*}
	\hoare{\pre}{\program}{\post} \textnormal{ is valid for incorr.} \qqiff \post \subseteq \asp{\program}{\pre}~.
\end{align*}%
We see above that for some exegeses we \emph{over}approximate the result of a predicate transformer, for others we \emph{under}approximate it.
%Also, we have used both pre- or postcondition transformers as well as angelic and demonic transformers.
In total, we have defined 8 different predicate transformers, which we can each over- and underapproximate, yielding potentially $2 \cdot 8 = 16$ different exegeses of triples and thus program logics.
Some of these will coincide, as we have already seen above for partial correctness.
Others will imply others, yet others will be contrapositives of each other. 
In the following, we give an overview of all 16 possibilities and study their relationships, thus yielding a taxonomy of predicate transformer definable program logics.

\subsection{Program Logics}
\label{ssec:logics}

As described just above, there are 16 possible exegeses of triples $\hoare{\pre}{\program}{\post}$ that can be obtained by over- or underapproximating each of the 8 predicate transformers defined in \Cref{sec:transformer}.
A full overview of these is provided in \Cref{fig:taxonomy}.
The \enquote{colloquial terms} (broadly interpreted) of these exegeses/logics (some more, some less common) are provided immediately below the individual exegeses.
Some of these names are more common than others.
We would like to mention, however, that these names are not necessarily accurately chosen, especially the attribution of being a \emph{correctness} or an \emph{incorrectness} logic\footnote{But we acknowledge that these attributes made sense for the individual originally intended purposes of those logics.}\hspace{-3pt}.\; 
For example, \emph{incorrectness logic} was prior to \cite{OHearn19} introduced by \citet{VK11} under the name \emph{reverse Hoare logic} and was intended for proving that all \emph{good} things can happen -- arguably more of a \emph{correctness} criterion.

Other attributions are being an \emph{over}- or an \emph{under}approximate logic. 
For example, incorrectness logic was attributed underapproximate and partial correctness overapproximate.
But this also is not entirely accurate since the supposedly overapproximate partial correctness can also be exegeted via an underapproximation, namely of $\dwlpsymbol$.
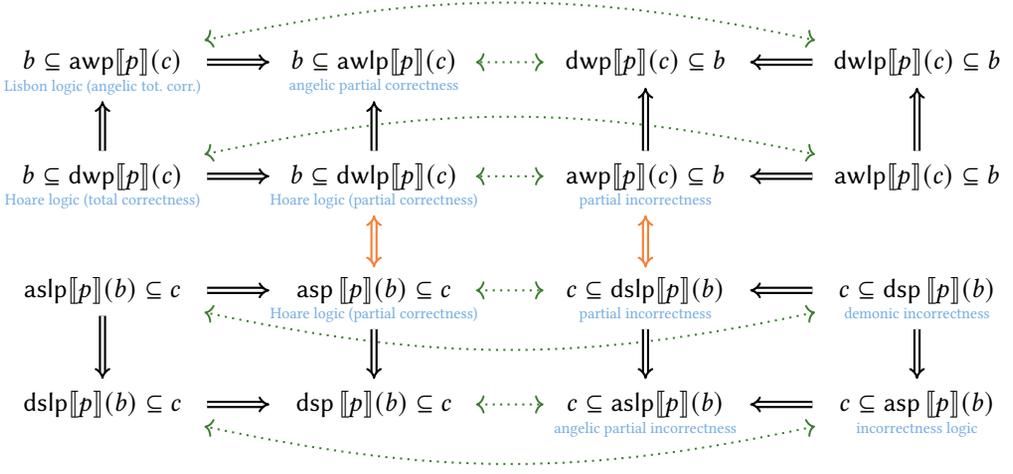
\begin{figure}[t]
\vspace*{-1\intextsep}
\begin{adjustbox}{max width=\textwidth}
	\begin{tikzpicture}[
		mynode/.style={minimum width=2.7cm},node distance=0.9cm and 0.9cm,
		mylabel/.style n args={1}{label={[label distance=-4mm,font=\tiny,text=DodgerBlue3!60,below]#1}}
		]
		\node[mynode, mylabel={Lisbon logic (angelic tot.\ corr.)}](awpLB) {$\pre \subseteq \awp{\program}{\post}$};
		\node[mynode, mylabel={Hoare logic (total correctness)},below=of awpLB](dwpLB) {$\pre \subseteq \dwp{\program}{\post}$};
		\node[mynode, mylabel={angelic partial correctness},right=of awpLB](awlpLB) {$\pre \subseteq \awlp{\program}{\post}$};
		\node[mynode,right=of awlpLB]  (dwpUB) {$\dwp{\program}{\post} \subseteq \pre$};
		\node[mynode,right=of dwpUB]  (dwlpUB) {$\dwlp{\program}{\post} \subseteq \pre$};
		\node[mynode,below=of dwlpUB]  (awlpUB) {$\awlp{\program}{\post} \subseteq \pre$};
		
		\node[mynode, mylabel={Hoare logic (partial correctness)},below=of awlpLB] (dwlpLB) {$\pre \subseteq \dwlp{\program}{\post}$};
		\node[mynode,mylabel={partial incorrectness},right=of dwlpLB] (awpUB) {$\awp{\program}{\post} \subseteq \pre$};
		\node[mynode,mylabel={Hoare logic (partial correctness)}, below=of dwlpLB] (aspUB) {$\asp{\program}{\pre} \subseteq \post$};
		\node[mynode,below=of dwpLB] (aslpUB) {$\aslp{\program}{\pre} \subseteq \post$};
		\node[mynode,below=of aslpUB]  (dslpUB) {$\dslp{\program}{\pre} \subseteq \post$};
		\node[mynode,below=of aspUB] (dspUB) {$\dsp{\program}{\pre} \subseteq \post$};
		
		\node[mynode, mylabel={partial incorrectness},right=of aspUB] (dslpLB) {$\post \subseteq \dslp{\program}{\pre}$};
		\node[mynode, mylabel={demonic incorrectness},right=of dslpLB] (dspLB) {$\post \subseteq \dsp{\program}{\pre}$};
		\node[mynode, mylabel={angelic partial incorrectness},below=of dslpLB] (aslpLB) {$\post \subseteq \aslp{\program}{\pre}$};
		\node[mynode, mylabel={incorrectness logic},below=of dspLB] (aspLB) {$\post \subseteq \asp{\program}{\pre}$};
		
		\draw[contrapos] (awpUB) edge[<->] (dwlpLB);
		\draw[contrapos] (dwpUB) edge[<->] (awlpLB);
		\draw[contrapos,bend left=12] (awpLB) edge[<->] (dwlpUB);
		\draw[contrapos,bend left=12] (dwpLB) edge[<->] (awlpUB);
		
		\draw[contrapos] (aspUB) edge[<->] (dslpLB);
		\draw[contrapos,bend left=12] (aspLB) edge[<->] (dslpUB);
		\draw[contrapos] (dspUB) edge[<->] (aslpLB);
		\draw[contrapos,bend left=12] (dspLB) edge[<->] (aslpUB);
		
		\path (dwlpLB) edge[galoisshort] node {} (aspUB);
		\path (awpUB) edge[galoisshort] node {} (dslpLB);

		% text labels
%			\draw [decorate, ultra thick, decoration = {calligraphic brace, raise=5pt, amplitude=5pt}] (5,-6.3) -- node[pos=0.5] (middletoken1) {} (5,-2.4);
%			\node[mynode, left=of middletoken1,xshift=1cm] (part corr) {\small partial correctness\vphantom{g}};
%			
%			\node[mynode, left= of awpLB,xshift=9cm] (total corr) {\small Lisbon total correctness\vphantom{g}};
%			
%			\node[mynode, right=of aspLB2,xshift=-1.8cm] (total incorr) {\small incorrectness\vphantom{g}};
%			
%			\draw [decorate,ultra thick, decoration = {calligraphic brace, raise=5pt, amplitude=5pt}] (19.3,-2.4) -- node[pos=0.5] (middletoken2) {} (19.3,-6.3);
%			\node[mynode, right=of middletoken2,xshift=-1cm] (part incorr) {\small partial incorrectness\vphantom{g}};

		\path (dwpLB) edge[implicationup] node {} (awpLB);
		\path (dwpLB) edge[implication] node {} (dwlpLB);
		\path (awpLB) edge[implication] node {} (awlpLB);
		\path (dwlpLB) edge[implicationup] node {} (awlpLB);
		
		\path (awpUB) edge[implication] node {} (dwpUB);
		\path (awlpUB) edge[implication] node {} (awpUB);
		\path (awlpUB) edge[implication] node {} (dwlpUB);
		\path (dwlpUB) edge[implication] node {} (dwpUB);
		
		\path (dspLB) edge[implicationdown] node {} (aspLB);
		\path (dspLB) edge[implication] node {} (dslpLB);
		\path (aspLB) edge[implication] node {} (aslpLB);
		\path (dslpLB) edge[implicationdown] node {} (aslpLB);
		
		\path (aspUB) edge[implicationdown] node {} (dspUB);
		\path (aslpUB) edge[implication] node {} (aspUB);
		\path (aslpUB) edge[implication] node {} (dslpUB);
		\path (dslpUB) edge[implication] node {} (dspUB);

	\end{tikzpicture}
\end{adjustbox}
\caption{
    A taxonomy of predicate transformer-based program logics.
    Black arrows are simple implications, green dotted arrows are contrapositive relations, and orange two-sided arrows are Galois connections.
}
\label{fig:taxonomy}
\end{figure}%
Before we next provide intuitions and background on the individual logics, let us first have a closer look at the \emph{structure} of \Cref{fig:taxonomy}.

\paragraph{The Implications} 
If we divide \Cref{fig:taxonomy} into four quadrants (i.e.\ horizontally and vertically in the middle), we obtain four squares of implications. 
For example (top left quadrant), total correctness implies partial correctness and the Lisbon logic exegesis.
The latter two each imply angelic partial correctness.
For each quadrant, the respective implications stem entirely from the ordering of the predicate transformers provided in \Cref{theo:order-trans}.

\paragraph{The Contrapositions}
If we divide \Cref{fig:taxonomy} vertically into two halves, then every exegesis in one half has a mirrored contrapositive in the other half, indicated by green dotted arrows.
For example, partial correctness and partial incorrectness are contrapositive to each other, meaning that%
\begin{align*}
	\hoare{\pre}{\program}{\post} \textnormal{ is valid for part.\ corr.} 
	\qqiff
	\hoare{\negate{\pre}}{\program}{\negate{\post}} \textnormal{ is valid for part.\ incorr.}
\end{align*}
In that sense, an exegesis and its contrapositive are \enquote{equivalent} and one of the two halves could be discarded.
We would argue, however, that -- depending on the proof objective -- it may well be more intuitive to annotate code in non-negated versions so that program proofs remain more understandable.
All contrapositions of exegesis in \Cref{fig:taxonomy} arise from the contrapositivities of the predicate transformers described in \Cref{theo:complements}.
%Proofs can be found in \Cref{ssec:proof-contrapositives}.\lv{proofs kann man sich sparen -> sind draußen}

\paragraph{The Equivalences}
In the very center of \Cref{fig:taxonomy}, we have a square of, respectively, two vertically equivalent exegesis.
For example, partial correctness can be exegeted equivalently through $\pre \subseteq \dwlp{\program}{\post}$ or $\asp{\program}{\pre} \subseteq \post$.
Analogously, partial incorrectness can be exegeted in two equivalent ways.
The partial correctness equivalence stems from the very well-known Galois connection%
\begin{align*}
	\pre \subseteq \dwlp{\program}{\post} &\qqiff \asp{\program}{\pre} \subseteq \post~.
\intertext{%
	The partial incorrectness equivalence stems from the much less well-known Galois connection~\cite{ZhangKaminski22}%
}
%%
%\[
	\awp{\program}{\post} \subseteq \pre  &\qqiff \post \subseteq \dslp{\program}{\pre}~.
\end{align*}
Note that such Galois connections are not only of theoretical interest:
\emph{They} allow to reason either forward or backward through a program, whatever is more feasible.
They even allow to reason in both directions simultaneously:
For example, $\hoare{\pre}{\COMPOSE{\program_1}{\program_2}}{\post}$ is valid for partial correctness if $\asp{\program_1}{\pre} \subseteq \dwlp{\program_2}{\post}$, meaning that we have done backward reasoning on $\program_2$ and forward reasoning on $\program_1$.
Analogous bidirectional reasoning can be performed for partial incorrectness.

In the following, we will discuss those logics of \Cref{fig:taxonomy} that have a colloquial name in more detail.
We will proceed more or less in (partial) order of popularity.
As a running example, we use a program $\plogin$ which realizes some login procedure.
The user must enter a password and is then either granted access or not.

\paragraph{Hoare logic for partial correctness \textnormal{\cite{Hoa69}}}
$\hoare{\pre}{\program}{\post}$ is valid for \emph{(demonic) partial correctness} iff $\pre \subseteq \dwlp{\program}{\post}$ or equivalently $\asp{\program}{\pre} \subseteq \post$.
This is the \emph{standard} notion of partial correctness, stating that all computations of $\program$ started in $\pre$ \emph{must} either diverge or terminate in $\post$.
%As mentioned before, this is equivalent to what we call \emph{demonic partial incorrectness}, i.e.\ $\asp{\program}{\pre} \subseteq \post$ and $\negate{\post} \subseteq \dslp{\program}{\negate{\pre}}$.
%From this perspective, all states in $\negate{\post}$ must be either unreachable or exclusively reachable from $\negate{\pre}$.

For example, validity of the partial correctness Hoare triple $\hoare{\pwdincorrect}{\plogin}{\noaccess}$ expresses the following:
Should the user provide the wrong password, the login will fail by either terminating in a state where the user is denied access, or diverge.
This is a correctness property as it expresses behavior we would expect from a login procedure (except perhaps the divergence).

The contrapositive triple $\awp{\program}{{\post}} \subseteq {\pre}$ was discussed by \citet{Cousot13} under the name \emph{necessary precondition}.
The intuition for this is that $\pre$ necessarily has to hold for computation to terminate in $\post$.
All other computation is guaranteed to either diverge or terminate outside of $\post$.

\paragraph{Hoare logic for total correctness \textnormal{\cite{Hoa69}}}
$\hoare{\pre}{\program}{\post}$ is valid for \emph{(demonic) total correctness} iff $\pre \subseteq \dwp{\program}{\post}$.
This is the \emph{standard} notion of total correctness, stating that all computations of $\program$ started in $\pre$ \emph{must} terminate in $\post$.
This can be classically used to specify correctness properties of the program.
For example, the triple $\hoare{\pwdcorrect}{\plogin}{\access}$ expresses that if a user enters the correct password, access is definitely granted.

\paragraph{Incorrectness logic \textnormal{\cite{OHearn19,VK11}}}
$\hoare{\pre}{\program}{\post}$ is valid for \emph{(angelic) incorrectness} iff $\post \subseteq \asp{\program}{\pre}$.
% or equivalently $\dslp{\program}{\negate{\pre}} \subseteq \negate{\post}$.
This exegesis was popularized by \citet{OHearn19} under the name \emph{incorrectness logic} and ensures that \emph{all} states in $\post$ must be reachable from \emph{some} state in $\pre$.

Why \emph{incorrectness}?
O'Hearn thought of $\post$ as a set of bugs whose (true positive) presence was supposed to be proved.
For example, the triple $\hoare{\true}{\plogin}{\error}$ expresses that the error state is reachable, thus proving the existence of some bug.

As mentioned before, incorrectness logic was described some 8 years earlier by \citet{VK11} under the name \emph{reverse Hoare logic}.
Contrary to O'Hearn, de Vries and Koutavas thought of $\post$ as a set of \emph{desirable} states who should \emph{all} be reachable, thus exemplifying that neither the attribute \emph{incorrectness} nor \emph{correctness} is entirely accurate for this logic.
%We can also prove \enquote{good} properties of programs, for example using $\hoare{\pwdcorrect}{\plogin}{\access}$.
%If this is valid, we know that all states where access is granted are reachable by entering some correct password.
% backwards underapproximation triples: also used by O'Hearn & co

\paragraph{Lisbon logic}
$\hoare{\pre}{\program}{\post}$ is valid for \emph{angelic total correctness} iff $\pre \subseteq \awp{\program}{\post}$, 
% or equivalently $\dwlp{\program}{\negate{\post}} \subseteq \negate{\pre}$.
stating that from {all} states in $\pre$, there must \emph{exist} a computation of $\program$ terminating in~$\post$.
%This is equivalent to saying that all states that from which $\program$ always either diverges or terminates outside of $\post$ are contained in the complement of $\pre$.
%
% Regarding its name, \citet{zilberstein2023outcome} describe that such triples have first been described in \cite{DBLP:conf/RelMiCS/MollerOH21} as \emph{backwards under-approximate triples}, for which recently a proof system was presented \cite{raad2024nontermination}.
% They were discussed as a possible foundation for \emph{incorrectness} reasoning during a meeting in Lisbon, hence \emph{Lisbon logic}.
Regarding its name, \citet{zilberstein2023outcome} describe that such triples have first been described in \cite{DBLP:conf/RelMiCS/MollerOH21} as \emph{backwards under-approximate triples}  and been discussed as a possible foundation for \emph{incorrectness} reasoning during a meeting in Lisbon, hence \emph{Lisbon logic}.
Recently, a proof system for the logic was developed \cite{raad2024nontermination}.
However, underapproximating angelic weakest pre(conditions) has been studied much earlier, for example by \citet{hoare1978some} as \emph{possible correctness} and later by \citet{McIverM05}.
Another name for Lisbon logic is \emph{sufficient incorrectness logic}, due to \citet{ascari2023sufficient}. % and \emph{backwards underapproximating logic} \cite{raad2024nontermination}.
% A proof system for this logic was developed in \cite{raad2024nontermination}.

As an example, consider again the triple $\hoare{\pwdcorrect}{\plogin}{\access}$.
Exegeted as angelic total correctness, this expresses that with a correct password, it is always possible to get access \mbox{--
indeed} a \emph{correctness} property.
On the other hand, if $\hoare{\pwdabcd}{\plogin}{\access}$ is valid for angelic total correctness, the password \enquote{1234} can always result in access.
This is likely \emph{incorrect} behavior.
Again, neither \emph{correctness} nor \emph{incorrectness} seem appropriate attributes.
%Dually, we can use incorrectness logics to express wanted behavior.

\paragraph{Partial incorrectness \textnormal{\cite{ZhangKaminski22}}}
$\hoare{\pre}{\program}{\post}$ is valid for \emph{(demonic) partial incorrectness} iff $\post \subseteq \dslp{\program}{\pre}$, requiring all states in $\post$ to be either unreachable or \emph{exclusively} reachable from $\pre$.
% By exclusive reachability we mean that there is no state in $\negate{\pre}$ that can reach a state in $\post$.
In other words, computation starting in $\negate{\pre}$ may only terminate in $\negate{\post}$, or: $\hoare{\negate{\pre}}{\program}{\negate{\post}}$ is valid for demonic partial correctness.
This is not surprising because of contrapositivity.
%\lv{example?}

\paragraph{Angelic partial correctness}
$\hoare{\pre}{\program}{\post}$ is valid for \emph{angelic partial correctness} iff $\pre \subseteq \awlp{\program}{\post}$,
% or equivalently $\dwp{\program}{\negate{\post}} \subseteq \negate{\pre}$.
stating that for all states in $\pre$, there must exist either a computation of $\program$ terminating in $\post$, or the possibility of $\program$ to diverge.
As the name suggests, this is very closely related to angelic total correctness, the only difference being that divergence is deemed acceptable.

This logic can be used to identify states from which divergence is possible by choosing the post to be empty. %, i.e.\ $\hoare{\pre}{\program}{\false}$.
\citet{raad2024nontermination} emphasize the relevance of reasoning about divergence, defining an \emph{under-approximate non-termination logic} which aligns with the aforementioned angelic partial correctness triple $\hoare{\pre}{\program}{\false}$.

% For the running example, we can express by $\hoare{\pwdcorrect}{\plogin}{\noaccess}$ that when entering the correct password, there is an execution of the program that either diverges or terminates without granting access.
% Similar to incorrectness logic, this proves the existence of a bug.

\paragraph{Demonic incorrectness}
$\hoare{\pre}{\program}{\post}$ is valid for \emph{demonic incorrectness} iff $\post \subseteq \dsp{\program}{\pre}$,
% or equivalently $\aslp{\program}{\negate{\pre}} \subseteq \negate{\post}$.
stating that all final states in $\post$ must be \emph{exclusively} reachable from the states in $\pre$.
% Formally, this can be written as
% \[
% 	\forall \tau \in \post, \forall \sigma \in \evalInv{\tau}: \sigma \in \pre,
% \]
% so all initial states $\sigma$ that can reach $\post$ are in $\pre$.
In particular, all states in~$\post$ must be reachable.
If we waive this requirement, we end up with partial incorrectness.
This is dual to going from total to partial correctness by waiving the termination requirement.

To the best of our knowledge, demonic incorrectness logic has not been described in the literature before.
It can be used, for example, to restrict the initial states from which errors can occur.
This is similar to the motivation for outcome logic and sufficient incorrectness logic \cite{zilberstein2023outcome,ascari2023sufficient}.
If $\hoare{\pwdincorrect}{\plogin}{\error}$ is valid for demonic incorrectness, errors can only be reached when entering the wrong password, possibly making a bug at hand less critical.%

\paragraph{Angelic partial incorrectness}
$\hoare{\pre}{\program}{\post}$ is valid for \emph{angelic partial incorrectness} iff $\post \subseteq \aslp{\program}{\pre}$,
% or equivalently $\dsp{\program}{\negate{\pre}} \subseteq \negate{\post}$.
stating that all states in the $\post$ are either unreachable or reachable from a state in $\pre$.

To the best of our knowledge, this is also a novel logic.
If $\hoare{\pwdincorrect}{\plogin}{\access}$ is valid for angelic partial incorrectness, we know that all states where access was granted are either entirely unreachable or can be reached with an incorrect password, which is definitely undesired.

\subsection{Related Taxonomies}
\label{ssec:comparison}

Several other works have developed taxonomies for program logics.
Both \citet{ZhangKaminski22} and \citet{ascari2023sufficient} concentrate on logics based on angelic semantics, which align with our \awpsymbol\ and \aspsymbol\ transformers.
\citet{cousot2024calculational} presents a framework that defines program logics by applying various abstraction functions to a collecting semantics.
Even though the semantics in this paper in essence is also angelic, divergence is explicitly represented by the symbol $\bot$.
In this way, demonic variants of correctness logics are expressible as well.
% The postcondition is allowed to contain $\bot$.

The abstract interpretation approach offers a versatile and expressive structure for capturing a wide range of logics.
Many of these logics, including several if not all of the prominent ones, are organized in a 4-dimensional cube-like schema (see \cite[Figure 3]{cousot2024calculational}).
At first glance, this cube does not appear to resemble our taxonomy in \Cref{fig:taxonomy}.
However, a closer inspection reveals that the logics occupying the upper half of the cube once again correspond to those defined using the \awpsymbol\ and \aspsymbol\ transformers, as well as their contrapositives.
The main distinction to the corresponding fragment of \Cref{fig:taxonomy} lies in how the logics are arranged.

When comparing the cube to our planar diagram, the sets at the cube's upper corners correspond to \awpsymbol, \aspsymbol, \dwlpsymbol, and \dslpsymbol\ for a concrete pre- or postcondition.
The cube's arrangement emphasizes the symmetry between logics through \emph{set inclusions}, aligning with the close connection between Hoare logic and incorrectness logic.
This particular symmetry is less immediate in our diagram, which instead emphasizes the \emph{implications} and \emph{contrapositives} across different logics.

The logics in the lower half of Cousot's cube are more challenging to align with our taxonomy.
These logics intuitively mirror their counterparts in the upper half but introduce additional conditions related to termination behavior, allowing, for instance, the expression of \emph{demonic} partial correctness.
However, since the interpretation of these lower-half logics depends on whether nontermination (represented by $\bot$) is included in the postcondition, they are somewhat incompatible with our logics.
In particular, we do not consider incorrectness logics with termination constraints.

Notably absent from Cousot’s cube are the two novel predicate transformers we proposed - demonic strongest post and angelic strongest liberal post - and the corresponding logics.
However, the logics represented in the cube are only a subset of what is possible within the framework of abstract interpretation.
We conjecture that these new logics could be accommodated by extending the framework to include an additional symbol representing unreachability, dual to $\bot$ for divergence.

Cousot additionally introduces what he calls \emph{Hoare incorrectness logic}, which is included in the cube but does not neatly fit within its structure.
This logic represents the negation of a standard Hoare triple.
Specifically, $\hoare{\pre}{\program}{\post}$ is valid for Hoare incorrectness if it is \emph{not} valid for partial Hoare logic.
Intuitively, this means that there exists an execution of $\program$ starting in $\pre$ and terminating outside of $\post$.
Unlike traditional Hoare logic or incorrectness logic, Hoare incorrectness logic does not require conditions to hold for \emph{all} states in the precondition or postcondition; it only requires the existence of a single execution path violating the postcondition.
In our framework, this corresponds to the condition $\awp{\program}{\negate{\post}} \cap \pre \neq \emptyset$ \cite{verscht2023hoare}.
So, while in principle expressible in our framework, Hoare incorrectness logic does not align with the structure of the other logics — just as it does not fit within the cube.

\subsection{On the Impact of Additional Assumptions}
\label{ssec:assumptions}

% \todoin{structure start}
% \begin{enumerate}
% \item we will go over those
% \item dann jeweils:
% \begin{enumerate}
% \item erklärung
% \item was fällt zusammen? bild!
% \item kriterien
% \end{enumerate}
% \end{enumerate}%
% \todoin{structure end}%
% \noindent%

It is well known that under the assumption of program \emph{termination} (say on all initial states), standard partial and total correctness coincide, or rather \emph{collapse} to one notion.
More symbolically,%
\begin{align*}
	\textnormal{termination} \qqimplies \textnormal{partial corr.} \morespace{\Longleftrightarrow} \textnormal{total corr.}
\end{align*}
In the following, we will inspect what other of the 16 logics from \Cref{fig:taxonomy} collapse under the assumption of termination, and we will explore three more natural assumptions which will make yet other logics collapse, namely \emph{reachability}, \emph{determinism}, and \emph{reversibility}.

%It was hinted at before that the transformers we consider are equal under certain assumptions on the program.
%For example, the weakest pre equals the weakest liberal pre for programs that always terminate.
%From the strongest post perspective, we have the same situation for programs where all final states are reachable.
%We will discuss in this section how such assumptions impact the taxonomy in \Cref{fig:taxonomy}.

\subsubsection{Termination}
\label{sssec:termination}

Liberality in weakest precondition transformers is about whether they deem nontermination acceptable behavior or not.
If $\program$ terminates on all states, then $\dwlpsymbol$ and $\dwpsymbol$ coincide for all postconditions.
The same goes for $\awlpsymbol$ and $\awpsymbol$.
Formally, we have the following theorem:%
\begin{theorem}[Termination Collapse]
	\label{theo:termination}
	Let $\program$ be a program that \emph{must} terminate on all initial states.
%	Fix a state $\sigma \in \Sigma$.
	Then
	\[
		\dwp{\program}{\post} \eeq \dwlp{\program}{\post}
	\qqand
		\awp{\program}{\post} \eeq \awlp{\program}{\post}~.
	\]
\end{theorem}%
\begin{proof}
	See \Cref{ssec:proof-termination}. %\cite[Appendix F.2.1]{verscht2025taxonomy}.% 
\end{proof}%
\noindent%
\begin{figure}[t]
	\begin{adjustbox}{max width=\textwidth}
		\begin{tikzpicture}[mynode/.style={minimum width=2.7cm},node distance=0.9cm and 0.9cm]
			\node[mynode](awpLB) {$\pre \subseteq \awp{\program}{\post}$};
			\node[mynode,below=of awpLB](dwpLB) {$\pre \subseteq \dwp{\program}{\post}$};
			\node[mynode,right=of awpLB](awlpLB) {$\pre \subseteq \awlp{\program}{\post}$};
			\node[mynode,right=of awlpLB]  (dwpUB) {$\dwp{\program}{\post} \subseteq \pre$};
			\node[mynode,right=of dwpUB]  (dwlpUB) {$\dwlp{\program}{\post} \subseteq \pre$};
			\node[mynode,below=of dwlpUB]  (awlpUB) {$\awlp{\program}{\post} \subseteq \pre$};
			
			\node[mynode,below=of awlpLB] (dwlpLB) {$\pre \subseteq \dwlp{\program}{\post}$};
			\node[mynode,right=of dwlpLB] (awpUB) {$\awp{\program}{\post} \subseteq \pre$};
			\node[mynode,below=of dwlpLB] (aspUB) {$\asp{\program}{\pre} \subseteq \post$};
			\node[mynode,below=of dwpLB] (aslpUB) {$\aslp{\program}{\pre} \subseteq \post$};
			\node[mynode,below=of aslpUB]  (dslpUB) {$\dslp{\program}{\pre} \subseteq \post$};
			\node[mynode,below=of aspUB] (dspUB) {$\dsp{\program}{\pre} \subseteq \post$};
			
			\node[mynode,right=of aspUB] (dslpLB) {$\post \subseteq \dslp{\program}{\pre}$};
			\node[mynode,right=of dslpLB] (dspLB) {$\post \subseteq \dsp{\program}{\pre}$};
			\node[mynode,below=of dslpLB] (aslpLB) {$\post \subseteq \aslp{\program}{\pre}$};
			\node[mynode,below=of dspLB] (aspLB) {$\post \subseteq \asp{\program}{\pre}$};
			
			\draw[weakcontrapos] (awpUB) edge[<->] (dwlpLB);
			\draw[weakcontrapos] (dwpUB) edge[<->] (awlpLB);
			\draw[weakcontrapos,bend left=10] (awpLB) edge[<->] (dwlpUB);
			\draw[weakcontrapos,bend left=10] (dwpLB) edge[<->] (awlpUB);
			
			\draw[weakcontrapos] (aspUB) edge[<->] (dslpLB);
			\draw[weakcontrapos,bend left=10] (aspLB) edge[<->] (dslpUB);
			\draw[weakcontrapos] (dspUB) edge[<->] (aslpLB);
			\draw[weakcontrapos,bend left=10] (dspLB) edge[<->] (aslpUB);
			
			\path (aspUB) edge[weakgalois] node {} (dwlpLB);
			\path (awpUB) edge[weakgalois] node {} (dslpLB);		
			
			\path (dwpLB) edge[weakimplication] node {} (awpLB);
			\path (dwpLB) edge[implication] node {} (dwlpLB);
			\path (awpLB) edge[implication] node {} (awlpLB);
			\path (dwlpLB) edge[weakimplication] node {} (awlpLB);
			
			\path (awpUB) edge[weakimplication] node {} (dwpUB);
			\path (awlpUB) edge[implication] node {} (awpUB);
			\path (awlpUB) edge[weakimplication] node {} (dwlpUB);
			\path (dwlpUB) edge[implication] node {} (dwpUB);
			
			\path (dspLB) edge[weakimplication] node {} (aspLB);
			\path (dspLB) edge[implication] node {} (dslpLB);
			\path (aspLB) edge[implication] node {} (aslpLB);
			\path (dslpLB) edge[weakimplication] node {} (aslpLB);
			
			\path (aspUB) edge[weakimplication] node {} (dspUB);
			\path (aslpUB) edge[implication] node {} (aspUB);
			\path (aslpUB) edge[weakimplication] node {} (dslpUB);
			\path (dslpUB) edge[implication] node {} (dspUB);
			
			\def\innerxsepp{-2.5pt}
			\def\innerysepp{-1pt}
			\def\roundedcorn{0.1cm}
			\node (anim1) [collapsewp,fit = {(awpLB)(awlpLB)}] {}; 
			\node (anim1) [collapsewp,fit = {(dwpLB)(dwlpLB)}] {}; 
			\node (anim1) [collapsewp,fit = {(awpUB)(awlpUB)}] {}; 
			\node (anim1) [collapsewp,fit = {(dwpUB)(dwlpUB)}] {};
		
			\node (anim1) [collapsesp,fit = {(aspLB)(aslpLB)}] {}; 
			\node (anim1) [collapsesp,fit = {(dspLB)(dslpLB)}] {}; 
			\node (anim1) [collapsesp,fit = {(aspUB)(aslpUB)}] {}; 
			\node (anim1) [collapsesp,fit = {(dspUB)(dslpUB)}] {};	
		\end{tikzpicture}
	\end{adjustbox}
	\caption{
		Collapse of program logics under assumptions regarding totality and partiality.
		Blue \emph{dashed} boxes enclose logics that collapse (i.e.\ the implications become equivalences) under \emph{termination}.
		Blue \emph{solid} boxes enclose logics that collapse under \emph{reachability}.
	}
	\label{fig:taxonomy-termination-reachability}
\end{figure}
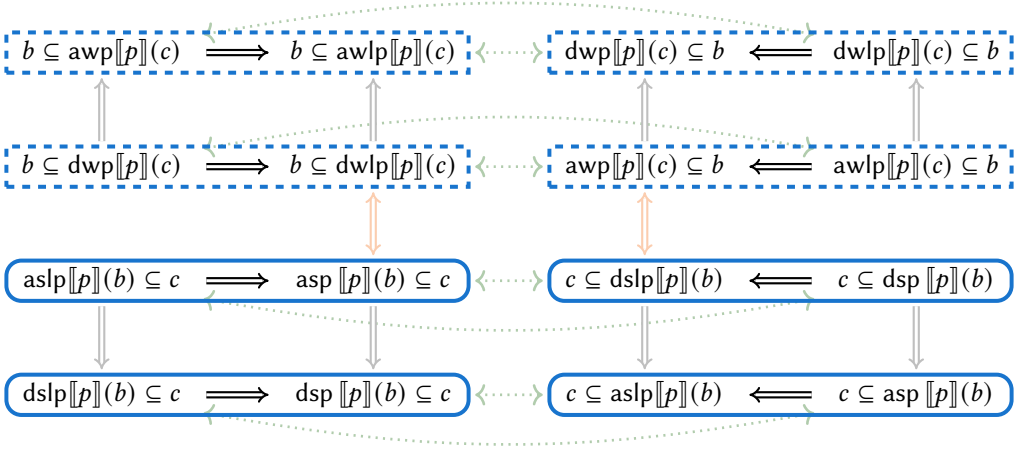%
If two transformers $\textsf{t}_1$ and $\textsf{t}_2$ coincide, then two logics whose definitions are equal up to whether their definition invokes $\textsf{t}_1$ or $\textsf{t}_2$ immediately collapse into a single logic. 
Hence, for instance, partial and total correctness collapse to one notion.
Logics that collapse are illustrated by the dashed boxes in \Cref{fig:taxonomy-termination-reachability}.
We see that in the upper half, logics collapse along a horizontal axis.
%We can think of this as a collapsing of the two outer logics in the upper half of the picture.
The bottom half of the logics is unaffected by termination, as strongest post based logics are incapable of reasoning about termination or divergence.

For the purpose of logics collapsing, note that we can loosen the requirement to computations starting in initial states of interest, i.e.\ those satisfying the precondition $\pre$.
Concretely, if a program $\program$ terminates on all initial states satisfying $\pre$, then $\pre \subseteq \dwp{\program}{\post} \iff \pre \subseteq \dwlp{\program}{\post}$.
The same holds for the other $\wpsymbol$ based logics.

We can moreover use the predicate transformers to express termination properties, which follows directly from \Cref{theo:termination}.%
\begin{corollary}[Expressing Termination Properties]
	\label{theo:termination-properties}
	A program $\program$ \emph{must} terminate on all initial states\footnote{%
		One could be a bit more fine-grained here and distinguish between \emph{may} and \emph{must} termination.
		While must implies may termination, we could only check for may termination, which is the case if and only if%
		\[
			\awp{\program}{\true} = \true
		\qqorequiv
			\awlp{\program}{\false} = \false.
		\]
		In that case, only the angelic but not necessarily the demonic weakest precondition transformers would coincide (cf.~\Cref{theo:termination}) and in \Cref{fig:taxonomy-termination-reachability} only the dashed boxes containing angelic transformers would be present.
	} if and only if
	\[
		\dwp{\program}{\true} \eeq \true
%	\]
	\qqorequiv
%	\[
		\dwlp{\program}{\false} \eeq \false.
	\]%
\end{corollary}%
\noindent%
\Cref{theo:termination} and \Cref{theo:termination-properties} together yield precisely the well-known technique of proving partial correctness and proving termination separately in order to obtain total correctness.

\subsubsection{Reachability}
\label{sssec:reachability}

For weakest preconditions, we saw that -- under the assumption of \mbox{termination --} liberal and non-liberal weakest preconditions coincide.
This immediately raises the question whether such a criterion can be found for strongest postcondition transformers to coincide.
The answer is yes: \emph{reachability}.
If all (final) states are reachable from some initial state by some computation of $\program$, then liberal and non-liberal strongest postconditions coincide:%
\begin{theorem}[Reachability Collapse]
	\label{theo:reachability}
%	Let $\program$ be a program and $\pre$ be a precondition.
%	Fix a state $\tau \in \Sigma$.
	Let all final states be reachable by some computation of a program $\program$.
	Then%
	\[
		\asp{\program}{\pre} \eeq \aslp{\program}{\pre}
%	\]
	\qqand
%	\[
		\dsp{\program}{\pre} \eeq \dslp{\program}{\pre}~.
	\]
\end{theorem}
\begin{proof}
	See \Cref{ssec:proof-reachability}. % \cite[Appendix F.2.2]{verscht2025taxonomy}. %
\end{proof}
\noindent%
As in the case for termination, coincidence of transformers causes the associated logics to collapse. 
Logics that collapse for reachability are illustrated by the solid boxes in \Cref{fig:taxonomy-termination-reachability}.
%We can think of this as a collapsing of the two outer logics in the upper half of the picture.
The top half of the logics is unaffected by reachability.

Reachability of \emph{all} final states is a \emph{very} strong (and perhaps sometimes even undesired) assumption.
As a silver lining, for the collapse of logics, we can loosen this requirement to the final states satisfying $\post$, similar to what we have seen for termination.
% Assuming all states in $\post$ are reachable, then $\post \subseteq \asp{\program}{\pre} \iff \post \subseteq \aslp{\program}{\pre}$.
% The same holds for the other $\spsymbol$ based logics.

In \Cref{theo:termination}, we had to assume that \emph{all} computations of $\program$ terminate.
In contrast to this, \Cref{theo:reachability} only requires a final state to be reachable by \emph{some} computation.
Again, this difference is caused by \Cref{obs:branching}: A final state can either be unreachable or not, but never both.

As with termination, we can express reachability in terms of predicate transformers:%
\begin{corollary}[Expressing Reachability Properties]
	\label{theo:reachability-properties}
	All final states are reachable from some initial state by some computation of program $\program$, if and only if
	\begin{align*}
		&\dsp{\program}{\true}(\tau) = \true
%	\]
	\qqorequiv
%	\[
		\dslp{\program}{\false}(\tau) = \false
%	\]
	\qqorequiv\hspace*{-2em}\\
%	\[
		&\asp{\program}{\true}(\tau) = \true
%	\]
	\qqorequiv
%	\[
		\aslp{\program}{\false}(\tau) = \false.
	\end{align*}
\end{corollary}
\begin{proof}
	See \Cref{ssec:proof-reachability-properties}. %\cite[Appendix F.2.3]{verscht2025taxonomy}. %
\end{proof}%
This interestingly differs from the analogous corollary for termination (\Cref{theo:termination-properties}), where we had to distinguish between \emph{may} and \emph{must} termination (and opted for must as it is more general).
This is again rooted in \Cref{obs:branching}, as the existence of a path to a final state is equivalent to the final state being reachable.

As with total correctness, \Cref{theo:reachability} and \Cref{theo:reachability-properties} together yield a technique for proving incorrectness by proving angelic partial incorrectness and reachability.

\subsubsection{Determinism}
\label{sssec:determinism}

So far, we have seen how to produce \enquote{horizontal collapses} in \Cref{fig:taxonomy} by assuming that total and partial predicate transformers coincide.
Now, we will see that we can get similar results for angelic and demonic transformers, thus producing \enquote{vertical collapses}.
For weakest preconditions, angelic and demonic weakest preconditions differ in their treatment of explicit nondeterministic branching of the program.
Therefore, it is evident that when restricting to \emph{deterministic} programs (even syntactically), they should be equivalent.\footnote{Of course, determinism is also a semantic property, but we will not go into this as syntactic determinism obviously implies semantic determinism and semantic determinism is difficult to reason about.}%
\begin{theorem}[Determinism Collapse]
	\label{theo:determinism}
	If program $\program$ is \emph{deterministic} then
	\[
		\awp{\program}{\post} \eeq \dwp{\program}{\post}
%	\]
	\qqand
%	\[
		\awlp{\program}{\post} \eeq \dwlp{\program}{\post}~.
	\]
\end{theorem}
\begin{proof}
	See \Cref{ssec:proof-determinism}.
\end{proof}
\noindent%
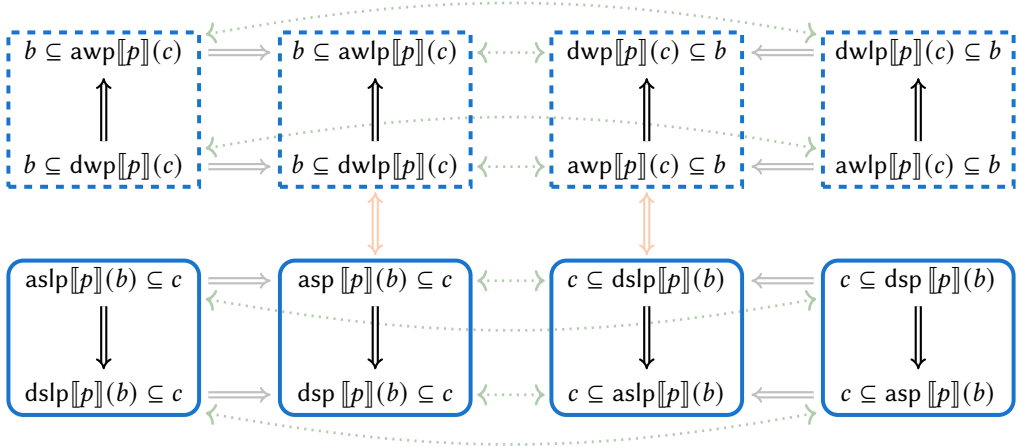
\begin{figure}[t]
	\begin{adjustbox}{max width=\textwidth}
		\begin{tikzpicture}[mynode/.style={minimum width=2.7cm},node distance=0.9cm and 0.9cm]
			\node[mynode](awpLB) {$\pre \subseteq \awp{\program}{\post}$};
			\node[mynode,below=of awpLB](dwpLB) {$\pre \subseteq \dwp{\program}{\post}$};
			\node[mynode,right=of awpLB](awlpLB) {$\pre \subseteq \awlp{\program}{\post}$};
			\node[mynode,right=of awlpLB]  (dwpUB) {$\dwp{\program}{\post} \subseteq \pre$};
			\node[mynode,right=of dwpUB]  (dwlpUB) {$\dwlp{\program}{\post} \subseteq \pre$};
			\node[mynode,below=of dwlpUB]  (awlpUB) {$\awlp{\program}{\post} \subseteq \pre$};
			
			\node[mynode,below=of awlpLB] (dwlpLB) {$\pre \subseteq \dwlp{\program}{\post}$};
			\node[mynode,right=of dwlpLB] (awpUB) {$\awp{\program}{\post} \subseteq \pre$};
			\node[mynode,below=of dwlpLB] (aspUB) {$\asp{\program}{\pre} \subseteq \post$};
			\node[mynode,below=of dwpLB] (aslpUB) {$\aslp{\program}{\pre} \subseteq \post$};
			\node[mynode,below=of aslpUB]  (dslpUB) {$\dslp{\program}{\pre} \subseteq \post$};
			\node[mynode,below=of aspUB] (dspUB) {$\dsp{\program}{\pre} \subseteq \post$};
			
			\node[mynode,right=of aspUB] (dslpLB) {$\post \subseteq \dslp{\program}{\pre}$};
			\node[mynode,right=of dslpLB] (dspLB) {$\post \subseteq \dsp{\program}{\pre}$};
			\node[mynode,below=of dslpLB] (aslpLB) {$\post \subseteq \aslp{\program}{\pre}$};
			\node[mynode,below=of dspLB] (aspLB) {$\post \subseteq \asp{\program}{\pre}$};
			
			\draw[weakcontrapos] (awpUB) edge[<->] (dwlpLB);
			\draw[weakcontrapos] (dwpUB) edge[<->] (awlpLB);
			\draw[weakcontrapos,bend left=10] (awpLB) edge[<->] (dwlpUB);
			\draw[weakcontrapos,bend left=10] (dwpLB) edge[<->] (awlpUB);
			
			\draw[weakcontrapos] (aspUB) edge[<->] (dslpLB);
			\draw[weakcontrapos,bend left=10] (aspLB) edge[<->] (dslpUB);
			\draw[weakcontrapos] (dspUB) edge[<->] (aslpLB);
			\draw[weakcontrapos,bend left=10] (dspLB) edge[<->] (aslpUB);
			
			\path (aspUB) edge[weakgalois] node {} (dwlpLB);
			\path (awpUB) edge[weakgalois] node {} (dslpLB);		
			
			\path (dwpLB) edge[implication] node {} (awpLB);
			\path (dwpLB) edge[weakimplication] node {} (dwlpLB);
			\path (awpLB) edge[weakimplication] node {} (awlpLB);
			\path (dwlpLB) edge[implication] node {} (awlpLB);
			
			\path (awpUB) edge[implication] node {} (dwpUB);
			\path (awlpUB) edge[weakimplication] node {} (awpUB);
			\path (awlpUB) edge[implication] node {} (dwlpUB);
			\path (dwlpUB) edge[weakimplication] node {} (dwpUB);
			
			\path (dspLB) edge[implication] node {} (aspLB);
			\path (dspLB) edge[weakimplication] node {} (dslpLB);
			\path (aspLB) edge[weakimplication] node {} (aslpLB);
			\path (dslpLB) edge[implication] node {} (aslpLB);
			
			\path (aspUB) edge[implication] node {} (dspUB);
			\path (aslpUB) edge[weakimplication] node {} (aspUB);
			\path (aslpUB) edge[implication] node {} (dslpUB);
			\path (dslpUB) edge[weakimplication] node {} (dspUB);
			
			\def\innerxsepp{-2.5pt}
			\def\innerysepp{-1pt}
			\def\roundedcorn{0.1cm}
			\node (anim1) [collapsewp,fit = {(awpLB)(dwpLB)}] {}; 
			\node (anim1) [collapsewp,fit = {(dwlpLB)(awlpLB)}] {}; 
			\node (anim1) [collapsewp,fit = {(awpUB)(dwpUB)}] {}; 
			\node (anim1) [collapsewp,fit = {(awlpUB)(dwlpUB)}] {};
			
			\node (anim1) [collapsesp,fit = {(aspLB)(dspLB)}] {}; 
			\node (anim1) [collapsesp,fit = {(dslpLB)(aslpLB)}] {}; 
			\node (anim1) [collapsesp,fit = {(aspUB)(dspUB)}] {}; 
			\node (anim1) [collapsesp,fit = {(aslpUB)(dslpUB)}] {};

		\end{tikzpicture}
	\end{adjustbox}
	\caption{
		Collapse of program logics under certain assumptions, part 2.
		Blue \emph{dashed} boxes enclose logics that collapse (i.e.\ the implications become equivalences) under \emph{determinism}.
		Blue \emph{solid} boxes enclose logics that collapse under \emph{reversibility}.%
	}
	\label{fig:taxonomy-determinsm-reversibility}
\end{figure}%
In \Cref{fig:taxonomy-determinsm-reversibility}, we now see what we were expecting:
The top row of our taxonomy collapses into the second but top row under the assumption of determinism, symbolized by the blue dashed boxes.
As before, for the collapse it suffices to require determinism of computation started in $\pre$.

\citet[Prop.\ 5.5]{ascari2023sufficient} established that angelic total correctness (their sufficient incorrectness logic) and demonic partial correctness (Hoare logic) are equivalent for deterministic programs which terminate.
By looking at the bigger picture, we see why this holds:
Combining \Cref{theo:termination,theo:determinism},  all four logics in the upper left (as well as the upper right) quadrat collapse, including the two logics mentioned above.

\subsubsection{Reversibility}
\label{sssec:reversibility}

Finally, we are missing to collapse the bottom row of \Cref{fig:taxonomy-determinsm-reversibility} into the second but last row.
This can be achieved by ensuring \enquote{backward determinism} of $\program$, meaning that every final state could have only been reached from (at most) one initial state.
In other words, the computation of the program is \emph{reversible}.
This is of interest, for example, in compression algorithms or quantum computations.
\citet{ascari2023sufficient} observed that under reversibility, incorrectness logic implies demonic partial incorrectness.
This can be seen as a consequence of the following theorem:%
\begin{theorem}[Reversibility Collapse]
	\label{theo:reversibility}
	It $\program$ is a \emph{reversible} program (i.e.\ $\seminv{\program}{\tau}$ is either a singleton or the empty set for all $\tau \in \Sigma$), then
	\[
		\asp{\program}{\pre} \eeq \dsp{\program}{\pre}
%	\]
	\qqand
%	\[
		\aslp{\program}{\pre} \eeq \dslp{\program}{\pre}.
	\]
\end{theorem}
\begin{proof}
	See \Cref{ssec:proof-reversibility}. %\cite[Appendix F.2.5]{verscht2025taxonomy}. %
\end{proof}%
\noindent%
Analogously to the other cases, reversibility does not effect the top half of the logics for the same reason that (non)determinism has no effect on the lower half.
Also, we can again weaken the requirement to reversibility of computation terminating in $\post$.

\subsection{Symmetries and Asymmetries}
\label{ssec:asymmetries}

\begin{flushright}
\textsl{Program correctness and incorrectness are two sides of the same coin.}\\[.25em]
--- Peter \citet{OHearn19}
\end{flushright}%
\noindent%
\Cref{fig:taxonomy} is at first glance full of symmetry and duality:
The upper half contains all weakest precondition based logics, the lower half dually contains all strongest postcondition based logics.
Also, the left half is a contrapositive mirroring of the right half.
The assumptions discussed in the preceding section also act completely symmetrically.
On the top left we have \enquote{correctness logics}, on the bottom right we have \enquote{incorrectness logics}.
Indeed, this all seems like two opposite sides of a (multidimensional) coin.
% Note that the classic logics, e.g., partial correctness and incorrectness, are situated in the upper left and the lower right part as they are underapproximations of the respective transformers.

\paragraph{Two sides of the same coin? Not quite.}
However, when taking a closer look, this fully symmetric picture starts to become a bit brittle.
First of all, correctness and incorrectness are not really at opposite sides of the \enquote{coin} represented by \Cref{fig:taxonomy}.
If total correctness (second row left), was truly on the opposite side of incorrectness logic (bottom row right), we should either expect the standard notion of program correctness to be \emph{Lisbon logic} (top row left), or alternatively the standard notion of incorrectness to be \emph{demonic incorrectness} (third row right).
It seems off that the current standard notions of correctness and incorrectness are \emph{not} really at opposite sides of \Cref{fig:taxonomy}.

% In \Cref{ssec:comparison}, we compared our diagram to the cube proposed by \citet{cousot2024calculational}, where the standard notions of correctness and incorrectness \emph{are} positioned on opposite sides.
% If we imagine folding our planar \Cref{fig:taxonomy} into a similar 4-dimensional cube, this symmetry would also emerge, placing correctness and incorrectness in opposite corners. 

\paragraph{Missing Galois connections.}
Amongst the logics in \Cref{fig:taxonomy}, there is \emph{essentially only one} Galois connection (there are two, but one is the contrapositive of the other).
In particular, since there is a Galois connection amongst the two transformer-based definitions of partial correctness, one might expect more Galois connections, e.g.\ between \emph{angelic} partial correctness and the contrapositive of angelic partial incorrectness.
Such Galois connection does not exists, however: 
Assume that $\hoare{\pre}{\program}{\post}$ is valid for angelic partial correctness, i.e.\ $\pre \subseteq \awlp{\program}{\post}$.
So all states in $\pre$ must either have a diverging path or a path to $\post$.
Let $\tau \in \post$ be a final state which is exclusively reachable from $\negate{\pre}$.
This does not violate the assumption of angelic partial correctness.
However, the contrapositive of angelic partial incorrectness, $\dsp{\program}{\negate{\pre}} \subseteq \negate{\post}$ requires all states that are exclusively reachable from~$\negate{\pre}$ to be included in $\negate{\post}$.
This implies that $\tau \in \negate{\post}$, which contradicts the previous assumption.
Hence, $\hoare{\pre}{\program}{\post}$ is \emph{not} valid for the contrapositive of angelic partial incorrectness.
Similar counterexamples exist for all other combinations of triples (see \Cref{sec:counter} for a collection).
We conclude that there is (essentially) only one Galois connection between the logics we considered.
We can, however, force such Galois connections by additional assumptions.
E.g., the above discussed Galois connection \emph{is} valid for programs that are both \emph{deterministic and reversible}.

\paragraph{Missing Asymmetries}
We discussed an intrinsic asymmetry of forward and backward analyses in \Cref{obs:branching}.
This asymmetry is not (yet) visible in \Cref{fig:taxonomy}.
To gain more insights on the logics' relationships and (a)symmetries, we will now take another look at program properties from the perspective of Kleene algebras.

% !TEX root = ../main.tex

\section{Kleene Algebra with Top and Tests}
\label{sec:kat}
% We transition from transformer-style logic to an algebraic characterization of program properties.
% By treating programs as elements of an algebra, we can express properties within the corresponding equational system.
% This perspective enables us to validate and refine the results obtained in the previous section.

\paragraph{Kleene Algebra with Tests (\KAT)}
Introduced by \citet{kozen1997kleene}, \KAT is an algebraic approach to specifying and reasoning about program properties.
For us, it will suffice to describe \KAT at a rather high level.
For reference, the formal definitions we require can be found in \Cref{sec:kat-definitions}.

\KAT terms are generalized regular expressions over a \emph{two-sorted alphabet} consisting of (i)~programs ($p,q,\ldots$) and (ii)~tests ($b,c,\ldots$).
We interpret these symbols as relations: 
A program $\program$ relates (or maps) initial states to final states through its execution.
Initial states on which $\program$ \emph{must} diverge are not related to any final states.
Dually, unreachable final states are not related to \mbox{any initial states}.

A test $\pre$ maps initial states that satisfy $\pre$ to themselves.
All other states are not contained in the relation.
Hence, tests act as filters.
Testing for $\textsf{false}$ (the \emph{empty} relation) is denoted by $0$ and is the least element
in the lattice of relations (ordered by set inclusion).

Composing symbols corresponds to composing relations.
For example, the term $bpqc$ intuitively means: 
First test for $\pre$, then execute~$p$, then execute~$q$, and finally test for $\post$.
Executions that fail to satisfy $\pre$ initially or $\post$ finally are filtered out and do not become part of the resulting relation.%

\paragraph{Kleene Algebra with Top and Tests (\TopKAT)} 
In \TopKAT~\cite{zhang2022incorrectness}, an additional $\top$ element (the \emph{universal} relation relating \emph{all} states with each other) is added to \KAT.
Notice the difference to the identity relation, denoted $1$, relating every state to itself.
$\top$ can be used to \enquote{select} the domain or codomain of a relation, as the following example illustrates.
\begin{figure}[t]
	\begin{subfigure}[t]{0.23\textwidth}
		\begin{center}
		\begin{adjustbox}{max width=0.99\textwidth}
		\begin{tikzpicture}[node distance=0.8cm and 2cm,line width=1.1pt,->]
			% right of T nodes
			\node[mystate, pre] (I1) {$\statelabel{1}$};
			\node[mystate, selectnode, pre] (I2) [below =of I1] {$\statelabel{2}$};
			\node[mystate, selectnode, pre] (I3) [below =of I2] {$\statelabel{3}$};
			\node[mystate] (I4) [below =of I3] {$\statelabel{4}$};
			\node[mystate] (I5) [below =of I4] {$\statelabel{5}$};

			% leftmost nodes (left of T)
			% \node[mystate, selectnode] (I1T) [left =of I1]  {$\statelabel{1}$};
			% \node[mystate, selectnode] (I2T) [left =of I2] {$\statelabel{2}$};
			% \node[mystate, selectnode] (I3T) [left =of I3] {$\statelabel{3}$};
			% \node[mystate, selectnode] (I4T) [left =of I4] {$\statelabel{4}$};
			% \node[mystate, selectnode] (I5T) [left =of I5] {$\statelabel{5}$};
			
			% right of b nodes
			\node[mystate,pre] (I1a) [right =of I1]  {$\statelabel{1}$};
			\node[mystate, selectnode,pre] (I2a) [right =of I2] {$\statelabel{2}$};
			\node[mystate, selectnode,pre] (I3a) [right =of I3] {$\statelabel{3}$};
			\node[mystate] (I4a) [right =of I4] {$\statelabel{4}$};
			\node[mystate] (I5a) [right =of I5] {$\statelabel{5}$};
			
			% right of p nodes
			\node[mystate] (I1b) [right =of I1a]  {$\statelabel{1}$};
			\node[mystate, selectnode,post] (I2b) [right =of I2a] {$\statelabel{2}$};
			\node[mystate,post] (I3b) [right =of I3a] {$\statelabel{3}$};
			\node[mystate,post] (I4b) [right =of I4a] {$\statelabel{4}$};
			\node[mystate] (I5b) [right =of I5a] {$\statelabel{5}$};
			
			\node[rectangle,fill=black,scale=1.2] (choice) at ($(I4a)!0.5!(I4b)$) {};
			\node (spiral) at ($(I5a)!0.6!(I5b)$) {};

			% right of c nodes
			\node[mystate] (I1c) [right =of I1b]  {$\statelabel{1}$};
			\node[mystate,selectnode,post] (I2c) [right =of I2b] {$\statelabel{2}$};
			\node[mystate,post] (I3c) [right =of I3b] {$\statelabel{3}$};
			\node[mystate,post] (I4c) [right =of I4b] {$\statelabel{4}$};
			\node[mystate] (I5c) [right =of I5b] {$\statelabel{5}$};

			%\node[font=\LARGE] (T) at ($(I1T)!0.5!(I1)$) [yshift=0.8cm] {$\top$};
			\node[font=\LARGE] (b) at ($(I1)!0.5!(I1a)$) [yshift=0.8cm] {$\pre$};
			\node[font=\LARGE] (p) at ($(I1a)!0.5!(I1b)$) [yshift=0.68cm] {$\program$};
			\node[font=\LARGE] (c) at ($(I1b)!0.5!(I1c)$) [yshift=0.8cm] {$\post$};
			
			\coordinate (spiralstart) at ($(spiral)+(-0.1,-0.62)$); % y coordinate is -0.002*30*30*0.35 + 0.1 (last plotted point times scale, plus 0.1 to match)
			\draw let \p1=(I5b) in 
				[scale=0.35,domain=0:30,variable=\t,smooth,samples=75,-,shift={(spiralstart)}] plot  ({\t r}: {-0.002*\t*\t});
			
			% \path (I1T) edge node {} (I1);
			% \path (I1T) edge node {} (I2);
			% \path (I1T) edge node {} (I3);
			% \path (I1T) edge node {} (I4);
			% \path (I1T) edge node {} (I5);
			
			% \path (I2T) edge node {} (I1);
			% % \path (I2T) edge node {} (I2);
			% % \path (I2T) edge node {} (I3);
			% \path (I2T) edge node {} (I4);
			% \path (I2T) edge node {} (I5);
			
			% \path (I3T) edge node {} (I1);
			% % \path (I3T) edge node {} (I2);
			% % \path (I3T) edge node {} (I3);
			% \path (I3T) edge node {} (I4);
			% \path (I3T) edge node {} (I5);
			
			% \path (I4T) edge node {} (I1);
			% % \path (I4T) edge node {} (I2);
			% % \path (I4T) edge node {} (I3);
			% \path (I4T) edge node {} (I4);
			% \path (I4T) edge node {} (I5);
			
			% \path (I5T) edge node {} (I1);
			% % \path (I5T) edge node {} (I2);
			% % \path (I5T) edge node {} (I3);
			% \path (I5T) edge node {} (I4);
			% \path (I5T) edge node {} (I5);

			% b
			\path (I1) edge node {} (I1a);
			% \path (I2) edge node {} (I2a);
			% \path (I3) edge node {} (I3a);

			% p
			\path (I1a) edge node {} (I1b);
			% \path (I2a) edge node {} (I2b);
			% \path (I3a) edge node {} (I2b);
			\path (I4a) edge node {} (I4b);
			\path ($(I4a)!0.5!(I4b)$) edge node {} (I3b);
			\path (I5a) edge [-] node {} (spiral);
			
			% c
			\path[selected] (I2b) edge node {} (I2c);
			\path (I3b) edge node {} (I3c);
			\path (I4b) edge node {} (I4c);

			% extra bold paths:
			\path[selected] (I2) edge node {} (I2a);
			\path[selected] (I3) edge node {} (I3a);
			\path[selected] (I2a) edge node {} (I2b);
			\path[selected] (I3a) edge node {} (I2b);
		
			% \path[selected] (I1T) edge node {} (I2);
			% \path[selected] (I2T) edge node {} (I2);
			% \path[selected] (I3T) edge node {} (I2);
			% \path[selected] (I4T) edge node {} (I2);
			% \path[selected] (I5T) edge node {} (I2);
			
			% \path[selected] (I1T) edge node {} (I3);
			% \path[selected] (I2T) edge node {} (I3);
			% \path[selected] (I3T) edge node {} (I3);
			% \path[selected] (I4T) edge node {} (I3);
			% \path[selected] (I5T) edge node {} (I3);

			%\node[font=\LARGE] (equ) at ($(I5)!0.5!(I5a)$) [yshift=-2.5cm] {};
			
		\end{tikzpicture} 
	\end{adjustbox}
	\end{center}
	\captionsetup{font=footnotesize,labelfont=footnotesize}
	\caption{ \footnotesize
		Illustration of a \TopKAT term $\pre \program \post$.
		The resulting relation is highlighted in bold red.
	}
	\label{subfig:codomain-a}
	\end{subfigure}
	\qquad
%	\vspace{10pt}
%
	%%%%%%%%%%%%%%%%%%%%%%%%%
	%%%%%%%%%%%%%%%%%%%%%%%%%
	%%%%%%%%%%%%%%%%%%%%%%%%%
	%% CODOMAIN
	%%%%%%%%%%%%%%%%%%%%%%%%%
	%%%%%%%%%%%%%%%%%%%%%%%%%
	%%%%%%%%%%%%%%%%%%%%%%%%%
	\begin{subfigure}[t]{0.3\textwidth}
			\begin{center}
			\begin{adjustbox}{max width=0.99\textwidth}
			\begin{tikzpicture}[node distance=0.8cm and 2cm,line width=1.1pt,->]
				% right of T nodes
				\node[mystate, pre] (I1) {$\statelabel{1}$};
				\node[mystate, selectnode, pre] (I2) [below =of I1] {$\statelabel{2}$};
				\node[mystate, selectnode, pre] (I3) [below =of I2] {$\statelabel{3}$};
				\node[mystate] (I4) [below =of I3] {$\statelabel{4}$};
				\node[mystate] (I5) [below =of I4] {$\statelabel{5}$};
	
				% leftmost nodes (left of T)
				\node[mystate, selectnode] (I1T) [left =of I1]  {$\statelabel{1}$};
				\node[mystate, selectnode] (I2T) [left =of I2] {$\statelabel{2}$};
				\node[mystate, selectnode] (I3T) [left =of I3] {$\statelabel{3}$};
				\node[mystate, selectnode] (I4T) [left =of I4] {$\statelabel{4}$};
				\node[mystate, selectnode] (I5T) [left =of I5] {$\statelabel{5}$};
				
				% right of b nodes
				\node[mystate,pre] (I1a) [right =of I1]  {$\statelabel{1}$};
				\node[mystate, selectnode,pre] (I2a) [right =of I2] {$\statelabel{2}$};
				\node[mystate, selectnode,pre] (I3a) [right =of I3] {$\statelabel{3}$};
				\node[mystate] (I4a) [right =of I4] {$\statelabel{4}$};
				\node[mystate] (I5a) [right =of I5] {$\statelabel{5}$};
				
				% right of p nodes
				\node[mystate] (I1b) [right =of I1a]  {$\statelabel{1}$};
				\node[mystate, selectnode,post] (I2b) [right =of I2a] {$\statelabel{2}$};
				\node[mystate,post] (I3b) [right =of I3a] {$\statelabel{3}$};
				\node[mystate,post] (I4b) [right =of I4a] {$\statelabel{4}$};
				\node[mystate] (I5b) [right =of I5a] {$\statelabel{5}$};
				
				\node[rectangle,fill=black,scale=1.2] (choice) at ($(I4a)!0.5!(I4b)$) {};
				\node (spiral) at ($(I5a)!0.6!(I5b)$) {};

				\node[mystate] (I1c) [right =of I1b]  {$\statelabel{1}$};
				\node[mystate,selectnode,post] (I2c) [right =of I2b] {$\statelabel{2}$};
				\node[mystate,post] (I3c) [right =of I3b] {$\statelabel{3}$};
				\node[mystate,post] (I4c) [right =of I4b] {$\statelabel{4}$};
				\node[mystate] (I5c) [right =of I5b] {$\statelabel{5}$};
				
				\node[font=\LARGE] (T) at ($(I1T)!0.5!(I1)$) [yshift=0.8cm] {$\top$};
				\node[font=\LARGE] (b) at ($(I1)!0.5!(I1a)$) [yshift=0.8cm] {$\pre$};
				\node[font=\LARGE] (p) at ($(I1a)!0.5!(I1b)$) [yshift=0.68cm] {$\program$};
				\node[font=\LARGE] (c) at ($(I1b)!0.5!(I1c)$) [yshift=0.8cm] {$\post$};
				
				\coordinate (spiralstart) at ($(spiral)+(-0.1,-0.62)$); % y coordinate is -0.002*30*30*0.35 + 0.1 (last plotted point times scale, plus 0.1 to match)
				\draw let \p1=(I5b) in 
					[scale=0.35,domain=0:30,variable=\t,smooth,samples=75,-,shift={(spiralstart)}] plot  ({\t r}: {-0.002*\t*\t});
			
				\path (I1T) edge node {} (I1);
				% \path (I1T) edge node {} (I2);
				% \path (I1T) edge node {} (I3);
				\path (I1T) edge node {} (I4);
				\path (I1T) edge node {} (I5);
				
				\path (I2T) edge node {} (I1);
				% \path (I2T) edge node {} (I2);
				% \path (I2T) edge node {} (I3);
				\path (I2T) edge node {} (I4);
				\path (I2T) edge node {} (I5);
				
				\path (I3T) edge node {} (I1);
				% \path (I3T) edge node {} (I2);
				% \path (I3T) edge node {} (I3);
				\path (I3T) edge node {} (I4);
				\path (I3T) edge node {} (I5);
				
				\path (I4T) edge node {} (I1);
				% \path (I4T) edge node {} (I2);
				% \path (I4T) edge node {} (I3);
				\path (I4T) edge node {} (I4);
				\path (I4T) edge node {} (I5);
				
				\path (I5T) edge node {} (I1);
				% \path (I5T) edge node {} (I2);
				% \path (I5T) edge node {} (I3);
				\path (I5T) edge node {} (I4);
				\path (I5T) edge node {} (I5);

				% b
				\path (I1) edge node {} (I1a);
				% \path (I2) edge node {} (I2a);
				% \path (I3) edge node {} (I3a);

				% p
				\path (I1a) edge node {} (I1b);
				% \path (I2a) edge node {} (I2b);
				% \path (I3a) edge node {} (I2b);
				\path (I4a) edge node {} (I4b);
				\path ($(I4a)!0.5!(I4b)$) edge node {} (I3b);
				\path (I5a) edge [-] node {} (spiral);
				
				% c
				\path[selected] (I2b) edge node {} (I2c);
				\path (I3b) edge node {} (I3c);
				\path (I4b) edge node {} (I4c);
				
				% extra bold paths:
				\path[selected] (I2) edge node {} (I2a);
				\path[selected] (I3) edge node {} (I3a);
				\path[selected] (I2a) edge node {} (I2b);
				\path[selected] (I3a) edge node {} (I2b);
			
				\path[selected] (I1T) edge node {} (I2);
				\path[selected] (I2T) edge node {} (I2);
				\path[selected] (I3T) edge node {} (I2);
				\path[selected] (I4T) edge node {} (I2);
				\path[selected] (I5T) edge node {} (I2);
				
				\path[selected] (I1T) edge node {} (I3);
				\path[selected] (I2T) edge node {} (I3);
				\path[selected] (I3T) edge node {} (I3);
				\path[selected] (I4T) edge node {} (I3);
				\path[selected] (I5T) edge node {} (I3);

				%\node[font=\LARGE] (equ) at ($(I5)!0.5!(I5a)$) [yshift=-2.5cm] {};
				
			\end{tikzpicture}
		\end{adjustbox}
		\end{center}
		\captionsetup{font=footnotesize,labelfont=footnotesize}
		\caption{ \footnotesize
			Adding $\top$ to the left-hand side of the term $\pre \program \post$ from \Cref{subfig:codomain-a}, effectively selecting the \emph{codomain} of the underlying relation.
		}
		\label{subfig:codomain-b}
		\end{subfigure}% 
		\qquad
		%%%%%%%%%%%%%%%%%%%%%%%%%
		%%%%%%%%%%%%%%%%%%%%%%%%%
		%%%%%%%%%%%%%%%%%%%%%%%%%
		%% DOMAIN
		%%%%%%%%%%%%%%%%%%%%%%%%%
		%%%%%%%%%%%%%%%%%%%%%%%%%
		%%%%%%%%%%%%%%%%%%%%%%%%%
		\begin{subfigure}[t]{0.3\textwidth}
			\begin{center}
				\begin{adjustbox}{max width=0.99\textwidth}
				\begin{tikzpicture}[node distance=0.8cm and 2cm,line width=1.1pt,->]
					% right of T nodes
					\node[mystate, pre] (I1) {$\statelabel{1}$};
					\node[mystate, selectnode, pre] (I2) [below =of I1] {$\statelabel{2}$};
					\node[mystate, selectnode, pre] (I3) [below =of I2] {$\statelabel{3}$};
					\node[mystate] (I4) [below =of I3] {$\statelabel{4}$};
					\node[mystate] (I5) [below =of I4] {$\statelabel{5}$};

					% right of b nodes
					\node[mystate,pre] (I1a) [right =of I1]  {$\statelabel{1}$};
					\node[mystate, selectnode,pre] (I2a) [right =of I2] {$\statelabel{2}$};
					\node[mystate, selectnode,pre] (I3a) [right =of I3] {$\statelabel{3}$};
					\node[mystate] (I4a) [right =of I4] {$\statelabel{4}$};
					\node[mystate] (I5a) [right =of I5] {$\statelabel{5}$};
					
					% right of p nodes
					\node[mystate] (I1b) [right =of I1a]  {$\statelabel{1}$};
					\node[mystate, selectnode,post] (I2b) [right =of I2a] {$\statelabel{2}$};
					\node[mystate,post] (I3b) [right =of I3a] {$\statelabel{3}$};
					\node[mystate,post] (I4b) [right =of I4a] {$\statelabel{4}$};
					\node[mystate] (I5b) [right =of I5a] {$\statelabel{5}$};
					
					\node[rectangle,fill=black,scale=1.2] (choice) at ($(I4a)!0.5!(I4b)$) {};
					\node (spiral) at ($(I5a)!0.6!(I5b)$) {};

					\node[mystate] (I1c) [right =of I1b]  {$\statelabel{1}$};
					\node[mystate,selectnode,post] (I2c) [right =of I2b] {$\statelabel{2}$};
					\node[mystate,post] (I3c) [right =of I3b] {$\statelabel{3}$};
					\node[mystate,post] (I4c) [right =of I4b] {$\statelabel{4}$};
					\node[mystate] (I5c) [right =of I5b] {$\statelabel{5}$};

					% rightmost nodes
					\node[mystate, selectnode] (I1T) [right =of I1c]  {$\statelabel{1}$};
					\node[mystate, selectnode] (I2T) [right =of I2c] {$\statelabel{2}$};
					\node[mystate, selectnode] (I3T) [right =of I3c] {$\statelabel{3}$};
					\node[mystate, selectnode] (I4T) [right =of I4c] {$\statelabel{4}$};
					\node[mystate, selectnode] (I5T) [right =of I5c] {$\statelabel{5}$};

					\node[font=\LARGE] (T) at ($(I1c)!0.5!(I1T)$) [yshift=0.8cm] {$\top$};
					\node[font=\LARGE] (b) at ($(I1)!0.5!(I1a)$) [yshift=0.8cm] {$\pre$};
					\node[font=\LARGE] (p) at ($(I1a)!0.5!(I1b)$) [yshift=0.68cm] {$\program$};
					\node[font=\LARGE] (c) at ($(I1b)!0.5!(I1c)$) [yshift=0.8cm] {$\post$};
					
					\coordinate (spiralstart) at ($(spiral)+(-0.1,-0.62)$); % y coordinate is -0.002*30*30*0.35 + 0.1 (last plotted point times scale, plus 0.1 to match)
					\draw let \p1=(I5b) in 
						[scale=0.35,domain=0:30,variable=\t,smooth,samples=75,-,shift={(spiralstart)}] plot  ({\t r}: {-0.002*\t*\t});

					% b
					\path[nonselected] (I1) edge node {} (I1a);
					\path[selected] (I2) edge node {} (I2a);
					\path[selected] (I3) edge node {} (I3a);

					% p
					\path[nonselected] (I1a) edge node {} (I1b);
					\path[selected] (I2a) edge node {} (I2b);
					\path[selected] (I3a) edge node {} (I2b);
					\path[nonselected] (I4a) edge node {} (I4b);
					\path[nonselected] ($(I4a)!0.5!(I4b)$) edge node {} (I3b);
					\path[nonselected] (I5a) edge [-] node {} (spiral);
					
					% c
					\path[selected] (I2b) edge node {} (I2c);
					\path[nonselected] (I3b) edge node {} (I3c);
					\path[nonselected] (I4b) edge node {} (I4c);
					
					% T
					\path[nonselected] (I1c) edge node {} (I1T);
					\path[nonselected] (I1c) edge node {} (I2T);
					\path[nonselected] (I1c) edge node {} (I3T);
					\path[nonselected] (I1c) edge node {} (I4T);
					\path[nonselected] (I1c) edge node {} (I5T);

					\path[nonselected] (I3c) edge node {} (I1T);
					\path[nonselected] (I3c) edge node {} (I2T);
					\path[nonselected] (I3c) edge node {} (I3T);
					\path[nonselected] (I3c) edge node {} (I4T);
					\path[nonselected] (I3c) edge node {} (I5T);
					
					\path[nonselected] (I4c) edge node {} (I1T);
					\path[nonselected] (I4c) edge node {} (I2T);
					\path[nonselected] (I4c) edge node {} (I3T);
					\path[nonselected] (I4c) edge node {} (I4T);
					\path[nonselected] (I4c) edge node {} (I5T);
					
					\path[nonselected] (I5c) edge node {} (I1T);
					\path[nonselected] (I5c) edge node {} (I2T);
					\path[nonselected] (I5c) edge node {} (I3T);
					\path[nonselected] (I5c) edge node {} (I4T);
					\path[nonselected] (I5c) edge node {} (I5T);

					% drawn last so that they are on top	
					\path[selected] (I2c) edge node {} (I1T);
					\path[selected] (I2c) edge node {} (I2T);
					\path[selected] (I2c) edge node {} (I3T);
					\path[selected] (I2c) edge node {} (I4T);
					\path[selected] (I2c) edge node {} (I5T);
					
			\end{tikzpicture}
			\end{adjustbox}
			\end{center}
			\captionsetup{font=footnotesize,labelfont=footnotesize}
			\caption{ \footnotesize
				Adding $\top$ to the right-hand side of the term $\pre \program \post$ from \Cref{subfig:codomain-a}, effectively selecting the \emph{domain} of the underlying relation.}
			\label{subfig:codomain-c}
			\end{subfigure}%

	\caption{Using $\top$ to select the codomain \hyperref[subfig:codomain-b]{(b)} or the domain \hyperref[subfig:codomain-c]{(c)} of a \TopKAT term $\pre \program \post$.}
	\label{fig:codomain}
\end{figure}%
\begin{example}[$\top$ as (Co)domain Selector]
\label{ex:ka}
	Consider a precondition $\pre$, a program $\program$, and a postcondition $\post$ over a state space of five states $\Sigma = \{1,2,3,4,5\}$, where%
	\begin{align*}
		\pre & \eeq \{(1,1),(2,2),(3,3)\}, \\
		\program & \eeq \{(1,1),(2,2),(3,2),(4,3),(4,4)\}, \quad \text{and} \\
		\post & \eeq \{(2,2),(3,3),(4,4)\}.
	\end{align*}%
	As described above, the $\KAT$ term $\pre \program \post$ corresponds to the composition of the underlying relations, in this example being
	\[
		\pre \program \post \eeq \{(2,2),(3,2)\}.
	\]
	This is illustrated in \Cref{subfig:codomain-a}.
	The initial states satisfying $\pre$ are green, the final states satisfying $\post$ are blue.
	Nondeterministic choices in $\program$ are visualized by a square and divergence by a spiral.
	The composed relation $\pre \program \post$ is highlighted in red:
	We can see that the pairs in the relation correspond to the paths through the graph, originating in either initial state $2$ or $3$ and leading to final state $2$.
	Intuitively, these are the executions of $\program$ starting in $\pre$ and terminating in $\post$.

	%The effect of appending $\top$ to its left-hand (\Cref{subfig:codomain-b}) and its right-hand side (\Cref{subfig:codomain-c}).
%	In relational $\TopKAT$, the largest element $\top$ corresponds to the universal relation, relating every state to every other state.
	The effect of appending $\top$ (i.e.\ the universal relation) on the left of $\pre \program \post$ is visualized in \Cref{subfig:codomain-b} and yields the relation%
	\[
		\top \pre \program \post \eeq \{(1,2),(2,2),(3,2),(4,2),(5,2)\} = \{ (\sigma,2) \mid \sigma \in \Sigma \}. 
	\]
	Whereas in $\pre \program \post$ only initial states $2$ and $3$ were related to final state $2$, in $\top \pre \program \post$ \emph{all} initial states are related to $2$.
	Appending $\top$ on the left thus in some sense erases the information about the initial states and leaves only information about the final states -- or in other words: the \emph{codomain}.

	Dually, the effect of appending $\top$ on the right is visualized in \Cref{subfig:codomain-c} and yields the relation%
	\[
		\pre \program \post \top \eeq \{ (2,\tau),(3,\tau) \mid \tau \in \Sigma \},
	\]
	Whereas in $\pre \program \post$ initial states $2$ and $3$ were related only to final state $2$, in $\top \pre \program \post$ initial states $2$ and $3$ are related to \emph{all} final states.
	Appending $\top$ on the right thus erases the information about the final states and leaves only information about the initial states -- or in other words: the \emph{domain}.
	\qedtriangle%
\end{example}%
\noindent%
As we have seen, appending $\top$ to the right (left) of \emph{any} \KAT term amounts to selecting the (co)domain of the underlying relation.
In particular, for any two KAT terms $s$ and $t$, we have that 
\begin{align*}
	\top s &\eeq \top t \qiff \text{the codomains of $s$ and $t$ are equal,} \qquad \text{and} \\
	s \top &\eeq t \top \qiff \text{the domains of $s$ and $t$ are equal}.
\end{align*}%
For expressing incorrectness logic, an explicit comparison of codomains is necessary.
Hoare logic, on the other hand, can also be expressed without $\top$.
For details, we refer to \cite{zhang2022incorrectness}.

\paragraph{\TopKAT and Predicate Transformers}
There is a close relation between $\TopKAT$ and predicate transformers namely that some predicate transformers can be expressed as \TopKAT terms.
Let us express, for instance, $\asp{\program}{\pre}$ in \TopKAT.
Consider for this first the term
$
	\pre \program = \{ (\sigma,\tau) \mid \sigma \in \pre \text{ and } (\sigma,\tau) \in \program \}
$ describing all executions of $\program$ that start in $\pre$.
As described above, appending $\top$ on the left selects the codomain of that term, i.e.\ 
\[
	\top \pre \program \eeq \{ (\sigma',\tau) \mid \sigma' \in \Sigma \text{ and } \exists\, sigma \in \pre \text{ and } (\sigma,\tau) \in \program \}~.
\]%
Since $\asp{\program}{\pre}$ is precisely the set of states reachable by executing $\program$ on initial states satisfying $\pre$ and the codomain of $\pre \program$ is also precisely that set, we can morally equate $\asp{\program}{\pre}$ and $\top \pre \program$.

Similarly, appending $\top$ on the right selects the domain and $\program \post \top$ hence describes the set of states that can reach $\post$:%
\[
	\program \post \top \eeq \{ (\sigma,\tau') \mid \tau' \in \Sigma \text{ and } \exists\, tau \in \post \text{ and } (\sigma,\tau) \in \program \}.
\]
The domain of $\program \post$ is precisely the set $\awp{\program}{\post}$.
Notably, $\aspsymbol$ and $\awpsymbol$ are the only two transformers that are directly expressible in a $\TopKAT$ term.%

\begin{wrapfigure}[10]{R}{.4\textwidth}%
	\vspace*{-0.2\intextsep}
%	\begin{subfigure}{0.4\textwidth}  
		\begin{center}
			\begin{tikzpicture}[->,
				node distance=1cm and 3cm,
				thick,
				mystate/.style = {circle,inner sep=3pt,draw,font=\small}
				]

				\node[mystate] (I4) {$\sigma$};
				\node[mystate] (F4) [right =of I4] {$\tau$};

				\node[rectangle,fill=black,scale=0.3] (choice) at ($(I4)!0.5!(F4)$) {b};
				\node (spiral) at ($(I4)!0.6!(F4)$) {};
				
				\path (I4) edge node {} (F4);
				\path (I4) edge [-] node {} (spiral);
				
				\draw let \p1=(I4) in
				[scale=0.25,domain=0:30,variable=\t,smooth,samples=75,-,shift={(7.31,-1.75)}] plot  ({\t r}: {-0.002*\t*\t});
			\end{tikzpicture}
			
			\medskip
			
			$p_1 = \{(\sigma,\tau)\}$
%		\end{center}
%	\end{subfigure}
	%
	
	\bigskip\smallskip
	\hrule
	\bigskip\bigskip
	
%	\begin{subfigure}{0.4\textwidth}
%	\begin{center}
		\begin{tikzpicture}[->,
			node distance=1cm and 3cm,
			thick,
			mystate/.style = {circle,inner sep=3pt,draw,font=\small}
			]

			\node[mystate] (I4) {$\sigma$};
			\node[mystate] (F4) [right =of I4] {$\tau$};

			% \node[rectangle,fill=black,scale=0.3] (choice) at ($(I4)!0.5!(F4)$) {b};
			\node (spiral) at ($(I4)!0.6!(F4)$) {};
			
			% for spacing
			\draw let \p1=(I4) in
			[scale=0.25,draw = white,domain=0:30,variable=\t,smooth,samples=75,-,shift={(7.31,-1.75)}] plot  ({\t r}: {-0.002*\t*\t});

			\path (I4) edge node {} (F4);
			% \path (I4) edge [-] node {} (spiral);
			% 
		\end{tikzpicture}
		
		\vspace{-1\intextsep}
		$p_2 = \{(\sigma,\tau)\}$
	\end{center}
%	\end{subfigure}
%	\caption{Branching divergence, relationally.}
%	\label{fig:relational}
\end{wrapfigure}%
\paragraph{The Relational Perspective and Divergence.}
The purely relational perspective on programs introduces some limitations, particularly in reasoning about divergence.
An initial state from which computation always diverges is related to no final state, making it distinguishable from states that lead to some final state.
However, consider an initial state $\sigma$ from which the computation can nondeterministically diverge or terminate in $\tau$, as illustrated by the top program $p_1$ to our right.
The relation corresponding to $\program_1$ contains only the pair $(\sigma, \tau)$.

Now, consider the bottom program $p_2$ to our right which \emph{always} terminates in $\tau$ when started in $\sigma$.
The relational representation of $\program_2$ is the same as the one of $\program_1$, making these two programs relationally indistinguishable.
Consequently, information about branching divergence is inevitably lost in a relational perspective, as was already observed in \Cref{rem:diverge-semantics}.

\subsection{Expressing Program Logics in \TopKAT}
\label{ssec:kat-properties}

Let us now explore how to express the program properties / logics of \Cref{fig:taxonomy} in the equational system of \TopKAT.%
\paragraph{Hoare Logic for Partial Correctness}
Partial correctness requires that all computation started in~$\pre$ can only terminate in $\post$.
This can be expressed in \KAT, for example, by $\pre \program \negate{\post} = 0$ or $\pre \program \post = \pre \program$.
The first equation intuitively states that there is no computation starting in $\pre$ and terminating outside of~$\post$.
The latter compares all states in $\post$ in which $\program$ can terminate starting from $\pre$ to the states in which $\program$ can terminate from $\pre$ at all.
In general, the equations do not uniquely express partial correctness.
In fact, there are many more (in)equations characterizing partial correctness, in particular the \TopKAT equation \mbox{$\top \pre \program \post = \top \pre \program$}.
We choose the latter as it aligns well with the equations for the other logics.%
%Prepending $\top$ on both sides of the equation amounts to comparing their \emph{codomains}, i.e.\ their \emph{final states}.
%Thus, demonic partial correctness is also \KAT expressible.
%
%
\paragraph{Partial Incorrectness}
Partial \emph{in}correctness is contrapositive to partial correctness and thus immediately both \KAT and \TopKAT expressible by negating the conditions: \mbox{$\top \negate{\pre} \program \negate{\post} = \top \negate{\pre} \program$}.
An equivalent equation, however, is $\pre \program \post \top = \program \post \top$ and we will consider this latter one.%
\paragraph{Incorrectness Logic}
Incorrectness logic requires that all states in $\post$ be reachable from $\pre$.
This cannot be captured in a standard \KAT equation, as it involves reasoning about the codomain of relations.
Given that overcoming this was a key motivation behind \TopKAT, it is not surprising that incorrectness logic \emph{is} expressible in \TopKAT \cite{zhang2022incorrectness}, namely by $\top bpc = \top c$.
The left-hand side of the equation selects all final states in $\post$ reachable by executing~$\program$ on~$\pre$.
The right-hand side selects \emph{all} final states in~$\post$.%
\paragraph{Hoare Logic for Total Correctness and Angelic Partial Correctness}
For total correctness, all computation started in $\pre$ must terminate and do so in $\post$.
This requires reasoning about branching divergence, which, as argued earlier, is impossible.
Thus, standard total correctness is inexpressible both in \KAT~\cite{von2002kleene} and \TopKAT.
The same goes for angelic partial correctness.%
\paragraph{Lisbon logic}
Lisbon logic expresses that from all initial states in $\pre$, it is \emph{possible} for $\program$ to terminate in $\post$.
As this requires reasoning about the domain of a relation, Lisbon logic cannot be expressed in \KAT.
However, it \emph{is} indeed expressible in \TopKAT and to the best of our knowledge a \emph{novel result}:%
\green{%
\begin{theorem}[Expressibility of Lisbon Logic (Angelic Total Correctness) in \TopKAT]%
\[
	\hoare{\pre}{\program}{\post} \textnormal{ is a valid Lisbon triple / valid for angelic total correctness} \qqiff
		bpc\top = b\top
\]
\end{theorem}%
}%
\noindent%
%
%
% \begin{align*}
% 	\begin{array}{c}
% 		\hoare{\pre}{\program}{\post} \textnormal{ is valid for angelic total correctness}\\[.5em]
% 		\textnormal{iff}\\[.5em]
% 		bpc\top = b\top
% 	\end{array}
% \end{align*}%
% %
% \vspace*{-.25em}%
% \begin{narrowbluebox}%
% \abovedisplayskip=-.25em%
% %\belowdisplayskip=-1em%
% \begin{align*}
% 	\begin{array}{c}
% 		\hoare{\pre}{\program}{\post} \textnormal{ is valid for angelic total correctness}\\[.5em]
% 		\textnormal{iff}\\[.5em]
% 		bpc\top = b\top
% 	\end{array}
% \end{align*}%
% \normalsize%
% %
% \end{narrowbluebox}%
% \vspace*{-.25em}%
% %
%We presented two new logics in this paper:
% Demonic total incorrectness and angelic partial incorrectness.
\paragraph{Demonic Incorrectness}
The novel logic demonic incorrectness requires all states in $\post$ to be exclusively reachable from $\pre$.
We can divide this into two requirements, namely \emph{partial} incorrectness (\mbox{$\pre \program \post \top = \program \post \top$}) and reachability of all states in $\post$ ($\top \post = \top \program \post$).
These two cannot be expressed as a single equation.%
\paragraph{Angelic Partial Incorrectness}
The second novel logic, angelic partial incorrectness, requires that all states in post are either unreachable or can be reached from $\pre$.
This is expressible in \KAT by $\pre \program \post = \program \post$ and in \TopKAT by the equivalent equation $\top \pre \program \post = \top \program \post$, comparing all final states in~$\post$ reachable from $\pre$ to the reachable fragment of $\post$.
This is somewhat surprising (and asymmetric), since angelic partial \emph{correctness} cannot be expressed.
The root cause is, again, \Cref{obs:branching}.%
\begin{table}[t]
	\begin{tabular}{r | c | l}
		angelic total correctness \phantom{te}&\phantom{te} $\pre \subseteq \awp{\program}{\post}$ \phantom{te}&\phantom{te} $\pre \program \post \top = \pre \top$ \\	\hline	
		demonic partial correctness \phantom{te}&\phantom{te} $\pre \subseteq \dwlp{\program}{\post}$ \phantom{te}&\phantom{te} $\top \pre \program \post = \top \pre \program$ \phantom{te}\\ \hline
		angelic total incorrectness \phantom{te}&\phantom{te} $\post \subseteq \asp{\program}{\pre}$ \phantom{te}&\phantom{te}  $\top \pre \program \post = \top \post$  \\\hline
		angelic partial incorrectness \phantom{te}&\phantom{te} $\post \subseteq \aslp{\program}{\pre}$ \phantom{te}&\phantom{te}  $\top \pre \program \post = \top \program \post $ \phantom{te}\\\hline
		demonic partial incorrectness \phantom{te}&\phantom{te} $\post \subseteq \dslp{\program}{\pre}$ \phantom{te}&\phantom{te} $\pre \program \post \top=  \program \post \top$  \\	\hline
		???\phantom{te}&\phantom{te} ??? \phantom{te}&\phantom{te} $\pre \program \post \top= \pre \program \top$ \\ 
	\end{tabular}
	\caption{Overview of \TopKAT expressible program logics.}
	\label{tab:KAT}	
\end{table}
An overview of the five \TopKAT expressible program logics is given in \Cref{tab:KAT}.
Their contrapositives are of course also expressible by negation of all tests.
All equations are syntactically very similar.
There are more equations following this pattern, one of which is shown in the sixth row of \Cref{tab:KAT}.
% Starting in $\pre$, either all computation diverges or we can reach the post.
We will discuss this logic in the following section.
% Note that for programs without branching divergence, this is equivalent to angelic partial correctness, but in general it is not.
While other equations also syntactically fit into the scheme, we exclude them from further investigation as they would, for example, interpret the precondition $\pre$ over final states or do other semantically nonsensical things.
%Semantically, this does not make sense.
% There are two more equations following this pattern that are not a part of the table:
% $\top \pre \program \post = \top \pre$ and $ \pre \program \post \top = \post \top$.
% While they do syntactically fit into the scheme, we exclude them from further investigation as they semantically do not make sense in our setting.
% The former interprets the precondition $\pre$ over final states and the latter interprets the postcondition $\post$ over initial states, which both conflicts with what we want to express.
% A combination of $\pre \program \post$ with $\top$ is compared to the term without $\pre$ or $\post$, and sometimes $\program$.
% The six logics presented in \Cref{tab:KAT} are all that we can distinguish with such equations.

\subsection{The In-Between Logics}
\label{ssec:unnamed}

\Cref{tab:KAT} shows that five of the six basic \TopKAT equations directly correspond to a predicate transformer-based logic, as outlined in \Cref{sec:taxonomy}.
However, the sixth equation, $\pre \program \post \top = \pre \program \top$, stands apart.
This equation expresses that, for all states in $\pre$, the program either \emph{always} diverges, or there exists a terminating path to $\post$. %  if a terminating path exists, then all paths must terminate and end in 
Interestingly, in \Cref{ssec:backward-forward}, we discussed a set that characterizes precisely such program states: $\awp{\program}{\post} \cup \dwlp{\program}{\post}$.
Therefore, when the precondition $\pre$ is chosen to underapproximate this set, we obtain:
\[
	\pre \program \post \top= \pre \program \top \qqiff \pre \subseteq \awp{\program}{\post} \cup \dwlp{\program}{\post}
\]%
%For programs without branching divergence, this condition is equivalent to angelic partial correctness.
In \Cref{ssec:backward-forward}, we also examined the intersection $\awp{\program}{\post} \cap \dwlp{\program}{\post}$.
Underapproximating this intersection yields a logic which combines partial correctness Hoare logic and Lisbon logic. %that requires all states in $\pre$ to be able to reach $\post$, and not be able to terminate outside of $\post$.
% Underapproximating this intersection yields a logic that requires all states in $\pre$ to eventually reach $\post$, with at least one terminating path starting from $\pre$.
Notably, this is a special case of outcome logic \cite{zilberstein2023outcome}, when restricted to traditional assertions, i.e.\ predicates, and instantiated to the powerset monad.
% For programs without branching divergence, this is equivalent to demonic total correctness.
%
%
% When we underapproximate this union, it also makes sense to look at an overapproximation.
% Note that this is contrapositive to an underapproximation of the intersection, i.e.\
% \[
% 	\pre \subseteq \awp{\program}{\post} \cap \dwlp{\program}{\post}.
% \]
% From all states in the precondition, this requires all paths that terminate to do so in the postcondition, and there must exist at least one such path.
% Equivalently, from all states in the precondition, there must be a path terminating in $\post$ and no path terminating in $\negate{\post}$.
% The reason why this is not equal to demonic total correctness is again branching divergence.
Similar to the equations for demonic total incorrectness, we cannot give a single \TopKAT equation for this, but a \emph{system} of two:
$
	\pre \program \post \neq 0 \; \text{and} \; \pre \program \negate{\post} = 0.
$

% !TEX root = ../main.tex

\section{A Taxonomy of Program Logics: Revisited}
\label{sec:taxonomy-rev}

\begin{figure}[t]
	\begin{adjustbox}{max width=\textwidth}
		\begin{tikzpicture}[ %[mynode/.style={minimum width=2.7cm},node distance=2.5cm and 3.2cm]
				align=center,	
				mynode/.style={minimum width=2.7cm},node distance=1.3cm and 1.3cm,%2.8cm and 3.3cm,
				mylabel/.style n args={1}{label={[label distance=-4mm,font=\tiny,text=DodgerBlue3!60,below]#1}},
				mylabelkatone/.style n args={1}{label={[label distance=-9.5mm,font=\tiny,text=DodgerBlue3!60,below]#1}},
				mylabelkatoneshift/.style n args={1}{label={[label distance=-9.5mm,font=\tiny,text=DodgerBlue3!60,below,xshift=1pt]#1}},
				mylabelkattwo/.style n args={1}{label={[label distance=-13.5mm,font=\tiny,text=DodgerBlue3!60,below]#1}}
				]

			% Upper left
			\node[mynode, mylabelkatone={Lisbon logic (angelic tot.\ corr.)}](awpLB) {
				$\pre \subseteq \awp{\program}{\post}$ \\
				\katexpr{$\pre \program \post \top = \pre \top$} \vspaace };
			\node[mynode, mylabel={Hoare logic (total correctness)},below=of awpLB](dwpLB) {
				$\pre \subseteq \dwp{\program}{\post}$  };
			\node[mynode,right=of awpLB, mylabel={angelic partial correctness}](awlpLB) {
				$\pre \subseteq \awlp{\program}{\post}$ };
			\node[mynode,right=of dwpLB, mylabelkattwo={Hoare logic (partial correctness)}] (dwlpLB) {
				$\pre \subseteq \dwlp{\program}{\post}$\\
				\katexpr{$\top \pre \program \post = \top \pre \program $} \vspaace \\
				\katexpr{$\negate{\pre} \program \negate{\post} \top = \program \negate{\post} \top$}};

			% Upper right
			\node[mynode,right= 0.9cm of awlpLB]  (dwpUB) {$\dwp{\program}{\post} \subseteq \pre$};
			\node[mynode,right=of dwpUB]  (dwlpUB) {
				$\dwlp{\program}{\post} \subseteq \pre$ \\
				\katexpr{$\negate{\pre} \program \negate{\post} \top = \negate{\pre} \top$ \vspaace}};
			\node[mynode,below=of dwlpUB]  (awlpUB) {$\awlp{\program}{\post} \subseteq \pre$};			
			\node[mynode,left=of awlpUB,mylabelkattwo={partial incorrectness}] (awpUB) {
				$\awp{\program}{\post} \subseteq \pre$ \\
				\katexpr{$\top \negate{\pre} \program \negate{\post} = \top \negate{\pre} \program $} \vspaace \\
				\katexpr{$\pre \program \post \top = \program \post \top$}};
			
			% Lower left
			\node[mynode,below= 1cm of dwlpLB,mylabelkattwo={Hoare logic (partial correctness)}] (aspUB) {
				$\asp{\program}{\pre} \subseteq \post$ \\
				\katexpr{$\top \pre \program \post = \top \pre \program $} \vspaace \\
				\katexpr{$\negate{\pre} \program \negate{\post} \top = \program \negate{\post} \top$}};
			\node[mynode,left=of aspUB] (aslpUB) {$\aslp{\program}{\pre} \subseteq \post$};
			\node[mynode,below=of aslpUB]  (dslpUB) {
				$\dslp{\program}{\pre} \subseteq \post$ \\
				\katexpr{$\top \negate{\pre} \program \negate{\post} = \top \negate{\post}$}\vspaace};
			\node[mynode,right=of dslpUB] (dspUB) {
				$\dsp{\program}{\pre} \subseteq \post$ \\
				\katexpr{$\top \negate{\pre} \program \negate{\post} = \top \program \negate{\post}$}\vspaace};

			% Lower right
			\node[mynode,below= 1cm of awpUB, mylabelkattwo={partial incorrectness}] (dslpLB) {
				$\post \subseteq \dslp{\program}{\pre}$ \\
				\katexpr{$\top \negate{\pre} \program \negate{\post} = \top \negate{\pre} \program $}\vspaace \\
				\katexpr{$\pre \program \post \top = \program \post \top$}};
			\node[mynode,right=of dslpLB, mylabel={demonic incorrectness}] (dspLB) {$\post \subseteq \dsp{\program}{\pre}$};
			\node[mynode,below=of dspLB, mylabelkatone={incorrectness logic}] (aspLB) {
				$\post \subseteq \asp{\program}{\pre}$ \\
				\katexpr{$\top \pre \program \post = \top \post$}\vspaace};
			\node[mynode,left=of aspLB, mylabelkatone={angelic partial incorrectness}] (aslpLB) {
					$\post \subseteq \aslp{\program}{\pre}$ \\
					\katexpr{$\top \pre \program \post = \top \program \post$}\vspaace};
				
			% Contrapositives
			\draw[contrapos] (awpUB) edge[<->] (dwlpLB);
			\draw[contrapos] (dwpUB) edge[<->] (awlpLB);
			\draw[contrapos,bend left=12] (awpLB) edge[<->] (dwlpUB); %-- ($awpLB!0.5!awlpLB$) -- 
			\draw[contrapos,bend left=17] (dwpLB) edge[<->] (awlpUB);
			
			\draw[contrapos] (aspUB) edge[<->] (dslpLB);
			\draw[contrapos] (dspUB) edge[<->] (aslpLB);
			\draw[contrapos,bend left=17] (aspLB) edge[<->] (dslpUB);
			\draw[contrapos,bend left=21] (dspLB) edge[<->,shorten <=9pt] (aslpUB);
			
			% Galois
			\path (dwlpLB) edge[galoisshort] node {} (aspUB);
			\path (awpUB) edge[galoisshort] node {} (dslpLB);
			
			% Implications
			\path (dwpLB) edge[implicationupkat] node {} (awpLB);
			\path (dwpLB) edge[implication] node {} (dwlpLB);
			\path (awpLB) edge[implication] node {} (awlpLB);
			\path (dwlpLB) edge[implicationup] node {} (awlpLB);
			
			% mid stuff top left
			
			\path (awpUB) edge[implication] node {} (dwpUB);
			\path (awlpUB) edge[implication] node {} (awpUB);
			\path (awlpUB) edge[implication] node {} (dwlpUB);
			\path (dwlpUB) edge[implication] node {} (dwpUB);
			
			\path (dspLB) edge[implicationdown] node {} (aspLB);
			\path (dspLB) edge[implication] node {} (dslpLB);
			\path (aspLB) edge[implication] node {} (aslpLB);
			\path (dslpLB) edge[implicationdownkat] node {} (aslpLB);
			
			\path (aspUB) edge[implicationdownkat] node {} (dspUB);
			\path (aslpUB) edge[implication] node {} (aspUB);
			\path (aslpUB) edge[implication] node {} (dslpUB);
			\path (dslpUB) edge[implication] node {} (dspUB);

			% KAT expressible:
			\def\roundedcorn{0.1cm}
			\def\innerxsepp{-2.5pt}
			\def\innerysepp{-1pt}
			\node (anim1) [katexpr,fit = {(aspLB)}] {}; 
			\node (anim1) [katexpr,fit = {(dslpLB)}] {}; 
			\node (anim1) [katexpr,fit = {(dslpUB)}] {}; 
			\node (anim1) [katexpr,fit = {(aspUB)}] {};  
			\node (anim1) [katexpr,fit = {(dspUB)}] {}; 
			\node (anim1) [katexpr,fit = {(aslpLB)}] {}; 
			
			\node (anim1) [katexpr,fit = {(awpUB)}] {};
			\node (anim1) [katexpr,fit = {(dwlpLB)}] {}; 
			\node (anim1) [katexpr,fit = {(awpLB)}] {}; 
			\node (anim1) [katexpr,fit = {(dwlpUB)}] {};
			
		\end{tikzpicture}
	\end{adjustbox}
	\caption{
		The taxonomy presented in \Cref{fig:taxonomy} with corresponding \TopKAT equations, if existing, in red.
%		For better readability, the implications in the lower half are not drawn; since the added logics are equivalent to existing ones.
	}
	\label{fig:taxonomy-kat}
\end{figure}
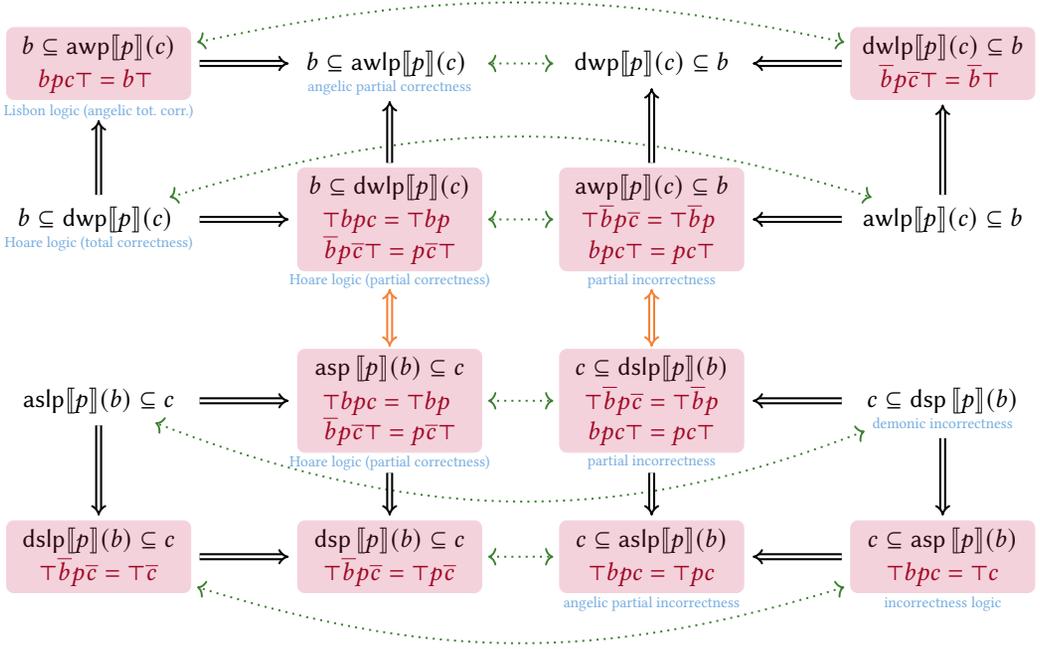
In \Cref{sec:taxonomy}, we presented a taxonomy of predicate transformer logics.
Taking that picture and highlighting the \TopKAT expressible logics in red yields \Cref{fig:taxonomy-kat}, \emph{revealing an asymmetry} which was invisible before:
The pattern of \TopKAT expressible logics in the upper half is not symmetric to the pattern in the lower half: 
%$\pre \subseteq \awlp{\program}{\post}$ cannot be expressed in \TopKAT, but $\dsp{\program}{\pre} \subseteq \post$ can.
This asymmetry is attributed to the fact that we cannot capture branching divergence, whereas the dual concept in forward analyses, confluence of unreachability, does not exist (see \Cref{obs:branching}).

\begin{figure}[t]
	\begin{adjustbox}{max width=\textwidth}
		\begin{tikzpicture}[ %[mynode/.style={minimum width=2.7cm},node distance=2.5cm and 3.2cm]
				align=center,	
				mynode/.style={minimum width=2.7cm},node distance=2.8cm and 3.3cm,
				mylabel/.style n args={1}{label={[label distance=-4mm,font=\tiny,text=DodgerBlue3!60,below]#1}},
				mylabelkatone/.style n args={1}{label={[label distance=-9.5mm,font=\tiny,text=DodgerBlue3!60,below]#1}},
				mylabelkatoneshift/.style n args={1}{label={[label distance=-9.5mm,font=\tiny,text=DodgerBlue3!60,below,xshift=1pt]#1}},
				mylabelkattwo/.style n args={1}{label={[label distance=-13.5mm,font=\tiny,text=DodgerBlue3!60,below]#1}}
				]
			% Upper left
			\node[mynode, mylabelkatone={Lisbon logic (angelic tot.\ corr.)}](awpLB) {
				$\pre \subseteq \awp{\program}{\post}$ \\
				\katexpr{$\pre \program \post \top = \pre \top$} \vspaace };
			\node[mynode, mylabel={Hoare logic (total correctness)},below=of awpLB](dwpLB) {
				$\pre \subseteq \dwp{\program}{\post}$  };
			\node[mynode,right=of awpLB, mylabel={angelic partial correctness}](awlpLB) {
				$\pre \subseteq \awlp{\program}{\post}$ };
			\node[mynode,right=of dwpLB, mylabelkattwo={Hoare logic (partial correctness)}] (dwlpLB) {
				$\pre \subseteq \dwlp{\program}{\post}$\\
				\katexpr{$\top \pre \program \post = \top \pre \program $} \vspaace \\
				\katexpr{$\negate{\pre} \program \negate{\post} \top = \program \negate{\post} \top$}};
			\node[mynode,below right= 1.7cm and -1cm of awpLB](intLB) {
				$\pre \subseteq \awpCAPdwlp{\program}{\post}$  };
			\node[mynode,below right= 0.3cm and 0.3cm of awpLB, mylabelkatoneshift={in-between logic}](uniLB) {
				$\pre \subseteq \awpCUPdwlp{\program}{\post}$ \\
				\katexpr{$\pre \program \post\top  = \pre \program \top $ \vspaace}};

			% Upper right
			\node[mynode,right= 0.9cm of awlpLB]  (dwpUB) {$\dwp{\program}{\post} \subseteq \pre$};
			\node[mynode,right=of dwpUB]  (dwlpUB) {
				$\dwlp{\program}{\post} \subseteq \pre$ \\
				\katexpr{$\negate{\pre} \program \negate{\post} \top = \negate{\pre} \top$ \vspaace}};
			\node[mynode,below=of dwlpUB]  (awlpUB) {$\awlp{\program}{\post} \subseteq \pre$};			
			\node[mynode,left=of awlpUB,mylabelkattwo={partial incorrectness}] (awpUB) {
				$\awp{\program}{\post} \subseteq \pre$ \\
				\katexpr{$\top \negate{\pre} \program \negate{\post} = \top \negate{\pre} \program $} \vspaace \\
				\katexpr{$\pre \program \post \top = \program \post \top$}};
			\node[mynode,below left= 1.7cm and -1cm of dwlpUB](uniUB) {$\awpCUPdwlp{\program}{\post} \subseteq \pre$};
			\node[mynode,below left= 0.3cm and 0.3cm of dwlpUB](intUB) {
				$\awpCAPdwlp{\program}{\post} \subseteq \pre$ \\
				\katexpr{$\negate{\pre} \program \negate{\post} \top = \negate{\pre} \program \top$} \vspaace};

			% Lower left
			\node[mynode,below= 1cm of dwlpLB,mylabelkattwo={Hoare logic (partial correctness)}] (aspUB) {
				$\asp{\program}{\pre} \subseteq \post$ \\
				\katexpr{$\top \pre \program \post = \top \pre \program $} \vspaace \\
				\katexpr{$\negate{\pre} \program \negate{\post} \top = \program \negate{\post} \top$}};
			\node[mynode,left=of aspUB] (aslpUB) {$\aslp{\program}{\pre} \subseteq \post$};
			\node[mynode,below=of aslpUB]  (dslpUB) {
				$\dslp{\program}{\pre} \subseteq \post$ \\
				\katexpr{$\top \negate{\pre} \program \negate{\post} = \top \negate{\post}$}\vspaace};
			\node[mynode,right=of dslpUB] (dspUB) {
				$\dsp{\program}{\pre} \subseteq \post$ \\
				\katexpr{$\top \negate{\pre} \program \negate{\post} = \top \program \negate{\post}$}\vspaace};
			\node[mynode,below left= 1.0cm and -1cm of aspUB](intUBsp) {
				$\aspCAPdslp{\program}{\pre} \subseteq \post$ \\
				\katexpr{$\top \negate{\pre} \program \negate{\post} = \top \program \negate{\post}$}\vspaace};
			\node[mynode,below left= 0.2cm and 0.4cm of aspUB](uniUBsp) {$\aspCUPdslp{\program}{\pre} \subseteq \post$};

			% Lower right
			\node[mynode,below= 1cm of awpUB, mylabelkattwo={partial incorrectness}] (dslpLB) {
				$\post \subseteq \dslp{\program}{\pre}$ \\
				\katexpr{$\top \negate{\pre} \program \negate{\post} = \top \negate{\pre} \program $}\vspaace \\
				\katexpr{$\pre \program \post \top = \program \post \top$}};
			\node[mynode,right=of dslpLB, mylabel={demonic incorrectness}] (dspLB) {$\post \subseteq \dsp{\program}{\pre}$};
			\node[mynode,below=of dspLB, mylabelkatone={incorrectness logic}] (aspLB) {
				$\post \subseteq \asp{\program}{\pre}$ \\
				\katexpr{$\top \pre \program \post = \top \post$}\vspaace};
			\node[mynode,left=of aspLB, mylabelkatone={angelic partial incorrectness}] (aslpLB) {
					$\post \subseteq \aslp{\program}{\pre}$ \\
					\katexpr{$\top \pre \program \post = \top \program \post$}\vspaace};
				\node[mynode,below right= 1cm and -1cm of dslpLB](uniLBsp) {
				$\post \subseteq \aspCUPdslp{\program}{\pre}$ \\
				\katexpr{$\top \pre \program \post = \top \program \post$}\vspaace};
			\node[mynode,below right= 0.2cm and 0.4cm of dslpLB](intLBsp) {$\post \subseteq \aspCAPdslp{\program}{\pre}$};

			% Contrapositives
			\draw[contrapos] (awpUB) edge[<->] (dwlpLB);
			\draw[contrapos] (dwpUB) edge[<->] (awlpLB);
			\draw[contrapos,bend left=10] (awpLB) edge[<->] (dwlpUB);
			\draw[contrapos,bend left=10] (dwpLB) edge[<->] (awlpUB);
			\draw[contrapos] (uniLB) edge[<->] (intUB);
			\draw[contrapos] (intLB) edge[<->] (uniUB);
			\draw[contrapos] (uniLBsp) edge[<->] (intUBsp);
			\draw[contrapos] (intLBsp) edge[<->] (uniUBsp);
			
			\draw[contrapos] (aspUB) edge[<->] (dslpLB);
			\draw[contrapos,bend left=13] (aspLB) edge[<->] (dslpUB);
			\draw[contrapos] (dspUB) edge[<->] (aslpLB);
			\draw[contrapos,bend left=13] (dspLB) edge[<->] (aslpUB);
			
			% Galois
			\path (dwlpLB) edge[galoisshort] node {} (aspUB);
			\path (awpUB) edge[galoisshort] node {} (dslpLB);
			
			% Implications
			\path (dwpLB) edge[implicationupkat] node {} (awpLB);
			\path (dwpLB) edge[implication] node {} (dwlpLB);
			\path (awpLB) edge[implication] node {} (awlpLB);
			\path (dwlpLB) edge[implicationup] node {} (awlpLB);
			\path (intLB) edge[implicationupmid] node {} (awpLB);
			\path (intLB) edge[implication] node {} (dwlpLB);
			\path (awpLB) edge[implication] node {} (uniLB);
			\path (dwlpLB) edge[implication] node {} (uniLB.325);
			% mid stuff top left
			\path (dwpLB) edge[implication,ultra thick] node {} (intLB.197);
			\path (intLB.71) edge[implication] node {} (uniLB.219);
			\path (uniLB.25) edge[implicationupmid,ultra thick] node {} (awlpLB);
			
			\path (awpUB) edge[implication] node {} (dwpUB);
			\path (awlpUB) edge[implication] node {} (awpUB);
			\path (awlpUB) edge[implication] node {} (dwlpUB);
			\path (dwlpUB) edge[implication] node {} (dwpUB);
			\path (uniUB) edge[implication] node {} (awpUB);
			\path (uniUB) edge[implicationshort] node {} (dwlpUB);
			\path (awpUB) edge[implicationshort] node {} (intUB);
			\path (dwlpUB) edge[implication] node {} (intUB);
			% mid
			\path (awlpUB) edge[implication,ultra thick] node {} (uniUB.343);
			\path (uniUB.109) edge[implication] node {} (intUB.321);
			\path (intUB.155) edge[implication,ultra thick] node {} (dwpUB);
			
			\path (dspLB) edge[implicationdown] node {} (aspLB);
			\path (dspLB) edge[implication] node {} (dslpLB);
			\path (aspLB) edge[implication] node {} (aslpLB);
			\path (dslpLB) edge[implicationdownkat] node {} (aslpLB);
			\path (dspLB) edge[equivalencemid,ultra thick] node {} (intLBsp);
			\path (uniLBsp) edge[equivalencekat,ultra thick] node {} (aslpLB);
			
			\path (aspUB) edge[implicationdownkat] node {} (dspUB);
			\path (aslpUB) edge[implication] node {} (aspUB);
			\path (aslpUB) edge[implication] node {} (dslpUB);
			\path (dslpUB) edge[implication] node {} (dspUB);
			\path (dspUB) edge[equivalencekat,ultra thick] node {} (intUBsp);
			\path (uniUBsp) edge[equivalence,ultra thick] node {} (aslpUB);

			% KAT expressible:
			\def\roundedcorn{0.1cm}
			\def\innerxsepp{-2.5pt}
			\def\innerysepp{-1pt}
			\node (anim1) [katexpr,fit = {(aspLB)}] {}; 
			\node (anim1) [katexpr,fit = {(dslpLB)}] {}; 
			\node (anim1) [katexpr,fit = {(dslpUB)}] {}; 
			\node (anim1) [katexpr,fit = {(aspUB)}] {};  
			\node (anim1) [katexpr,fit = {(dspUB)}] {}; 
			\node (anim1) [katexpr,fit = {(aslpLB)}] {}; 
			
			\node (anim1) [katexpr,fit = {(awpUB)}] {};
			\node (anim1) [katexpr,fit = {(dwlpLB)}] {}; 
			\node (anim1) [katexpr,fit = {(awpLB)}] {}; 
			\node (anim1) [katexpr,fit = {(dwlpUB)}] {};
			
			\node (anim1) [katexpr,fit = {(uniLB)}] {};
			\node (anim1) [katexpr,fit = {(intUB)}] {};
			
			\node (anim1) [katexpr,fit = {(intUBsp)}] {};
			\node (anim1) [katexpr,fit = {(uniLBsp)}] {};
			
		\end{tikzpicture}
	\end{adjustbox}
	\caption{
		The taxonomy including union and intersection logics and the \TopKAT equations, if existing, in red.
%		For better readability, the implications in the lower half are not drawn; since the added logics are equivalent to existing ones.
	}
	\label{fig:taxonomy-full}
\end{figure}
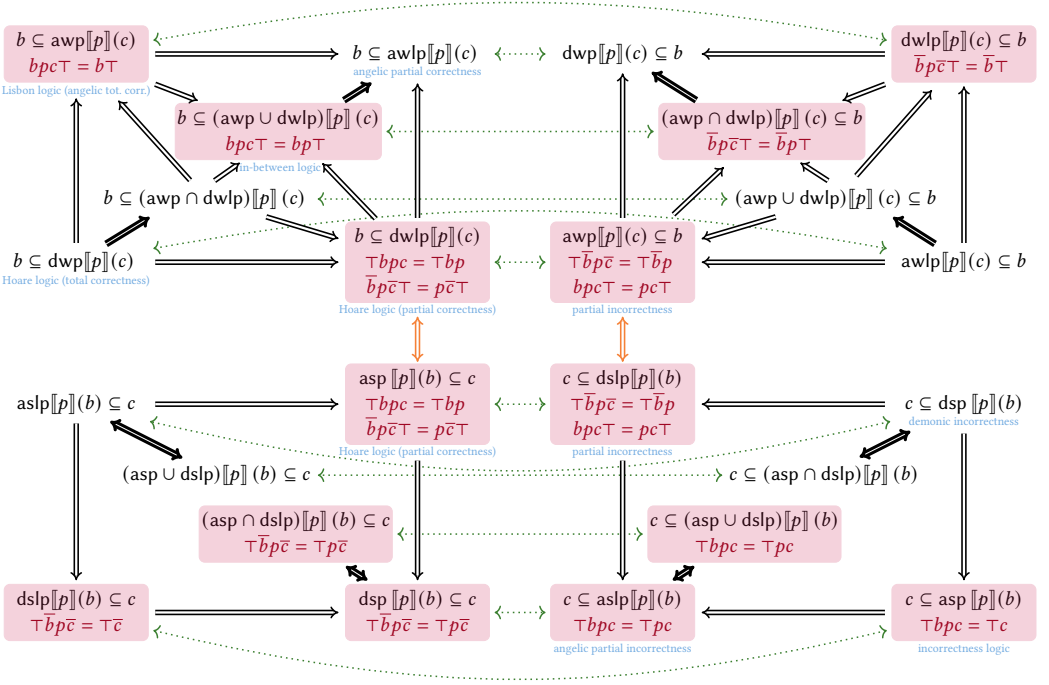

\subsection{Adding the In-Between Logics}
\label{ssec:adding_unnamed}

Missing from \Cref{fig:taxonomy-kat} are the logics that arise from unions and intersections of predicate transformers, as discussed in \Cref{ssec:unnamed}.
In \Cref{fig:taxonomy-full}, we include these logics at the center of each quadrant. %, reflecting their closeness to the existing ones.
%For example, Outcome Logic — specifically its restriction discussed in \Cref{ssec:unnamed} — is positioned near Hoare logic for total correctness.
For $\wpsymbol$ transformers, this yields indeed new logics.
%Additionally, we incorporate the union and intersection of the corresponding \spsymbol\ transformers.
For $\spsymbol$ transformers, the resulting logics are equivalent to existing ones, see \Cref{ssec:backward-forward}.
Notably, neither intersections nor unions of any $\spsymbol$ transformers generate new logics, nor do any of the \wpsymbol\ transformers, apart from \awpsymbol\ and \dwlpsymbol.
In the complete picture, the source of the earlier-mentioned asymmetry becomes clearer:
The bold connectives in the top half represent only implications, while those in the lower half represent equivalences, a distinction once again driven by \Cref{obs:branching}.

\subsection{Absence of Branching Divergence}
\label{ssec:assumption-2}

In \Cref{ssec:assumptions}, we examined how assumptions on $\program$ let our taxonomy partially collapse.
With the addition of union and intersection logics, we can now consider another assumption:
The absence of branching divergence.
If we assume that branching divergence is not present, the taxonomy becomes fully dual, as the bold implications in the upper half of the diagram turn into equivalences, 
% The underapproximation of the union becomes equivalent to angelic partial correctness, and the restricted version of Outcome Logic aligns with total correctness Hoare logic.
see \Cref{fig:taxonomy-branching}, where the logics added in the center collapse into the corners.
Meanwhile, the bottom half remains unchanged, as the logics within the blue dashed lines were already equivalent.

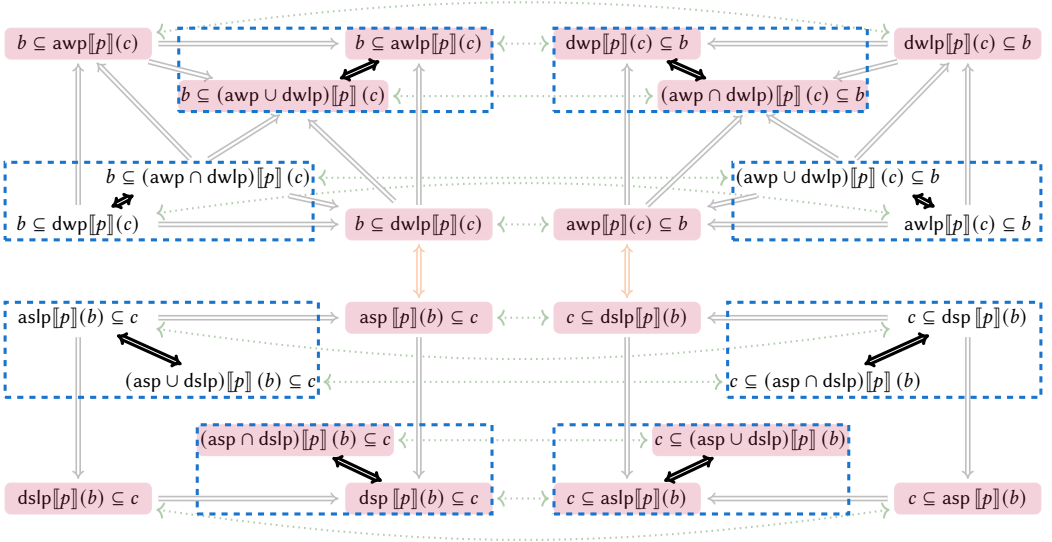
\begin{figure}[t]
	\begin{adjustbox}{max width=\textwidth}
		\begin{tikzpicture}
			[ %[mynode/.style={minimum width=2.7cm},node distance=2.5cm and 3.2cm]
				mynode/.style={minimum width=2.7cm},node distance=2.5cm and 3.2cm,
				mylabel/.style n args={1}{label={[label distance=-4mm,font=\tiny,text=DodgerBlue3!60,below]#1}},
				mylabelkat/.style n args={1}{label={[label distance=-5mm,font=\tiny,text=DodgerBlue3!60,below]#1}}
			]
			% Upper left
			\node[mynode](awpLB) {
				$\pre \subseteq \awp{\program}{\post}$ };
			\node[mynode,below=of awpLB](dwpLB) {
				$\pre \subseteq \dwp{\program}{\post}$  };
			\node[mynode,right=of awpLB](awlpLB) {
				$\pre \subseteq \awlp{\program}{\post}$ };
			\node[mynode,right=of dwpLB] (dwlpLB) {
				$\pre \subseteq \dwlp{\program}{\post}$};
			\node[mynode,below right= 1.7cm and -1cm of awpLB](intLB) {
				$\pre \subseteq \awpCAPdwlp{\program}{\post}$  };
			\node[mynode,below right= 0.3cm and 0.3cm of awpLB](uniLB) {
				$\pre \subseteq \awpCUPdwlp{\program}{\post}$ };

			% Upper right
			\node[mynode,right= 0.9cm of awlpLB]  (dwpUB) {$\dwp{\program}{\post} \subseteq \pre$};
			\node[mynode,right=of dwpUB]  (dwlpUB) {
				$\dwlp{\program}{\post} \subseteq \pre$ };
			\node[mynode,below=of dwlpUB]  (awlpUB) {$\awlp{\program}{\post} \subseteq \pre$};			
			\node[mynode,left=of awlpUB] (awpUB) {
				$\awp{\program}{\post} \subseteq \pre$};
			\node[mynode,below left= 1.7cm and -1cm of dwlpUB](uniUB) {$\awpCUPdwlp{\program}{\post} \subseteq \pre$};
			\node[mynode,below left= 0.3cm and 0.3cm of dwlpUB](intUB) {
				$\awpCAPdwlp{\program}{\post} \subseteq \pre$ };

			% Lower left
			\node[mynode,below= 1cm of dwlpLB] (aspUB) {$\asp{\program}{\pre} \subseteq \post$};
			\node[mynode,left=of aspUB] (aslpUB) {$\aslp{\program}{\pre} \subseteq \post$};
			\node[mynode,below=of aslpUB]  (dslpUB) {$\dslp{\program}{\pre} \subseteq \post$};
			\node[mynode,below=of aspUB] (dspUB) {$\dsp{\program}{\pre} \subseteq \post$};
			\node[mynode,below left= 1.5cm and -1cm of aspUB](intUBsp) {$\aspCAPdslp{\program}{\pre} \subseteq \post$};
			\node[mynode,below left= 0.5cm and 0.3cm of aspUB](uniUBsp) {$\aspCUPdslp{\program}{\pre} \subseteq \post$};
			
			% Lower right
			\node[mynode,below= 1cm of awpUB] (dslpLB) {$\post \subseteq \dslp{\program}{\pre}$};
			\node[mynode,right=of dslpLB] (dspLB) {$\post \subseteq \dsp{\program}{\pre}$};
			\node[mynode,below=of dslpLB] (aslpLB) {$\post \subseteq \aslp{\program}{\pre}$};
			\node[mynode,below=of dspLB] (aspLB) {$\post \subseteq \asp{\program}{\pre}$};
			\node[mynode,below right= 1.5cm and -1cm of dslpLB](uniLBsp) {$\post \subseteq \aspCUPdslp{\program}{\pre}$};
			\node[mynode,below right= 0.5cm and 0.3cm of dslpLB](intLBsp) {$\post \subseteq \aspCAPdslp{\program}{\pre}$};
			
			% Contrapositives
			\draw[weakcontrapos] (awpUB) edge[<->] (dwlpLB);
			\draw[weakcontrapos] (dwpUB) edge[<->] (awlpLB);
			\draw[weakcontrapos,bend left=8] (awpLB) edge[<->] (dwlpUB);
			\draw[weakcontrapos,bend left=8] (dwpLB) edge[<->] (awlpUB);
			\draw[weakcontrapos] (uniLB) edge[<->] (intUB);
			\draw[weakcontrapos] (intLB) edge[<->] (uniUB);
			\draw[weakcontrapos] (uniLBsp) edge[<->] (intUBsp);
			\draw[weakcontrapos] (intLBsp) edge[<->] (uniUBsp);
			
			\draw[weakcontrapos] (aspUB) edge[<->] (dslpLB);
			\draw[weakcontrapos,bend left=8] (aspLB) edge[<->] (dslpUB);
			\draw[weakcontrapos] (dspUB) edge[<->] (aslpLB);
			\draw[weakcontrapos,bend left=8] (dspLB) edge[<->] (aslpUB);
			
			% Galois
			\path (aspUB) edge[weakgalois] node {} (dwlpLB);
			\path (awpUB) edge[weakgalois] node {} (dslpLB);
			
			% Implications
			\path (dwpLB) edge[weakimplication] node {} (awpLB);
			\path (dwpLB) edge[weakimplication] node {} (dwlpLB);
			\path (awpLB) edge[weakimplication] node {} (awlpLB);
			\path (dwlpLB) edge[weakimplication] node {} (awlpLB);
			\path (intLB) edge[weakimplication] node {} (awpLB);
			\path (intLB) edge[weakimplication] node {} (dwlpLB);
			\path (awpLB) edge[weakimplication] node {} (uniLB);
			\path (dwlpLB) edge[implicationshort,lightgray] node {} (uniLB);
			% mid stuff top left
			\path (dwpLB) edge[equivalence,ultra thick] node {} (intLB.193);
			\path (intLB.95) edge[weakimplication] node {} (uniLB.265);
			\path (uniLB.17) edge[equivalence,ultra thick] node {} (awlpLB);
			
			\path (awpUB) edge[weakimplication] node {} (dwpUB);
			\path (awlpUB) edge[weakimplication] node {} (awpUB);
			\path (awlpUB) edge[weakimplication] node {} (dwlpUB);
			\path (dwlpUB) edge[weakimplication] node {} (dwpUB);
			\path (uniUB) edge[weakimplication] node {} (awpUB);
			\path (uniUB) edge[weakimplication] node {} (dwlpUB);
			\path (awpUB) edge[weakimplication] node {} (intUB);
			\path (dwlpUB) edge[weakimplication] node {} (intUB);
			% mid
			\path (awlpUB) edge[equivalence,ultra thick] node {} (uniUB.347);
			\path (uniUB.85) edge[weakimplication] node {} (intUB.275);
			\path (intUB.163) edge[equivalence,ultra thick] node {} (dwpUB);
			
			\path (dspLB) edge[weakimplication] node {} (aspLB);
			\path (dspLB) edge[weakimplication] node {} (dslpLB);
			\path (aspLB) edge[weakimplication] node {} (aslpLB);
			\path (dslpLB) edge[weakimplication] node {} (aslpLB);
			\path (dspLB) edge[equivalence,ultra thick] node {} (intLBsp);
			\path (uniLBsp) edge[equivalence,ultra thick] node {} (aslpLB);
			
			\path (aspUB) edge[weakimplication] node {} (dspUB);
			\path (aslpUB) edge[weakimplication] node {} (aspUB);
			\path (aslpUB) edge[weakimplication] node {} (dslpUB);
			\path (dslpUB) edge[weakimplication] node {} (dspUB);
			\path (dspUB) edge[equivalence,ultra thick] node {} (intUBsp);
			\path (uniUBsp) edge[equivalence,ultra thick] node {} (aslpUB);

			% KAT expressible:
			\def\roundedcorn{0.1cm}
			\def\innerxsepp{-2.5pt}
			\def\innerysepp{-1pt}
			\node (anim1) [katexpr,fit = {(aspLB)}] {}; 
			\node (anim1) [katexpr,fit = {(dslpLB)}] {}; 
			\node (anim1) [katexpr,fit = {(dslpUB)}] {}; 
			\node (anim1) [katexpr,fit = {(aspUB)}] {};  
			\node (anim1) [katexpr,fit = {(dspUB)}] {}; 
			\node (anim1) [katexpr,fit = {(aslpLB)}] {}; 
			
			\node (anim1) [katexpr,fit = {(awpUB)}] {};
			\node (anim1) [katexpr,fit = {(dwlpLB)}] {}; 
			\node (anim1) [katexpr,fit = {(awpLB)}] {}; 
			\node (anim1) [katexpr,fit = {(dwlpUB)}] {};
			
			\node (anim1) [katexpr,fit = {(uniLB)}] {};
			\node (anim1) [katexpr,fit = {(intUB)}] {};
			
			\node (anim1) [katexpr,fit = {(intUBsp)}] {};
			\node (anim1) [katexpr,fit = {(uniLBsp)}] {};

			\node (anim1) [katexpr,fit = {(awlpLB)}] {};
			\node (anim1) [katexpr,fit = {(dwpUB)}] {};
			
			% equivalences
			\node (anim1) [collapsewp,fit = {(awlpUB)(uniUB)}] {}; 
			\node (anim1) [collapsewp,fit = {(awlpLB)(uniLB)}] {}; 
			\node (anim1) [collapsewp,fit = {(dwpLB)(intLB)}] {}; 
			\node (anim1) [collapsewp,fit = {(dwpUB)(intUB)}] {};

			\node (anim1) [collapsewp,fit = {(aslpUB)(uniUBsp)}] {}; 
			\node (anim1) [collapsewp,fit = {(aslpLB)(uniLBsp)}] {}; 
			\node (anim1) [collapsewp,fit = {(dspLB)(intLBsp)}] {}; 
			\node (anim1) [collapsewp,fit = {(dspUB)(intUBsp)}] {};
		
		\end{tikzpicture}
	\end{adjustbox}
	\caption{
		Collapse of program logics under certain assumptions, part 3.
		Blue \emph{dashed} boxes enclose logics that collapse (i.e.\ the implications become equivalences) under \emph{absence of branching divergence}.
		Note that the logics in the lower half are equivalent in general, we mark them to demonstrate the symmetry.
	}
	\label{fig:taxonomy-branching}
\end{figure}

\begin{theorem}[Branching Divergence Collapse]
	\label{theo:branching}
	Let $\program$ be a program and $\post$ be a postcondition.
	%Fix a state $\sigma \in \Sigma$.
	If computation of $\program$ either always diverges or always terminates, we have
	\[
		\dwp{\program}{\post} \eeq \awp{\program}{\post} \cap \dwlp{\program}{\post}
%	\]
	\qqand
%	\[
		\awlp{\program}{\post} \eeq \awp{\program}{\post} \cup \dwlp{\program}{\post}.
	\]
\end{theorem}
\begin{proof}
	See \Cref{ssec:proof-branching}. %\cite[Appendix F.3.1]{verscht2025taxonomy}. %
\end{proof}

\subsection{On the Semantics of Syntactic \TopKAT Transformations}
\label{ssec:transformations}

\Cref{fig:taxonomy-full} shows \TopKAT equations for each logic, if expressible.
As mentioned earlier, the equations are syntactically very similar (see \Cref{sec:kat}).
In fact, the equations can be transformed into others by systematical syntactic means.
Consider, for instance, $\top \pre \program \post = \top \pre \program$ for partial correctness.
When moving $\top$ from left to right, we obtain $\pre \program \post \top = \pre \program \top$ for the in-between logic (see \Cref{ssec:unnamed}).
% In \Cref{fig:taxonomy-full}, we move outwards from the middle.
However, just switching $\top$-sides in itself is not a meaningful transformation: Switching sides on $\pre \program \post \top = \pre \top$ (Lisbon logic) yields $\top \pre \program \post = \top \pre$: By comparing \emph{co}domains, this equation would somehow treat the \emph{pre}condition $\pre$ as a set of \emph{final} states (i.e.\ a \emph{post}condition) which is not meaningful.
% This raises the question of (1) which syntactic transformations are possible or sensible, and (2) whether there is a semantic meaning associated with such transformations.

Another transformation is adding or removing $\program$ on the right-hand side of the equation (depending on whether or not $\program$ is already present or not).
For example, from $\top \pre \program \post = \top \post$ (incorrectness logic) we obtain $\top \pre \program \post = \top \program \post$ (angelic partial incorrectness).
Since adding (removing) $\program$ can be seen as a filtering (or not) of unreachable states, this seems like a meaningful transformation.
However, from  $\top \pre \program \post = \top \pre \program$ (partial correctness) we obtain $\top \pre \program \post = \top \pre$, which again interprets $\pre$ as a postcondition.
Consequently, when searching for meaningful syntactic transformations, we must ensure that $\pre$ is a pre- and $\post$ is a postcondition.
We propose hence the following set of meaningful transformations:%
\begin{enumerate}
	\item[(t1)]
	Switch the $\top$-side \emph{and} switch between $\pre$ and $\post$ (or $\pre \program$ and $\program \post$) on the right-hand side.
	This corresponds to switching between incorrectness and correctness reasoning.
	For example, we get from partial \emph{correctness} ($\top \pre \program \post = \top \pre \program$) to partial \emph{incorrectness} ($\pre \program \post\top = \program \post \top $).
	In \Cref{fig:taxonomy-full}, this corresponds to going from the upper to the lower part.
	
	\item[(t2)] Switching from $* \top$ to $ * \program \top$ or $\top *$ to $\top \program *$. %, so adding or removing the program.
	This corresponds to switching from liberal to non-liberal reasoning, which makes sense as adding the program to the equation concentrates the reasoning on reachable states or dually states from which computation terminates.
	% For example, we go from angelic \emph{total} incorrectness ($\top \pre \program \post = \top \post$) to angelic \emph{partial} incorrectness ($\top \pre \program \post = \top \program \post$).
	By requiring $\top$ to be a part of the equation, we avoid the issue described above and ensure that pre- and postconditions are interpreted correctly.

	\item[(t3)] Switching from $\pre \program$ to $\program \post$ and negating all conditions.
	Broadly speaking, this corresponds to switching between angelic and demonic resolution of nondeterminism.
	For example, we go from \emph{angelic} partial incorrectness ($\top \pre \program \post = \top \program \post$) to \emph{demonic} partial incorrectness ($\top \negate{\pre} \program \negate{\post} = \top \negate{\pre} \program$).
	However, starting with partial correctness, for example, we know that the corresponding angelic variant is inexpressible in \TopKAT.
	Applying this syntactic transformation on the equation $\negate{\pre} \program \negate{\post} \top = \program \negate{\post} \top$ for angelic total correctness, we end up with $\pre \program \post \top = \program \post \top$.
	This characterizes the in-between logic, which under the assumption of nonexistence of branching divergence is equivalent to partial correctness.
	So, even though we do not end up exactly with the demonic variant, we end up somewhere \enquote{close}.
\end{enumerate}%
Applying combinations of these three transformations, we can get from each logic in \Cref{fig:taxonomy-kat} to any other logic.
Additionally, the \TopKAT expressible logics are closed under t1 -- t3.
This speaks in favor of the meaningfulness of these transformations.
Nevertheless, the transformations considered appear still somewhat arbitrary to us and we wonder whether there is more underlying structure that we can find in these syntactic transformations.

% !TEX root = ../main.tex

\section{Open Questions}
\label{sec:open}
% or I know that I know nothing
%
% STRUCTURE
% complexity of strongest post versus weakest pre. which is more complex?
    % sp: no inductive rules for dsp and aslp 
    % wp: seven (not four) classes, identification of classes, structure in fig full looks more complex

% on KAT equations
    % complexity of KAT equations: four logics for which a single equation does not suffice: 
    % basic equations
    % basic syntatic transformations
    % variants of kat

% other questions:  
    % 2^7 or 2^4 transformers based on classes, which are sensible?
    % relation to runtime transformers
    % two sides of the coin: why is this not confirmed in picture
    % distinguish demonic and angelic termination? (no)
%
%
% This paper provides a first step towards a structured analysis of program logics, but many particularities remain unexplored.
% The world of program logics is not as well understood as one might think.
%
\paragraph{Complexity of Weakest Pre- and Strongest Postcondition Analyses}
Throughout this paper, we compared analyses using \wpsymbol\ with analyses using \spsymbol.
On several occasions, it appeared as if one is simpler than the other.
For example, in \Cref{sec:transformer}, we saw that there are no inductive definitions for \dspsymbol\ and \aslpsymbol.
This suggests that these transformers are harder to compute.
Is there an intuitive reason for this limitation?
Are there also \wpsymbol\ transformers for which no inductive definitions exists?

On the other hand, in \Cref{fig:wp}, we distinguished \emph{seven} coreachability classes related to $\wpsymbol$, but only \emph{four} reachability classes related to $\spsymbol$.
Also in \Cref{fig:taxonomy-full}, the structure of the bottom (i.e.: the $\spsymbol$) half, seems simpler.
This in turn leaves the impression that \wpsymbol\ is more complex than \spsymbol.

A related observation is that we can characterize each reachability class in \Cref{fig:sp} using (Boolean combinations of) $\spsymbol$ transformers:
\begin{enumerate}
    \item is \dspsymbol.
    \item is \aspsymbol\ without \dspsymbol.
    \item is \dslpsymbol\ without \dspsymbol.
    \item are all states without \aslpsymbol.
\end{enumerate}%
Making a similar list for coreachability classes using \wpsymbol\ transformers is not possible:
Classes (3) and (5) in \Cref{fig:wp} are always either contained in any particular $\wpsymbol$ transformer or not, as evidenced by the fact that they are always either both colored green or both left white.
The difference between (3) and (5) is that, while computation can terminate both inside and outside of the postcondition, (5) additionally allows divergence.
Is there a sensible transformer that we can define to distinguish classes (3) and (5)?
We also wonder whether there is a deeper reason as to why the full separation is possible for coreachability classes, but not for reachability classes.

Additionally, the \spsymbol\ transformers are complete in the sense that combining them does not yield any new logics, as can be seen in \Cref{fig:taxonomy-full}.
This is not the case for the $\wpsymbol$ setting, validated by the in-between logics presented in \Cref{ssec:unnamed}, which were defined as the union and intersection of transformers.
There seem to be more intricacies in the $\wpsymbol$ setting than in the $\spsymbol$ one, but this seems difficult to pin down.%
\paragraph{Kleene Algebraic Reasoning}
While numerous \TopKAT equations express the same program property, determining which should constitute the \enquote{canonical} equations remains an open challenge.
We chose one such a set in \Cref{sec:kat}.
Based on these, we saw that there are some logics which \emph{cannot} be expressed in \TopKAT, but \emph{can} be expressed using predicate transformers.
Can we thus infer that predicate transformers are more powerful than \TopKAT? 
Addressing this requires an argument on why a certain set of \TopKAT equations is indeed canonical.
Additionally, as previously noted in \Cref{sec:kat}, there exist extensions of Kleene algebra that can handle divergence.
With such an algebra, we might get new insights into the expressiveness of \KAT compared to predicate transformers.

Four logics were expressible using predicate transformers but required a \emph{system} of two \TopKAT equations, see \Cref{sec:kat,sec:taxonomy-rev}.
Are these logics inherently more complex than others, and if so, in what sense?
Are there formal arguments we can find for the complexity of program logics?%
\paragraph{Galois Connections in \TopKAT}
We have seen that there exist Galois connections between some of the logics we have considered.
These Galois connections do not seem to be visible yet in the \TopKAT equations.
We wonder whether there exists \emph{syntactical} transformations of \TopKAT equations that correspond to the Galois connections, perhaps under a different set of canonical \TopKAT (in)equations.
In this context, we also wonder whether the \emph{nonexistence} of other Galois connections can also be made visible somehow in \TopKAT.%
\paragraph{Algebraic Reasoning about Divergence}
Kleene algebra was extended to structures such as demonic refinement algebra \cite{von2004towards} or omega algebra \cite{cohen2000separation} in order to facilitate the analysis of nontermination.
Particularly interesting for our setting is \emph{weak omega algebra} \cite{moller2005wp}, which always has a top element and thus fits well into the setting of \TopKAT.
For a more elaborate discussion of the algebraic treatment of divergence, we refer to \cite{jules2011algebraic}.
In future work, we aim to explore how our findings can be extended to algebras capable of more nuanced handling of nontermination.%
%In \Cref{ssec:transformations}, we examined the meaning behind syntactic transformations applied to the \TopKAT equations.
%
%
\paragraph{Deriving Proof Rules}
Most of the predicate transformers presented in this paper are accompanied by a set of inductive rules.
For the novel transformers \dspsymbol\ and \aslpsymbol, establishing such rules is more challenging but possible, as discussed in \Cref{ssec:backward-forward}.
\citet{cousot2024calculational} demonstrates how proof systems for program logics can be constructed via abstractions of the semantics.
It would be interesting to explore whether proof rules could similarly be derived from \TopKAT equations.%
\paragraph{Healthy Transformers}
In theory, the seven coreachability classes give rise to $2^7 = 128$ $\wpsymbol$ transformers, and the four reachability classes give rise to $2^4 = 16$ $\spsymbol$ transformers.
Which of these are sensible or meaningful?
Dijkstra assessed the meaningfulness of predicate transformers based on \emph{healthiness conditions} like strictness, monotonicity, conjunctiveness, etc.
We wonder how many and which ones of the $\wpsymbol$ and $\spsymbol$ transformers meet these criteria.%
\paragraph{Relation to Runtime Transformers}
An extension of weakest pre transformers is the expected runtime transformer \cite{kaminski2018weakest}.
In accordance with \Cref{theo:termination-properties}, the expected runtime transformer can be used to prove termination of programs.
An interesting open question is whether a similar approach can be defined for reachability.
One idea would be to use a strongest postcondition style transformer to compute the expected number of steps required to reach a state.
If this number is finite, we might conclude that the state is reachable.%
\paragraph{The Coin}
Finally, in \Cref{ssec:asymmetries}, we spoke about the asymmetries in the taxonomy \Cref{fig:taxonomy}.
We have yet to find a satisfactory answer to why the \enquote{two sides of the same coin} are not mirrored.%
%
%
% !TEX root = ../main.tex

\section{Conclusion}
\label{sec:conclusion}

\begin{figure}[t]
    \vspace*{-1\intextsep}
    \begin{adjustbox}{max width=\textwidth}
        \begin{tikzpicture}[
            mynode/.style={minimum width=2.7cm},node distance=0.9cm and 0.9cm,
            mylabel/.style n args={1}{label={[label distance=-4mm,font=\tiny,text=DodgerBlue3!60,below]#1}}
            ]
            \node[mynode, mylabel={Lisbon logic (angelic tot.\ corr.)}](awpLB) {$\pre \subseteq \awp{\program}{\post}$};
            \node[mynode, mylabel={Hoare logic (total correctness)},below=of awpLB](dwpLB) {$\pre \subseteq \dwp{\program}{\post}$};
            \node[mynode, mylabel={angelic partial correctness},right=of awpLB](awlpLB) {$\pre \subseteq \awlp{\program}{\post}$};
            \node[mynode,right=of awlpLB]  (dwpUB) {$\dwp{\program}{\post} \subseteq \pre$};
            \node[mynode,right=of dwpUB]  (dwlpUB) {$\dwlp{\program}{\post} \subseteq \pre$};
            \node[mynode,below=of dwlpUB]  (awlpUB) {$\awlp{\program}{\post} \subseteq \pre$};
            
            \node[mynode, mylabel={Hoare logic (partial correctness)},below=of awlpLB] (dwlpLB) {$\pre \subseteq \dwlp{\program}{\post}$};
            \node[mynode,mylabel={partial incorrectness},right=of dwlpLB] (awpUB) {$\awp{\program}{\post} \subseteq \pre$};
            \node[mynode,mylabel={Hoare logic (partial correctness)}, below=of dwlpLB] (aspUB) {$\asp{\program}{\pre} \subseteq \post$};
            \node[mynode,below=of dwpLB] (aslpUB) {$\aslp{\program}{\pre} \subseteq \post$};
            \node[mynode,below=of aslpUB]  (dslpUB) {$\dslp{\program}{\pre} \subseteq \post$};
            \node[mynode,below=of aspUB] (dspUB) {$\dsp{\program}{\pre} \subseteq \post$};
            
            \node[mynode, mylabel={partial incorrectness},right=of aspUB] (dslpLB) {$\post \subseteq \dslp{\program}{\pre}$};
            \node[mynode, mylabel={demonic incorrectness},right=of dslpLB] (dspLB) {$\post \subseteq \dsp{\program}{\pre}$};
            \node[mynode, mylabel={angelic partial incorrectness},below=of dslpLB] (aslpLB) {$\post \subseteq \aslp{\program}{\pre}$};
            \node[mynode, mylabel={incorrectness logic},below=of dspLB] (aspLB) {$\post \subseteq \asp{\program}{\pre}$};
            
            \draw[contrapos] (awpUB) edge[<->] (dwlpLB);
            \draw[contrapos] (dwpUB) edge[<->] (awlpLB);
            \draw[contrapos,bend left=12] (awpLB) edge[<->] (dwlpUB);
            \draw[contrapos,bend left=12] (dwpLB) edge[<->] (awlpUB);
            
            \draw[contrapos] (aspUB) edge[<->] (dslpLB);
            \draw[contrapos,bend left=12] (aspLB) edge[<->] (dslpUB);
            \draw[contrapos] (dspUB) edge[<->] (aslpLB);
            \draw[contrapos,bend left=12] (dspLB) edge[<->] (aslpUB);
            
            \path (dwlpLB) edge[galoisshort] node {} (aspUB);
            \path (awpUB) edge[galoisshort] node {} (dslpLB);

            % text labels
    %			\draw [decorate, ultra thick, decoration = {calligraphic brace, raise=5pt, amplitude=5pt}] (5,-6.3) -- node[pos=0.5] (middletoken1) {} (5,-2.4);
    %			\node[mynode, left=of middletoken1,xshift=1cm] (part corr) {\small partial correctness\vphantom{g}};
    %			
    %			\node[mynode, left= of awpLB,xshift=9cm] (total corr) {\small Lisbon total correctness\vphantom{g}};
    %			
    %			\node[mynode, right=of aspLB2,xshift=-1.8cm] (total incorr) {\small incorrectness\vphantom{g}};
    %			
    %			\draw [decorate,ultra thick, decoration = {calligraphic brace, raise=5pt, amplitude=5pt}] (19.3,-2.4) -- node[pos=0.5] (middletoken2) {} (19.3,-6.3);
    %			\node[mynode, right=of middletoken2,xshift=-1cm] (part incorr) {\small partial incorrectness\vphantom{g}};

            \path (dwpLB) edge[implicationup] node {} (awpLB);
            \path (dwpLB) edge[implication] node {} (dwlpLB);
            \path (awpLB) edge[implication] node {} (awlpLB);
            \path (dwlpLB) edge[implicationup] node {} (awlpLB);
            
            \path (awpUB) edge[implication] node {} (dwpUB);
            \path (awlpUB) edge[implication] node {} (awpUB);
            \path (awlpUB) edge[implication] node {} (dwlpUB);
            \path (dwlpUB) edge[implication] node {} (dwpUB);
            
            \path (dspLB) edge[implicationdown] node {} (aspLB);
            \path (dspLB) edge[implication] node {} (dslpLB);
            \path (aspLB) edge[implication] node {} (aslpLB);
            \path (dslpLB) edge[implicationdown] node {} (aslpLB);
            
            \path (aspUB) edge[implicationdown] node {} (dspUB);
            \path (aslpUB) edge[implication] node {} (aspUB);
            \path (aslpUB) edge[implication] node {} (dslpUB);
            \path (dslpUB) edge[implication] node {} (dspUB);
            
            \path (aslpUB) edge[implication] node {} (aspUB);

            %% AXIS
            \node (ax1start) at ($(dwpLB)!0.5!(aslpUB) + (-1.5,-0.1)$) {};
            \node (ax1end) at ($(dspLB)!0.5!(awlpUB) + (1.5,-0.1) $) {};
            \path (ax1start) edge[axis,draw=rot] node [pos=1.02,opaque]{\textbf{\textcolor{rot}{(1)}}} (ax1end);

            % 2
            \node (ax2start) at ($(awpLB)!0.5!(awlpLB) + (0,1)$) {};
            \node (ax2end) at ($(dslpUB)!0.5!(dspUB) + (0,-1) $) {};
            \path (ax2start) edge[axis,draw=mygreen] node [pos=1.03,opaque]{\textbf{\textcolor{mygreen}{(2)}}}(ax2end);

            \node (ax2bstart) at ($(dwpUB)!0.5!(dwlpUB) + (0,1)$) {};
            \node (ax2bend) at ($(aslpLB)!0.5!(aspLB) + (0,-1) $) {};
            \path (ax2bstart) edge[axis,draw=mygreen] node [pos=1.03,opaque]{\textbf{\textcolor{mygreen}{(2)}}} (ax2bend);

            % 3
            \node (ax3start) at ($(awpLB)!0.5!(dwpLB) + (-1.5,0)$) {};
            \node (ax3end) at ($(dwlpUB)!0.5!(awlpUB) + (1.5,0) $) {};
            \path (ax3start) edge[axis,draw=violet] node [pos=1.02,opaque]{\textbf{\textcolor{violet}{(3)}}} (ax3end);

            \node (ax3bstart) at ($(aslpUB)!0.5!(dslpUB) + (-1.5,0)$) {};
            \node (ax3bend) at ($(dspLB)!0.5!(aspLB) + (1.5,0) $) {};
            \path (ax3bstart) edge[axis,draw=violet] node [pos=1.02,opaque]{\textbf{\textcolor{violet}{(3)}}} (ax3bend);

            \end{tikzpicture}
    \end{adjustbox}
    \caption{
        A taxonomy of predicate transformer-based program logics with axes (1) to (3) corresponding to the dimensions of program logics.
    }
    \label{fig:taxonomy-axis}
    \end{figure}
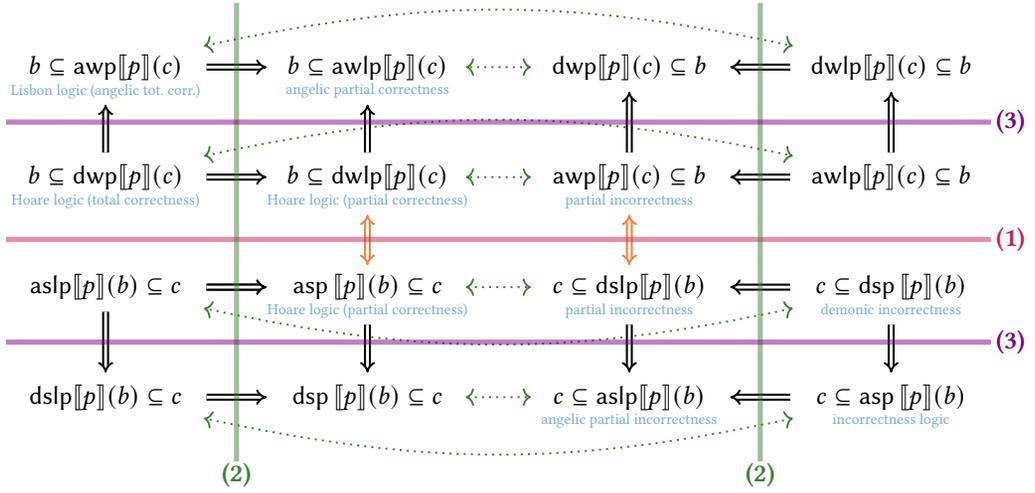%

We have provided a systematic overview of program logics defined by predicate transformers and Kleene algebra with top and tests.
Our graphical illustrations clarify the relationships among various logics.
A main point of interest was the symmetries and asymmetries between forward and backward reasoning.
Notably, we found that many asymmetries could be traced back to one main observation:
Running a nondeterministic program on some initial state can \emph{both} reach some final state \emph{and} diverge.
But no final state can be both reachable from somewhere and at the same time unreachable.
In other words: a nondeterministic computation has the potential to lead to somewhere or nowhere, but it cannot at the same time originate from somewhere or from nowhere.
We call this the \emph{absence of reachability confluence} (\Cref{obs:branching}).

Furthermore, we introduced new predicate transformers -- angelic strongest and demonic strongest liberal postconditions -- as well as novel logics involving union and intersection of transformers.
Thereby, we filled some gaps in the landscape of program logics which seemed to naturally arise when taking the Kleene algebraic view.
Additionally, we discussed in \Cref{ssec:assumptions} how assumptions about program properties, such as determinism or the reachability of final states, influence the taxonomy.
As conjectured at the very beginning, we can indeed identify three dimensions of program logics, each corresponding to an axis in \Cref{fig:taxonomy-axis}:
\begin{enumerate}
	\item
		\emph{correctness} (being able to reach) vs.~\emph{incorrectness} (being reachable)
	\item
		\emph{totality} vs.~\emph{partiality}
	\item
		\emph{angelic} vs.~\emph{demonic} resolution of nondeterminism
\end{enumerate}%
As discussed in \Cref{ssec:assumptions}, if we assume that totality and partiality coincide, i.e.\ if $\program$ always terminates and all states are reachable, the logics collapse along the vertical axis (2).
Dually, if we assume that $\program$ is deterministic and reversible, the logics collapse along the horizontal axis (3).

Apart from being of theoretical interest, the examination of the effect of assumptions is a step towards practical tools:
We explore conditions that have to be discharged, so that different logics happen to collapse.
The classical example of such a condition is
\[
    \text{partial correctness $+$ termination $=$ total correctness}.
\]
If we have a partial correctness proof, we \enquote{merely} have to prove termination to obtain a total correctness proof.
This is practically relevant because partial correctness is a lower bound on a greatest fixed point which can be discharged with invariant-based reasoning.
Total correctness, on the other hand, is a lower bound on a least fixed point, which is much harder to discharge.
Separation of concerns into partial correctness and termination aids to make proving total correctness more tractable.
In a similar manner, we have
\[
    \text{partial incorrectness $+$ reachability $=$ incorrectness}.
\]
Exactly the same least/greatest fixed point considerations apply to partial and \enquote{total} incorrectness.
Hence, partial incorrectness logic is easier to discharge and we obtain that an additional reachability proof would give us \enquote{total} incorrectness.

For Kleene algebra with top and tests, we investigated the relationship between \TopKAT expressible logics and predicate transformer logics.
In the course of this, we showed that we can express Lisbon logic (angelic total correctness) in \TopKAT.
%, which is equivalent to Lisbon logic.
We also saw that in \Cref{tab:KAT}, there is a basic \TopKAT equation which does not directly correspond to a predicate transformer equation.
However, we showed that this equation can be expressed by combining predicate transformers in \Cref{ssec:unnamed}.
This suggests that predicate transformers are stronger in the sense that all \TopKAT equations are expressible using predicate transformers, but not the other way around.
% We also found that the relational view in \KAT is somewhat incompatible with the predicate transformer approach, posing a challenge for future research.
% For the considered set of \TopKAT equations, we saw that predicate transformers are stronger in the sense that all \TopKAT equations are equivalent to some approximation of predicate transformers, but not the other way around.
This is, however, due to the limitations of the chosen \TopKAT approach and could be fixed by including some mechanism for identifying divergence.

%\paragraph{Future work}
%Open questions are collected in \Cref{sec:open}.
%Pre- and postconditions in this paper were simple \emph{predicates} over program states.
%A generalization of this are \emph{expectations}, which are functions over program states.
%Existing work extended weakest pre and strongest post to expectations \cite{McIverM05,ZhangKaminski22}.
%This enables \emph{quantitative} reasoning, and in particular reasoning about probabilistic programs.
%In future work, we want to explore the transferability of our findings to a quantitative or probabilistic setting.
%The latter is expected to be challenging (or impossible?), as the concept of strongest postcondition is provably non-existent for probabilistic programs \cite{claire90}.
%Nevertheless, it might be possible to obtain partial results by restricting to the weakest pre setting, for example.

%% Acknowledgments
\begin{acks}                            %% acks environment is optional
	We would like to thank Kevin Batz and Philipp Schröer for the valuable discussions on practical implications of this paper, as well as the anonymous reviewers for their very constructive and valuable feedback.
	This work was partially supported by the ERC Advanced Research Grant FRAPPANT (grant no. 787914).
\end{acks}

% Bibliography
\bibliography{literature}

%% Appendix - removed for camera-ready version
\pagebreak
\appendix
\section*{Appendix}

% !TEX root = ./main.tex

\section{Demonic Weakest Pre vs.\ Demonic Strongest Post}
\label{sec:illustration}

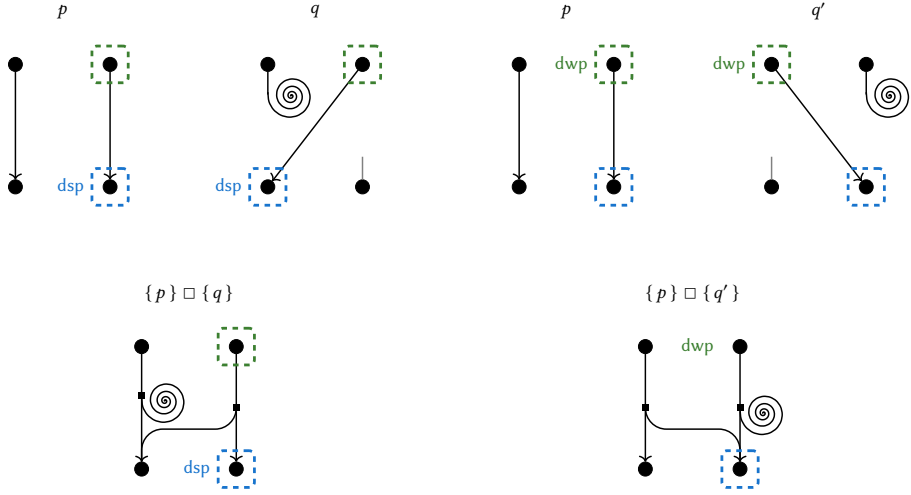
\begin{figure}[t]
	\begin{subfigure}{0.4\textwidth}  
		\begin{center}
			\scalebox{0.7}{
			\begin{tikzpicture}
				\def\statespace{1.8}
				\def\middlespace{3}
				\def\vertspace{2.3}
				\def\vertmiddlespace{3}

				\def\colorpre{mygreen}
				\def\colorpost{DodgerBlue3}
				
				% C1 / a
				\node (c1) at (\statespace/2,1) {$\program$};
				\node[state,fill=black,scale=0.3] (ai0) at (0,0) {};
				\node[state,fill=black,scale=0.3] (ai1) at (\statespace,0) {};
				\node[state,fill=black,scale=0.3] (af0) at (0,-\vertspace) {};
				\node[state,fill=black,scale=0.3] (af1) at (\statespace,-\vertspace) {};
				
				\path (ai0) edge[program] (af0);
				\path (ai1) edge[program] (af1);

				\node (posta) [line width=1.5pt,inner sep=6pt, dashed,rounded corners=0.1cm,fit = {(ai1)}, draw, color=\colorpre] {};
				\node (dspa) [line width=1.5pt,inner sep=6pt, dashed,rounded corners=0.1cm,fit = {(af1)}, draw, color=\colorpost, label=left:{\color{\colorpost}$\dspsymbol$}] {};
				
				% C2 / b
				\node (c2) at (\statespace+\middlespace+\statespace/2,1) {$\secprogram$};
				\node[state,fill=black,scale=0.3] (bi0) at (\statespace+\middlespace,0) {};
				\node[state,fill=black,scale=0.3] (bi1) at (\statespace+\middlespace+\statespace,0) {};
				\node[state,fill=black,scale=0.3] (bf0) at (\statespace+\middlespace,-\vertspace) {};
				\node[state,fill=black,scale=0.3] (bf1) at (\statespace+\middlespace+\statespace,-\vertspace) {};
				
				\node (spiralstartb0) at ($(bi0)!0.3!(bf0)$) {};
				\node (spiralstartb1) at ($(bi1)!0.7!(bf1)$) {};

				\path (bi0) edge[thick] (spiralstartb0);
				\path (bf1) edge[thick,color=gray] (spiralstartb1);
				\path (bi1) edge[program] (bf0);

				\node (spiralmidb0) at ($(spiralstartb0) - (-0.445,-0.1)$) {};
				\draw [scale=0.25,xscale=1,yscale=-1,rotate=90,domain=0:30,variable=\t,smooth,samples=75,thick,-,shift={(spiralmidb0)}] plot  ({\t r}: {-0.002*\t*\t});	

				\node (spiralmidb1) at ($(spiralstartb1) - (-0.445,0.1)$) {};
				% \draw [scale=0.25,xscale=1,yscale=1,rotate=90,domain=0:30,variable=\t,smooth,samples=75,thick,-,shift={(spiralmidb1)},color=gray] plot  ({\t r}: {-0.002*\t*\t});	

				\node (postb) [line width=1.5pt,inner sep=6pt, dashed,rounded corners=0.1cm,fit = {(bi1)}, draw, color=\colorpre] {};
				\node (dspb) [line width=1.5pt,inner sep=6pt, dashed,rounded corners=0.1cm,fit = {(bf0)}, draw, color=\colorpost, label=left:{\color{\colorpost}$\dspsymbol$}] {};

				% C1 square C2 / c
				\node (c12) at (\statespace/2+\middlespace/2+\statespace/2,-\vertspace-\vertmiddlespace+1) {$\NDCHOICE{\program}{\secprogram}$};
				\node[state,fill=black,scale=0.3] (ci0) at (\statespace/2+\middlespace/2,-\vertspace-\vertmiddlespace) {};
				\node[state,fill=black,scale=0.3] (ci1) at (\statespace/2+\middlespace/2+\statespace,-\vertspace-\vertmiddlespace) {};
				\node[state,fill=black,scale=0.3] (cf0) at (\statespace/2+\middlespace/2,-\vertspace-\vertmiddlespace-\vertspace) {};
				\node[state,fill=black,scale=0.3] (cf1) at (\statespace/2+\middlespace/2+\statespace,-\vertspace-\vertspace-\vertmiddlespace) {};
				
				\node[rectangle,fill=black,scale=0.5] (spiralstartc0) at ($(ci0)!0.4!(cf0)$) {};
				\node[rectangle,fill=black,scale=0.5] (spiralstartc1) at ($(ci1)!0.5!(cf1)$) {};

				\path (ci0) edge[thick] (spiralstartc0);
				\path (ci1) edge[thick] (spiralstartc1);
				\path (spiralstartc0) edge[program] (cf0);
				\path (spiralstartc1) edge[program] (cf1);

				\node (spiralmidc0) at ($(spiralstartc0) - (-0.445,0.1)$) {};
				\draw [scale=0.25,xscale=1,yscale=-1,rotate=90,domain=0:30,variable=\t,smooth,samples=75,thick,-,shift={(spiralmidc0)}] plot  ({\t r}: {-0.002*\t*\t});	

				\node (spiralmidc1) at ($(spiralstartc1) - (-0.445,-0.14)$) {};
				% \draw [scale=0.25,xscale=1,yscale=1,rotate=90,domain=0:30,variable=\t,smooth,samples=75,thick,-,shift={(spiralmidc1)},color=gray] plot  ({\t r}: {-0.002*\t*\t});	

				\draw [rounded corners=4mm,thick] (spiralstartc1) |- ++(-\statespace,-0.4) -- (cf0);

				\node (postc) [line width=1.5pt,inner sep=6pt, dashed,rounded corners=0.1cm,fit = {(ci1)}, draw, color=\colorpre] {};
				\node (dspc) [line width=1.5pt,inner sep=6pt, dashed,rounded corners=0.1cm,fit = {(cf1)}, draw, color=\colorpost, label=left:{\color{\colorpost}$\dspsymbol$}] {};
			
			\end{tikzpicture}}
		\end{center}
	\end{subfigure}
	%
	% WP PART
	\hspace{3em}
	\begin{subfigure}{0.4\textwidth}  
		\begin{center}
			\scalebox{0.7}{
			\begin{tikzpicture}
				\def\statespace{1.8}
				\def\middlespace{3}
				\def\vertspace{2.3}
				\def\vertmiddlespace{3}

				\def\colorpre{mygreen}
				\def\colorpost{DodgerBlue3}
				
				% C1 / a
				\node (c1) at (\statespace/2,1) {$\program$};
				\node[state,fill=black,scale=0.3] (ai0) at (0,0) {};
				\node[state,fill=black,scale=0.3] (ai1) at (\statespace,0) {};
				\node[state,fill=black,scale=0.3] (af0) at (0,-\vertspace) {};
				\node[state,fill=black,scale=0.3] (af1) at (\statespace,-\vertspace) {};
				
				\path (ai0) edge[program] (af0);
				\path (ai1) edge[program] (af1);

				\node (dwpa) [line width=1.5pt,inner sep=6pt, dashed,rounded corners=0.1cm,fit = {(ai1)}, draw, color=\colorpre, label=left:{\color{\colorpre}$\dwpsymbol$}] {};
				\node (posta) [line width=1.5pt,inner sep=6pt, dashed,rounded corners=0.1cm,fit = {(af1)}, draw, color=\colorpost] {};
				
				% C2 / b
				\node (c2) at (\statespace+\middlespace+\statespace/2,1) {$\secprogram'$};
				\node[state,fill=black,scale=0.3] (bi0) at (\statespace+\middlespace,0) {};
				\node[state,fill=black,scale=0.3] (bi1) at (\statespace+\middlespace+\statespace,0) {};
				\node[state,fill=black,scale=0.3] (bf0) at (\statespace+\middlespace,-\vertspace) {};
				\node[state,fill=black,scale=0.3] (bf1) at (\statespace+\middlespace+\statespace,-\vertspace) {};
				
				\node (spiralstartb0) at ($(bi0)!0.7!(bf0)$) {};
				\node (spiralstartb1) at ($(bi1)!0.3!(bf1)$) {};

				\path (bi1) edge[thick] (spiralstartb1);
				\path (bf0) edge[thick,color=gray] (spiralstartb0);
				\path (bi0) edge[program] (bf1);

				\node (spiralmidb1) at ($(spiralstartb1) - (-0.445,-0.1)$) {};
				\draw [scale=0.25,xscale=1,yscale=-1,rotate=90,domain=0:30,variable=\t,smooth,samples=75,thick,-,shift={(spiralmidb1)}] plot  ({\t r}: {-0.002*\t*\t});	

				\node (spiralmidb0) at ($(spiralstartb0) - (-0.445,0.1)$) {};
				% \draw [scale=0.25,xscale=1,yscale=1,rotate=90,domain=0:30,variable=\t,smooth,samples=75,thick,-,shift={(spiralmidb0)},color=gray] plot  ({\t r}: {-0.002*\t*\t});	

				\node (dwpb) [line width=1.5pt,inner sep=6pt, dashed,rounded corners=0.1cm,fit = {(bi0)}, draw, color=\colorpre, label=left:{\color{\colorpre}$\dwpsymbol$}] {};
				\node (postb) [line width=1.5pt,inner sep=6pt, dashed,rounded corners=0.1cm,fit = {(bf1)}, draw, color=\colorpost] {};

				% C1 square C2 / c
				\node (c12) at (\statespace/2+\middlespace/2+\statespace/2,-\vertspace-\vertmiddlespace+1) {$\NDCHOICE{\program}{\secprogram'}$};
				\node[state,fill=black,scale=0.3] (ci0) at (\statespace/2+\middlespace/2,-\vertspace-\vertmiddlespace) {};
				\node[state,fill=black,scale=0.3] (ci1) at (\statespace/2+\middlespace/2+\statespace,-\vertspace-\vertmiddlespace) {};
				\node[state,fill=black,scale=0.3] (cf0) at (\statespace/2+\middlespace/2,-\vertspace-\vertmiddlespace-\vertspace) {};
				\node[state,fill=black,scale=0.3] (cf1) at (\statespace/2+\middlespace/2+\statespace,-\vertspace-\vertspace-\vertmiddlespace) {};
				
				\node[rectangle,fill=black,scale=0.5] (spiralstartc0) at ($(ci0)!0.5!(cf0)$) {};
				\node[rectangle,fill=black,scale=0.5] (spiralstartc1) at ($(ci1)!0.5!(cf1)$) {};

				\path (ci1) edge[thick] (spiralstartc1);
				\path (ci0) edge[thick] (spiralstartc0);
				\path (spiralstartc1) edge[program] (cf1);
				\path (spiralstartc0) edge[program] (cf0);

				\node (spiralmidc1) at ($(spiralstartc1) - (-0.445,0.1)$) {};
				\draw [scale=0.25,xscale=1,yscale=-1,rotate=90,domain=0:30,variable=\t,smooth,samples=75,thick,-,shift={(spiralmidc1)}] plot  ({\t r}: {-0.002*\t*\t});	

				\node (spiralmidc0) at ($(spiralstartc0) - (-0.445,-0.14)$) {};
				% \draw [scale=0.25,xscale=1,yscale=1,rotate=90,domain=0:30,variable=\t,smooth,samples=75,thick,-,shift={(spiralmidc0)},color=gray] plot  ({\t r}: {-0.002*\t*\t});	

				\draw [rounded corners=4mm,thick] (spiralstartc0) |- ++(\statespace,-0.4) -- (cf1);

				\node (dwpc) [line width=1.5pt,inner sep=6pt, dashed,rounded corners=0.1cm,fit = {(ci1)}, draw=none, color=\colorpre, label=left:{\color{\colorpre}$\dwpsymbol$}] {};
				\node (postc) [line width=1.5pt,inner sep=6pt, dashed,rounded corners=0.1cm,fit = {(cf1)}, draw, color=\colorpost] {};
			
			\end{tikzpicture}}
		\end{center}
	\end{subfigure}
	\caption{
		Differences between computing the demonic strongest post and demonic weakest pre for nondeterministic choices.
		Preconditions are in green dashes, postconditions in blue.
		The black spirals indicate divergence. %, the grey backward spirals indicate unreachability.
	}
	\label{fig:adwpsp_asymm}
\end{figure}

% \begin{figure}[t]
% 	\includegraphics[width=\textwidth]{img/adwpsp_asymm}
% 	\caption{
% 		Differences between computing the demonic strongest post and demonic weakest pre for nondeterministic choices.
% 		Preconditions are circled in red, postconditions are circled in green.
% 		The black spirals indicate divergence, the grey backward spirals indicate unreachability.
% 	}
% 	\label{fig:adwpsp_asymm}
% \end{figure}

\Cref{obs:branching} described how divergence can branch from a terminating program path, but unreachability cannot confluence to any path.
In \Cref{fig:adwpsp_asymm}, we see a consequence of this observation.
The left-hand side shows three programs $\program$, $\secprogram$, and \mbox{$\NDCHOICE{\program}{\secprogram}$}.
We assume the state space to consist of two states.
The second state fulfills the precondition (in green).
The demonic strongest post w.r.t.\ this precondition for each program is circled in blue.
Notably, for the program $\NDCHOICE{\program}{\secprogram}$, the second state is in $\dspsymbol$, as it is exclusively reachable from the precondition.
The fact that state two was unreachable in $\program_2$ is lost.

On the right-hand side, we see a mirrored scenario for the weakest pre transformer.
Again, the second state fulfills the postcondition (in blue), and the demonic weakest pre w.r.t.\ this postcondition is circled in green.
In contrast to before, the demonic weakest pre of $\NDCHOICE{\program}{\secprogram'}$ is empty.
If \spsymbol\ and \wpsymbol\ were truly dual, the second state would be included, which it is not because divergence from this state is possible.
On the left-hand side, this was not a problem since unreachability cannot confluence by \Cref{obs:branching}.

% \section{$\boldsymbol{\top}$ as a (Co-)Domain Selector}
% \label{sec:illustration-codomain}

\section{Counterexamples for Galois Connections}
\label{sec:counter}

\Cref{sec:taxonomy} presents seven predicate transformer based program logics:
\begin{enumerate}
    \item (demonic) partial correctness (Hoare logic),
    \item (demonic) total correctness (Hoare logic),
    \item angelic total correctness (Lisbon logic),
    \item angelic partial correctness,
    \item (angelic total) incorrectness,
    \item demonic (total) incorrectness, and
    \item angelic partial incorrectness.
\end{enumerate}

We know that there is a Galois connection between demonic partial correctness and the contrapositive of demonic partial incorrectness.
In the following, we present counterexamples for all other potential Galois connections between these seven logics and also their contrapositives, which gives a total of $\binom{14}{2} = 91$ combinations.

\paragraph{Part One}
To do so, we first consider all pairs between the aforementioned seven logics, which are $\binom{7}{2} = 21$.
We can directly exclude some pairs.
First of all, \emph{partial} and \emph{total} (in)correctness logics are never equal.
For correctness logics, we get counterexamples by adding states from which computation diverges to the precondition.
For incorrectness logics, we can add unreachable states to the postcondition.
We can therefore exclude the following pairs:
\begin{enumerate}
    \item[(1) - (2)] demonic partial correctness and demonic total correctness,
    \item[(3) - (4)] angelic partial correctness and angelic total correctness, and
    \item[(5) - (7)] angelic total incorrectness and angelic partial incorrectness.
\end{enumerate}
% We can also exclude the corresponding pairs between the contrapositives and 
% % We are left with $21 - 3 = 18$.

Similarly, we note that \emph{demonic} and \emph{angelic} variants are never equal.
Counterexamples can be constructed by adding a state to the precondition from which computation can terminate both within and outside of the postcondition.
For incorrectness logics, we add a final state to the postcondition that is reachable from within and outside of the precondition.
Then, we can exclude the following pairs
\begin{enumerate}
    \item[(1) - (4)] demonic partial correctness and angelic partial correctness,
    \item[(2) - (3)] demonic total correctness and angelic total correctness, and
    \item[(5) - (6)] demonic total incorrectness and angelic total incorrectness.
\end{enumerate}
%There are $15$ pairs left.
Next, note that correctness and incorrectness logics are not equal:
Consider postcondition that includes a state which is unreachable from the precondition.
This does not affect correctness logics (1),(2),(3), and (4), however, such a triple can definitely not be valid for any of the incorrectness logics (5),(6), and (7).
We can therefore exclude all pairs between them:
\begin{enumerate}
    \item[(1) - (5)] demonic partial correctness and angelic total incorrectness
    \item[(1) - (6)] demonic partial correctness and demonic total incorrectness
    \item[(1) - (7)] demonic partial correctness and angelic partial incorrectness
    \item[(2) - (5)] demonic total correctness and angelic total incorrectness
    \item[(2) - (6)] demonic total correctness and demonic total incorrectness
    \item[(2) - (7)] demonic total correctness and angelic partial incorrectness
    \item[(3) - (5)] angelic total correctness and angelic total incorrectness
    \item[(3) - (6)] angelic total correctness and demonic total incorrectness
    \item[(3) - (7)] angelic total correctness and and angelic partial incorrectness
    \item[(4) - (5)] angelic partial correctness and angelic total incorrectness
    \item[(4) - (6)] angelic partial correctness and demonic total incorrectness
    \item[(4) - (7)] angelic partial correctness and angelic partial incorrectness
\end{enumerate}

Now there are only three pairs left, for which we now describe a condition under which a triple might be valid for the first, but not for the second logic.
\begin{enumerate}
    \item[(1) - (3)] demonic partial correctness and angelic total correctness:
    $\pre$ contains a state from which computation diverges
    
    \item[(4) - (2)] angelic partial correctness and demonic total correctness:
    $\pre$ contains a state from which computation diverges

    \item[(7) - (6)] angelic partial incorrectness and demonic total incorrectness:
    $\post$ contains an unreachable state
\end{enumerate}

We conclude that all seven logics are truly different.
This also implies that the seven contrapositive logics are different from another.

\paragraph{Part Two}
It remains to be shown that no logic is equivalent to the contrapositive of another.
For this, we start by examining the contrapositive of logic (1), denoted by ${(\negate{1})}$, and show that this is not equal to any other logic.

The characterizing subset equation is given by $\awp{\program}{\post} \subseteq \pre$, so all states from which computation can terminate in $\post$ are included in $\pre$.
In other words, this means that all states in $\post$ are either unreachable or exclusively reachable from $\pre$.
If $\post$ contains a state that is reachable from $\pre$ and from $\negate{\pre}$, the triple cannot be valid for ${(\negate{1})}$.
However, this does not matter for the correctness logics (1) to (4), and also for (7), implying that ${(\negate{1})}$ cannot be equal to those logics.
For logics (5) and (6), all states in $\post$ must be reachable.
For ${(\negate{1})}$, this is note the case, i.e.\ unreachable states can be included in $\post$.
Thus, ${(\negate{1})}$ is also not equivalent to (5) or (6).

Logic ${(\negate{2})}$ is described by $\awlp{\program}{\post} \subseteq \pre$, so all states from which computation can diverge or terminate in $\post$ are included in $\pre$.
% Note that we do not have to compare this to (1) anymore, since equivalence of ${(\negate{2})}$ and (1) would imply equivalence of (2) and ${(\negate{1})}$, which was disproven in the previous paragraph.
Let $\negate{\pre}$ contain a state from which computation diverges, which implies that the triple is not valid for ${(\negate{2})}$.
The correctness logics (1) to (4) do not require anything from states outside of $\pre$.
The incorrectness logics (5) to (6) do not require anything for all initial states.
Consequently, ${(\negate{2})}$ is not equal to any of the logics.

Logic ${(\negate{3})}$ is described by $\dwlp{\program}{\post} \subseteq \pre$, so all states from which computation always either diverges or terminates in $\post$ are included in $\pre$.
If $\negate{\pre}$ contains a state from which computation diverges, we can apply the same reasoning as for ${(\negate{2})}$.

Logic ${(\negate{4})}$ is described by $\dwp{\program}{\post} \subseteq \pre$, so all states from which computation always terminates in $\post$ are included in $\pre$.
Let $\negate{\pre}$ contain a state from which computation always terminates in $\post$.
The triple cannot be valid for ${(\negate{4})}$.
For the same reasons as above, the correctness logics are not affected by this.
For the incorrectness logics, assume another triple which is valid for ${(\negate{4})}$, but additionally there is a state in $\post$ which is only reachable from outside of $\pre$.
This is not a contradiction as long as the states outside of $\pre$ that can reach $\post$ can also either diverge or terminate in $\negate{\post}$.
Such a triple is invalid for all incorrectness logics, proving that also ${(\negate{4})}$ is not equal to any of the logics (1) to (7).

Logic ${(\negate{5})}$ is described by $\dslp{\program}{\pre} \subseteq \post$, so all states that are unreachable or exclusively reachable from $\pre$ must be included in $\post$.
Assume there is an unreachable state outside of $\post$.
A corresponding triple cannot be valid for ${(\negate{5})}$, but the logics (1) to (7) are unaffected by this.

Logic ${(\negate{6})}$ is described by $\aslp{\program}{\pre} \subseteq \post$, so all states that are unreachable or reachable from $\pre$ must be included in $\post$.
Again, assume there is an unreachable state outside of $\post$.
A corresponding triple cannot be valid for ${(\negate{6})}$, but the logics (1) to (7) are unaffected by this.

Lastly, logic ${(\negate{7})}$ is described by $\dsp{\program}{\pre} \subseteq \post$, so all states that are exclusively reachable from $\pre$ must be included in $\post$.
Assume there is a state outside of $\post$ which is exclusively reachable from $\pre$.
This does not affect the incorrectness logics.
For the correctness logics, we consider another triple which is valid for ${(\negate{7})}$, but there is a state in $\pre$ from which computation always terminates outside of $\post$.
Such a triple can exist as the final states in question are not necessarily exclusively reachable from $\pre$.
Then, this triple is not valid for any correctness logic (1) to (4).

We have now discussed all 91 combinations:
21 in the first part, which implied 21 more on the contrapositive side.
The second part excluded the remaining 49 pairs: seven possible connections for all seven contrapositive logics.

% !TEX root = ./main.tex

\section{Kleene Algebra}
\label{sec:kat-definitions}

Formally, a Kleene algebra is an idempotent (and thus partially ordered) semiring  endowed with a closure operator.
It generalizes the operations known from regular languages.

\begin{definition}[Kleene algebra]
	A \emph{Kleene algebra} is a set $A$ together with two binary operations $+\colon A \to A$ and $\cdot \colon A \to A$ and a function $*\colon A \to A$, such that the following axioms are fulfilled:
	\begin{itemize}
		\item associativity of $+$ and $\cdot$
		\item commutativity of $+$
		\item distributivity of $\cdot$ over $+$
		\item identity elements $0$ for $+$ and $1$ for $\cdot$
		\item annhiliation by $0$: $0 \cdot a = a\cdot 0 = 0 \quad \forall a \in A$
		\item idempotence of $+$: $a + a = a \quad \forall a \in A$
	\end{itemize}
	We can define a partial order as usual by $ a \leq b $ iff $a+b = b$.
	Additionally, for all $a,x \in A$:
	\begin{itemize}
		\item $1 + a(a^*) \leq a^*$
		\item $1 + (a^*)a \leq a^*$
		\item $ax \leq x \implies a^* x \leq x$
		\item $xa \leq x \implies x a^*  \leq x$
	\end{itemize}
\end{definition}

Intuitively, one can think of $a+b$ as the union or least upper bound of $a$ and $b$, and of $ab$ as some multiplication that is monotonic, in the sense that $a \leq b$ implies $ax \leq bx$.
The idea behind the star operation is $a^* = 1 + a + aa + aaa + \dots$. 
From the standpoint of prgramming language theory, one may interpret $+$ as nondeterministic choice, $\cdot$ as sequence, and $*$ as iteration.

\begin{definition}[Kleene algebra with tests \cite{kozen1997kleene}]
	A \emph{Kleene algebra with tests (\KAT)} is a two sorted algebra $(K,B,+,\cdot,*,0,1,\negate{\phantom{a}})$ where $B \subseteq K$ and $\negate{\phantom{a}}$ is a unary operator on $B$ such that
	\begin{itemize}
		\item $(K,+,\cdot,*,0,1)$ is a Kleene algebra and
		\item $(B,+,\cdot,\negate{\phantom{a}},0,1)$ is a Boolean algebra.
	\end{itemize}
	The elements of $B$ are called tests.
\end{definition}

Programs can be modelled in \KAT as follows:
\begin{itemize}
	\item fail is $0$ 
	\item skip is $1$
	\item $p;q$ is $pq$
	\item if $b$ then $p$ else $q$ is $bp + \negate{b}q$
	\item while $b$ do $p$ is $(bp)^* \negate{b}$	
\end{itemize}

% As mentioned before, we focus on relational Kleene algebras.
% Kleene algebras with tests subsume propositional Hoare logic.
% This means that we cannot reason about real programs (as there is no possibility for assignments), but can prove things like equivalences of program properties in an algebraic manner.

In \KAT, we can express e.g.\ partial correctness.
To express incorrectness, we need to assume the existence of a top element.

\begin{definition}[Kleene algebra with top and tests \cite{zhang2022incorrectness}]
	A \emph{Kleene algebra with top and tests (\TopKAT)} is a \KAT $(K,B,+,\cdot,*,0,1,\negate{\phantom{a}})$ with a top element $\top$ such that for all $a \in K$ it holds that $a \leq \top$.
\end{definition}
% !TEX root = ./main.tex

\section{Inductive Rules for Predicate Transformers}
\label{sec:rules}

\subsection{Weakest Pre}
\label{ssec:wp-rules}

\begin{table}[H]
	\begin{adjustbox}{max width=0.9\textwidth}
	\renewcommand{\arraystretch}{1.5}
	\begin{tabular}{@{\hspace{.5em}}l@{\hspace{2em}}l@{\hspace{2em}}l@{\hspace{.5em}}}
		\hline\hline
		$\boldsymbol{\program}$			& $\boldawp{\program}{\post}$ & $\bolddwp{\program}{\post}$\\
		\hline
		$\SKIP$				& $\post$ & \lightgray{$\post$}	\\
		$\DIVERGE$				& $\false$ & \lightgray{$\false$}	\\
		$\ASSIGN{x}{\ee}$			& $\post \subst{x}{\ee}$ & \lightgray{$\post \subst{x}{\ee}$} \\
		$\COMPOSE{\program_1}{\program_2}$		& $\awp{\program_1}{\vphantom{\big(}\awp{\program_2}{\post}}$ & \lightgray{$\awp{\program_1}{\vphantom{\big(}\dwp{\program_2}{\post}}$}\\
		$\NDCHOICE{\program_1}{\program_2}$		& $\awp{\program_1}{\post} \vee \awp{\program_2}{\post}$ & $\dwp{\program_1}{\post} \wedge \dwp{\program_2}{\post}$ \\
		$\ITE{\guard}{\program_1}{\program_2}$		& ${\guard} \wedge \awp{\program_1}{\post} \vee \iverson{\neg \guard} \wedge \awp{\program_2}{\post}$ & \lightgray{${\guard} \wedge \dwp{\program_1}{\post} \vee \iverson{\neg \guard} \wedge \dwp{\program_2}{\post}$} \\
		$\WHILEDO{\guard}{\program'}$		& $\lfp  X\mydot \quad {\neg \guard} \wedge \post \vee \iverson{\guard} \wedge \awp{\program'}{X}$ & \lightgray{$\lfp  X\mydot \quad {\neg \guard} \wedge \post \vee \iverson{\guard} \wedge \dwp{\program'}{X}$}\\[.25em]
		\hline\hline%\\[-1.5em]
	\end{tabular}%
	\end{adjustbox}%
	\vspace{.5em}
	\caption{Rules for $\awpsymbol$ and $\dwpsymbol$. %the quantitative weakest pre transformers. 
			$\lfp g\mydot \Phi(g)$ and $\gfp g\mydot \Phi(g)$ denote the least and greatest fixed point of $\Phi$.}%
	\label{table:wp}
	\vspace{-1em}
\end{table}%
%}

\subsection{Weakest Liberal Pre}
\label{ssec:wlp-rules}

\begin{table}[H]
	\begin{adjustbox}{max width=0.9\textwidth}
	\renewcommand{\arraystretch}{1.5}
	\begin{tabular}{@{\hspace{.5em}}l@{\hspace{2em}}l@{\hspace{2em}}l@{\hspace{.5em}}}
		\hline\hline
		$\boldsymbol{\program}$			& $\boldawlp{\program}{\post}$ & $\bolddwlp{\program}{\post}$\\
		\hline
		$\SKIP$				& $\post$ & \lightgray{$\post$}	\\
		$\DIVERGE$				& $\true$ & \lightgray{$\true$}	\\
		$\ASSIGN{x}{\ee}$			& $\post \subst{x}{\ee}$ & \lightgray{$\post \subst{x}{\ee}$} \\
		$\COMPOSE{\program_1}{\program_2}$		& $\awlp{\program_1}{\vphantom{\big(}\awlp{\program_2}{\post}}$ & \lightgray{$\dwlp{\program_1}{\vphantom{\big(}\dwlp{\program_2}{\post}}$}\\
		$\NDCHOICE{\program_1}{\program_2}$		& $\awlp{\program_1}{\post} \vee \awlp{\program_2}{\post}$ & $\dwlp{\program_1}{\post} \wedge \dwlp{\program_2}{\post}$\\
		$\ITE{\guard}{\program_1}{\program_2}$		& ${\guard} \wedge \awlp{\program_1}{\post} \vee \iverson{\neg \guard} \wedge \awlp{\program_2}{\post}$ & \lightgray{${\guard} \wedge \dwlp{\program_1}{\post} \vee {\neg \guard} \wedge \dwlp{\program_2}{\post}$} \\
		$\WHILEDO{\guard}{\program'}$		& $\gfp  X\mydot \quad {\neg \guard} \wedge \post \vee {\guard} \wedge \awlp{\program'}{X}$ & \lightgray{$\gfp  X\mydot \quad {\neg \guard} \wedge \post \vee {\guard} \wedge \dwlp{\program'}{X}$}\\[.25em]
		\hline\hline%\\[-1.5em]
	\end{tabular}%
	\end{adjustbox}%
	\vspace{.5em}
	\caption{Rules for $\awlpsymbol$ and $\dwlpsymbol$. %the quantitative weakest pre transformers. 
			$\lfp g\mydot \Phi(g)$ and $\gfp g\mydot \Phi(g)$ denote the least and greatest fixed point of $\Phi$.}%
	\label{table:wlp}
	\vspace{-1em}
\end{table}%

\subsection{Strongest Post}
\label{ssec:sp-rules}

\begin{table}[H]
	\begin{adjustbox}{max width=0.9\textwidth}
	\renewcommand{\arraystretch}{1.5}
	\begin{tabular}{@{\hspace{.5em}}l@{\hspace{2em}}l@{\hspace{2em}}l@{\hspace{.5em}}}
		\hline
		$\boldsymbol{\program}$			& $\boldasp{\program}{\pre}$ \\
		\hline
		$\SKIP$					& $\pre$ \\
		$\DIVERGE$				& $\false$ 										\\
		$\ASSIGN{x}{\ee}$			& $\exists{\alpha}\mydot {x = e\subst{x}{\alpha}}\, \wedge\, \pre\subst{x}{\alpha}$ \\
		$\COMPOSE{\program_1}{\program_2}$		& $\asp{\program_2}{\vphantom{\big(}\asp{\program_1}{\pre}}$ \\
		$\NDCHOICE{\program_1}{\program_2}$		& $\asp{\program_1}{\pre} \vee \asp{\program_2}{\pre}$ \\
		$\ITE{\guard}{\program_1}{\program_2}$		& $\asp{\program_1}{{\guard} \wedge \pre} \vee \asp{\program_2}{{\neg \guard} \wedge \pre}$  \\
 		$\WHILEDO{\guard}{C'}$		& ${\neg\guard} \wedge \bigl( \lfp Y\mydot \pre \vee \asp{\program'}{{\guard} \wedge Y}\bigr)$ \\
		\hline\hline
	\end{tabular}%
	\end{adjustbox}
	\caption{
		Rules for $\aspsymbol$.
		$\lfp g\mydot \Psi(g)$ denotes the least fixed point of $\Phi$.
	}
	\label{table:sp}
\end{table}%

\subsection{Strongest Liberal Post}
\label{ssec:slp-rules}

\begin{table}[H]
	\begin{adjustbox}{max width=0.9\textwidth}
	\renewcommand{\arraystretch}{1.5}
	\begin{tabular}{@{\hspace{.5em}}l@{\hspace{2em}}l@{\hspace{2em}}l@{\hspace{.5em}}}
		\hline
		$\boldsymbol{\program}$			& $\bolddslp{\program}{\pre}$ \\
		\hline
		$\SKIP$					& $\pre$ \\
		$\DIVERGE$				& $\true$ 										\\
		$\ASSIGN{x}{\ee}$			& $\forall{\alpha}\mydot {x \neq e\subst{x}{\alpha}} \, \vee\,  \pre\subst{x}{\alpha}$ \\
		$\COMPOSE{\program_1}{\program_2}$		& $\sp{\program_2}{\vphantom{\big(}\sp{\program_1}{\pre}}$ \\
		$\NDCHOICE{\program_1}{\program_2}$		& $\sp{\program_1}{\pre} \wedge \sp{\program_2}{\pre}$ \\
		$\ITE{\guard}{\program_1}{\program_2}$		& $\sp{\program_1}{{\neg \guard} \vee \pre} \wedge \sp{\program_2}{{\guard} \vee \pre}$  \\
 		$\WHILEDO{\guard}{C'}$		& ${\guard} \vee \bigl( \gfp Y\mydot \pre \wedge \sp{\program'}{\neg {\guard} \wedge Y}\bigr)$ \\
		\hline\hline
	\end{tabular}%
	\end{adjustbox}%
	\caption{
		Rules for $\dslpsymbol$.
		$\gfp g\mydot \Psi(g)$ denotes the greatest fixed point of $\Phi$.
	}
	\label{table:dslp}
\end{table}%

% !TEX root = ./main.tex

\section{Omitted Proofs}
\label{sec:proofs}

\subsection{Proof of \Cref{theo:complements}}
\label{ssec:proof-complements}

Properties (1), (2), and (4) of \Cref{theo:complements} are folklore knowledge.
We prove property (4):
\[
	\dsp{\program}{\pre} \eeq \negatedbl{\aslp{\program}{\negate{\pre}}}
\]
by showing that the (disjoint) union $\dsp{\program}{\pre} \quad \cupdot \quad  \aslp{\program}{\negate{\pre}}$ is equal to the full state space $\Sigma$
for all programs $p \in \ngcl$ and predicates $\pre$.

\begin{proof}
	\begin{align*}
		& \dsp{\program}{\pre} \quad \cupdot \quad  \aslp{\program}{\negate{\pre}} \\
		= &\ \mylambda{\tau}
		\begin{cases}
			\bigwedge\limits_{\sigma \in \seminv{\program}{\tau}} \pre(\sigma),  & \text{ if }\  \seminv{\program}{\tau} \neq \emptyset\\
			\false, & \text{ otherwise .}
		\end{cases}
		\quad \cupdot \quad  
		\mylambda{\tau}
		\begin{cases}
			\bigvee\limits_{\sigma \in \seminv{\program}{\tau}} \negate{\pre}(\sigma),  & \text{ if }\ \seminv{\program}{\tau} \neq \emptyset \\
			\true, & \text{ otherwise .}
		\end{cases}  \tag{\Cref{def:dsp,def:aslp}}\\
		= &\ \mylambda{\tau}
		\begin{cases}
			\left(\bigwedge\limits_{\sigma \in \seminv{\program}{\tau}} \pre(\sigma)\right) 
			\lordot \left(\bigvee\limits_{\sigma \in \seminv{\program}{\tau}} \negate{\pre}(\sigma) \right),  & \text{ if }\  \seminv{\program}{\tau} \neq \emptyset\\
			\true, & \text{ otherwise .}
		\end{cases} \\
		= &\ \mylambda{\tau}
		\begin{cases}
			\true  & \text{ if }\  \seminv{\program}{\tau} \neq \emptyset\\
			\true, & \text{ otherwise .}
		\end{cases} \\
		= &\ \mylambda{\tau} \true \\
		= &\ \Sigma
	\end{align*}
	This concludes the proof of \Cref{theo:complements}.
\end{proof}

\newpage
\subsection{Proofs from \Cref{ssec:assumptions}}
\label{ssec:proof-assumptions}

\subsubsection{Proof of \Cref{theo:termination}}
\label{ssec:proof-termination}

\begin{proof}
	Let $\program \in \ngcl$ be a program and $\post$ a postcondition.
	Fix a state $\sigma \in \Sigma$ such that there is no diverging path starting in $\sigma$.
	We prove the theorem pointwise for all states via the soundness criteria of the transformers.
	\begin{align*}
		& \awlp{\program}{\post}(\sigma) \\
		=\ & \begin{cases}
				\true & \text{ if there exists a diverging path starting in } \sigma \\
				\bigvee\limits_{\tau \in  \sem{\program}{\sigma}} F(\tau) & \text{ else}
			\end{cases} \tag{\Cref{def:wlp}}\\
		=\ & \bigvee\limits_{\tau \in \sem{\program}{\sigma}} F(\tau) \tag{by assumption}\\
		=\ & \awp{\program}{\post}(\sigma) \tag{\Cref{def:wp}}
	\end{align*}
\end{proof}

\subsubsection{Proof of \Cref{theo:reachability}}
\label{ssec:proof-reachability}

\begin{proof}
	Let $\program \in \ngcl$ be a program and $\pre$ a precondition.
	Assume $\tau \in \Sigma$ is reachable.
	This means that there exists a state $\sigma \in \Sigma$ such that $\sigma \in \seminv{\program}{\tau}$ and dually $\tau \in \sem{\program}{\sigma}$.
	Using the soundness criteria of the transformers, we get that
	\begin{align*}
		& \aslp{\program}{\pre}(\tau) \\
		=\ & \begin{cases}
				\bigvee\limits_{\sigma \in \seminv{\program}{\tau}} \pre(\sigma)  & \text{ if } \exists \sigma \in \seminv{\program}{\tau} \\
				\true & \text{ else}
			\end{cases} \tag{\Cref{def:aslp}}\\
		=\ & \bigvee\limits_{\sigma \in \seminv{\program}{\tau}} \pre(\sigma) \tag{by assumption}\\
		=\ & \asp{\program}{\pre}(\tau) \tag{\Cref{def:asp}}
	\end{align*}
	The second part of the theorem can be proved dually.
\end{proof}

\subsubsection{Proof of \Cref{theo:reachability-properties}}
\label{ssec:proof-reachability-properties}

\begin{proof}
	We prove \Cref{theo:reachability-properties} pointwise, i.e.\ we show that for $\program \in \ngcl$ and $\tau \in \Sigma$, iff $\tau$ is reachable, we have that
	\[
		\dsp{\program}{\true}(\tau) = \true
	\]
	and dually
	\[
		\dslp{\program}{\false}(\tau) = \false
	\]
	and dually
	\[
		\asp{\program}{\true}(\tau) = \true
	\]
	and dually
	\[
		\aslp{\program}{\false}(\tau) = \false.
	\]	

	We prove this using the soundness criteria of the transformers.
	We start by showing that if $\tau$ is reachable, all of the four properties above hold.
	First, assume that $\tau$ is reachable.
	This means that there exists a $\sigma \in \Sigma$ such that $\sigma \in \seminv{\program}{\tau}$.
	We get
	\begin{align*}
		& \dsp{\program}{\true}(\tau) \\
		=\ &	\begin{cases}
			\bigwedge\limits_{\sigma \in \seminv{\program}{\tau}} \true(\sigma)  & \text{ if } \exists \sigma \in \seminv{\program}{\tau} \\
			\false & \text{ else}
		\end{cases} \tag{\Cref{def:dsp}}\\
		=\ & \bigwedge\limits_{\sigma \in \seminv{\program}{\tau}} \true \\
		=\ & \true \tag{$\tau$ reachable}
	\end{align*}
	By \Cref{theo:order-trans}, we have that $\dsp{\program}{\true}(\tau) = \true$ implies that  $\asp{\program}{\true}(\tau) = \true$.

	For the liberal transformers, we get that
	\begin{align*}
		& \aslp{\program}{\false}(\tau) \\
		=\ & \begin{cases}
			\bigvee\limits_{\sigma \in \seminv{\program}{\tau}} \false(\sigma)  & \text{ if } \exists \sigma \in \seminv{\program}{\tau} \\
			\true & \text{ else}
		\end{cases} \tag{\Cref{def:aslp}}\\
		=\ & \begin{cases}
			\false  & \text{ if } \exists \sigma \in \seminv{\program}{\tau} \\
			\true & \text{ else}
		\end{cases}  \\
		=\ & \false \tag{$\tau$ reachable}
	\end{align*}
	Again using \Cref{theo:order-trans}, this implies $\dslp{\program}{\false}(\tau) = \false$.
	This concludes the first direction of the proof.

	For the second direction, we assume that $\asp{\program}{\true}(\tau) = \true$.
	We then have that
	\begin{align*}
		& \asp{\program}{\true}(\tau) \\
		=\ & \bigvee\limits_{\sigma \in \seminv{\program}{\tau}} \true \tag{\Cref{def:asp}} \\
		=\ & \true
	\end{align*}
	If $\seminv{\program}{\tau}$ were empty, which means that $\tau$ is unreachable, the above conjunction would evaluate to false.
	We can therefore conclude that $\tau$ must be reachable.

	Assume that $\dsp{\program}{\true}(\tau) = \true$.
	Then by \Cref{theo:order-trans}, this implies $\asp{\program}{\true}(\tau) = \true$, which by the above reasoning implies that $\tau$ is reachable.

	Now assume that $\dslp{\program}{\false}(\tau) = \false$.
	We then have that
	\begin{align*}
		& \dslp{\program}{\false}(\tau) \\
		=\ & \bigwedge\limits_{\sigma \in \seminv{\program}{\tau}} \false \tag{\Cref{def:dslp}} \\
		=\ & \false
	\end{align*}
	If $\seminv{\program}{\tau}$ were empty, which means that $\tau$ is unreachable, the above disjunction would evaluate to true.
	This is a contradiction and we can conclude that $\tau$ must be reachable.

	Lastly, assume that $\aslp{\program}{\false}(\tau) = \false$.
	Then by \Cref{theo:order-trans}, this implies $\dslp{\program}{\false}(\tau) = \false$, which by the above reasoning implies that $\tau$ is reachable.

	This concludes the proof of \Cref{theo:reachability-properties}.
\end{proof}

\subsubsection{Proof of \Cref{theo:determinism}}
\label{ssec:proof-determinism}

\begin{proof}
	Assume a deterministic program $\program \in \ngcl$.
	We know that for all $\sigma \in \Sigma$, the set of reachable states $\sem{\program}{\sigma}$ is either a singleton or the empty set.
	Therefore, we have that
	\begin{align*}
		&\awp{\program}{\post}(\sigma) \\
		=\ & \bigvee\limits_{\tau \in \sem{\program}{\sigma}} \post(\tau) \tag{\Cref{def:wp}} \\
		=\ & \begin{cases}
			\post(\tau)  & \text{ for the unique } \tau \in \sem{\program}{\sigma} \text{ if it exists} \\
			\false & \text{ if } \sem{\program}{\sigma} = \emptyset
		\end{cases} \\
		=\ & \begin{cases}
			\post(\tau)  & \text{ for the unique } \tau \in \sem{\program}{\sigma} \text{ if it exists} \\
			\false & \text{ if there exists a diverging path starting in } \sigma
		\end{cases} \\
		=\ & \begin{cases}
			\bigwedge\limits_{\tau \in \sem{\program}{\sigma}} \post(\tau) & \text{ if } \sem{\program}{\sigma} \neq \emptyset \\
			\false & \text{ if computation diverges from } \sigma \\
		\end{cases} \\
		=\ & \dwp{\program}{\post}(\sigma) \tag{\Cref{def:wp}}
	\end{align*}

	The proof for the liberal case is dual.
\end{proof}

\subsubsection{Proof of \Cref{theo:reversibility}}
\label{ssec:proof-reversibility}

\begin{proof}
	Assume a reversible program $\program \in \ngcl$.
	We know that for all $\tau \in \Sigma$, the set of states that can reach $\tau$, $\seminv{\program}{\tau}$, is either a singleton or the empty set.
	Therefore, we have that
	\begin{align*}
		&\asp{\program}{\pre}(\tau) \\
		=\ & \bigvee\limits_{\sigma \in \seminv{\program}{\tau}} \pre(\sigma) \tag{\Cref{def:asp}}\\
		=\ & \begin{cases}
			\pre(\sigma)  & \text{ for the unique } \sigma \in \seminv{\program}{\tau} \text{ if it exists} \\
			\false & \text{ if } \seminv{\program}{\tau} = \emptyset
		\end{cases} \\
		=\ & \begin{cases}
			\bigwedge\limits_{\sigma \in \seminv{\program}{\tau}} \pre(\sigma)  & \text{ if } \exists \sigma \in \seminv{\program}{\tau} \\
			\false & \text{ else}
		\end{cases} \\
		=\ & \dsp{\program}{\pre}(\tau) \tag{\Cref{def:dsp}}
	\end{align*}

	The proof for the liberal case is dual.
\end{proof}

\subsection{Proofs from \Cref{ssec:assumption-2}}

\subsubsection{Proof of \Cref{theo:branching}}
\label{ssec:proof-branching}

\begin{proof}
    Let $\program \in \ngcl$ be a program and $\post$ a postcondition.
	Fix a state $\sigma \in \Sigma$ such that computation started in $\sigma$ either always diverges or always terminates.
    Assume we are in the first case, so computation always diverges and we have $\sem{\program}{\sigma} = \emptyset$.

    \begin{align*}
		& \awlp{\program}{\post}(\sigma) \\
        \eeq & \begin{cases}
			\true & \text{ if there exists a diverging path starting in } \sigma \\
			\bigvee\limits_{\tau \in \sem{\program}{\sigma}} \post(\tau) & \text{ else}
		\end{cases} \\
        \eeq & \true \tag{by assumption} \\
        \eeq & \false \quad \lor \quad  \true  \\
        \eeq & \bigvee\limits_{\tau \in \sem{\program}{\sigma}} \post(\tau)  \quad \lor \quad \bigwedge\limits_{\tau \in \sem{\program}{\sigma}} \post(\tau) \tag{since $\sem{\program}{\sigma} = \emptyset$}\\
        \eeq & \awp{\program}{\post}(\sigma) \quad \lor \quad \dwlp{\program}{\post}(\sigma)
	\end{align*}

    Now assume we are in the second case, so computation always terminates.
    Note that this implies that $\sem{\program}{\sigma}$ is not empty.
    We get:
    \begin{align*}
		& \awlp{\program}{\post}(\sigma) \\
        \eeq & \begin{cases}
			\true & \text{ if there exists a diverging path starting in } \sigma \\
			\bigvee\limits_{\tau \in \sem{\program}{\sigma}} \post(\tau) & \text{ else}
		\end{cases} \\
        \eeq & \bigvee\limits_{\tau \in \sem{\program}{\sigma}} \post(\tau) \tag{by assumption} \\
        \eeq & \bigvee\limits_{\tau \in \sem{\program}{\sigma}} \post(\tau)  \quad \lor \quad \bigwedge\limits_{\tau \in \sem{\program}{\sigma}} \post(\tau) \tag{since $\sem{\program}{\sigma} \neq \emptyset$}\\
        \eeq & \awp{\program}{\post}(\sigma) \quad \lor \quad \dwlp{\program}{\post}(\sigma)
	\end{align*}

    We can conclude that
	\[
		\awlp{\program}{\post}(\sigma) = \awp{\program}{\post}(\sigma) \cup \dwlp{\program}{\post}(\sigma).
	\]
	The proof for showing
	\[
		\dwp{\program}{\post}(\sigma) = \awp{\program}{\post}(\sigma) \cap \dwlp{\program}{\post}(\sigma)
	\]
    is dual.
\end{proof}
% \input{sections/app-basteln}

% qsp paper
% \input{sections/app-nGCL}
% \input{sections/app-weakest-pre}
% \input{sections/app-strongest-post}
% \input{sections/app-relationship}
% \input{sections/app-examples}

\end{document}